\documentclass[a4paper,11pt]{article} 
\usepackage[english]{babel}
\usepackage[a4paper]{geometry}
\usepackage{amsmath}
\usepackage{amsthm}
\usepackage{amssymb}
\usepackage{color}
\usepackage{hyperref} 
\usepackage{enumerate}
\usepackage{bbm}
\usepackage{mathtools}

\def\bR{\mathbb{R}}

\def\bN{\mathbb{N}}

\def\bZ{\mathbb{Z}}
\def\cN{\mathcal{N}}
\def\cE{\mathcal{E}}
\def\cK{\mathcal{K}}
\def\cF{\mathcal{F}}
\def\cH{\mathcal{H}}
\def\eps{\varepsilon}

\def\cV{\mathcal{V}}
\def\cC{\mathcal{C}}
\def\cL{\mathcal{L}}
\def\cG{\mathcal{G}}

\def\cQ{\mathcal{Q}}
\def\cT{\mathcal{T}}
\def\cW{\mathcal{W}}

\def\wt{\widetilde}

\def\be{\begin{equation*}}
\def\ee{\end{equation*}}

\def\aa{\mathfrak{a}}
\newcommand{\norm}[1]{\lVert #1 \rVert}
\def\la{\langle}
\def\ra{\rangle}

\def\hc{\text{h.c.}\ }

\newcommand{\abs}[1]{\lvert #1 \rvert}

\DeclareFontFamily{U}{mathx}{}
\DeclareFontShape{U}{mathx}{m}{n}{<-> mathx10}{}
\DeclareSymbolFont{mathx}{U}{mathx}{m}{n}
\DeclareMathAccent{\widehat}{0}{mathx}{"70}
\DeclareMathAccent{\widecheck}{0}{mathx}{"71}

\newtheorem{theorem}{Theorem}   
\newtheorem{cor}[theorem]{Corollary}

\newtheorem{lemma}[theorem]{Lemma}

\title{Third Order Upper Bound for the Ground State Energy of the Dilute Bose Gas}

\author{
Morris Brooks\thanks{Institute for Mathematics, University of Zurich, Winterthurerstrasse 190, 8057 Zurich, Switzerland}
\and Jakob Oldenburg\footnotemark[1]
\and Diane Saint Aubin\footnotemark[1]
\and Benjamin Schlein\footnotemark[1]
}

\begin{document}

\maketitle

\begin{abstract}
    We consider a dilute Bose gas in the thermodynamic limit. We prove an upper bound for the ground state energy per unit volume, capturing the expected third order term, as predicted by Wu, Hugenholtz-Pines and Sawada.
\end{abstract}

\section{Introduction}

We consider a dilute Bose gas, consisting of $N$ particles moving in a large, thermodynamic box $\Lambda_t$ with Dirichlet boundary conditions and interacting through a two-body potential $V$. The Hamilton operator of the system is given by 
\begin{equation}\label{eq:HL} 
    H_L = \sum_{i=1}^N -\Delta_i + \sum_{1\leq i < j \leq N} V (x_i-x_j)
\end{equation}
and acts on $L^2_s (\Lambda_t^N)$, the subspace of $L^2 (\Lambda_t^N)$ consisting of all functions that are symmetric with respect to permutations of the $N$ particles. 

We are going to assume that $V \in L^2 (\bR^3)$ is radial, stable (see \cite[Section 3.5]{Rue}), with compact support, with no two-body bound states, and with positive scattering length $\aa > 0$ as defined in \cite[Appendix C]{LSSY}. In particular, we can consider any non-negative, radial and compactly supported $V\in L^2(\mathbb{R}^3)$.
Recall that the scattering length of $V$ can be recovered from the solution of the zero-energy scattering equation 
\begin{equation}\label{eq:0en} \left[ -\Delta + \frac{1}{2} V \right] f = 0 \end{equation} 
with the boundary condition $f (x) \to 1$ as $|x| \to \infty$, since $f (x) = 1 - \frak{a} / |x|$,  outside the support of $V$. Equivalently, 
\[ 8\pi \frak{a} = \int V(x) f(x) dx \, . \]
Notice that $f \in L^\infty (\bR^3)$, as shown in \cite[Theorem 11.7]{LL}. 

We denote by $E (N,\Lambda_t)$ the ground state energy of the gas, ie. the lowest eigenvalue of the operator (\ref{eq:HL}). We consider the thermodynamic limit $N, |\Lambda_t| \to \infty$, at fixed density $\rho = N/ |\Lambda_t|$. We are interested in the energy per unit volume  
\[ e(\rho) = \lim_{N, |\Lambda_t| \to \infty , \rho = N / |\Lambda_t|} \frac{E (N, \Lambda_t)}{| \Lambda_t |}  \]
in the dilute limit $\rho \frak{a}^3 \to 0$. The following theorem is the main result of our paper.
\begin{theorem}\label{theorem:main}
Let $V\in L^2(\mathbb{R}^3)$ be radial, stable and with compact support. Assume that $V$ does not admit two-body bound states and let $\frak{a} > 0$ be its scattering length. Then, there is a constant $C > 0$ such that 
\begin{equation} \label{eq:main} 
        e(\rho) \leq 4\pi \mathfrak{a}\rho^2
        \left(
        1+\frac{128}{15\sqrt{\pi}}\sqrt{\rho\mathfrak{a}^3} 
        + 8\left(\frac{4\pi}{3}-\sqrt{3}\right) \rho\mathfrak{a}^3 \log(\rho\mathfrak{a}^3) 
        + C \rho \frak{a}^3 
        \right)
    \end{equation}
for all $\rho \frak{a}^3 > 0$ small enough.  
\end{theorem}

The bound (\ref{eq:main}) is expected to capture the correct ground state energy per unit volume, up to corrections of order $\aa \rho^2 (\rho \aa^3)$, in the dilute limit $\rho \aa^3 \to 0$, that are expected to no longer depend universally on $\aa$. The expression on the r.h.s. of (\ref{eq:main}) has a long history in the physics literature. The first two terms in the asymptotic expansion have been predicted in 1957, by Lee-Huang-Yang \cite{LHY}, based on  previous work of Bogoliubov \cite{Bo}. Two years later, in 1959, Wu \cite{Wu} and then also Hugenholtz-Pines \cite{HP} and Sawada \cite{Sa} predicted the third logarithmic correction. At the level of rigorous mathematical derivation, the leading order term has been established by Dyson \cite{Dy} as an upper bound and, 40 years later, by Lieb-Yngvason \cite{LY} as a lower bound. For potentials with negative part, the corresponding result was obtained in \cite{Yin}. An upper bounds including the Lee-Huang-Yang term has been proven in \cite{YY}, for sufficiently regular interactions (previously, an upper bound reaching the Lee-Huang-Yang order, but only recovering the correct constant in the limit of weak coupling, was shown in \cite{ESY}). More recently, a simpler upper bound has been derived in \cite{BCS}, for a larger class of interactions. The case of hard-sphere potential, however, remains open (for hard-spheres, the currently best available upper bound on the ground state energy per unit volume has been shown in \cite{BCGOPS}, with an error of the Lee-Huang-Yang order). A lower bound resolving the Lee-Huang-Yang contribution has been first shown in \cite{FS1} and, for more general interactions including hard spheres, in \cite{FS2}. Recently, a new proof of the lower bound, applying also to the free energy at low positive temperature, has been obtained in \cite{HHNST} (a matching low temperature upper bound has been derived in \cite{HHST}). The result of \cite{HHNST} has been extended to potentials including hard spheres in \cite{FJGMOT}.

Theorem \ref{theorem:main} establishes an upper bound on the ground state energy per unit volume of a dilute Bose gas, interacting through $V \in L^2 (\bR^3)$, consistent with Wu's prediction. A complete proof of the validity of Wu's formula in the thermodynamic limit requires a matching lower bound, which remains a challenging open problem. Remark, however, that a rigorous derivation of Wu's correction is available in the Gross-Pitaevskii regime, where $N$ particles are confined in the unit torus and interact through a potential with scattering length of the order $1/N$. In this limit, the ground state energy and the low-energy excitation spectrum have been resolved to the Lee-Huang-Yang order in \cite{BBCS1,BBCS2}, through a rigorous version of Bogoliubov theory (leading order estimates and a proof of Bose-Einstein condensation have been previously obtained in \cite{LSY,LS,LS2}) . Recently, extending the approach of  \cite{BBCS1,BBCS2}, Wu's logarithmic correction to the ground state energy has been resolved in \cite{COSS}; we will use ideas from \cite{COSS} also in the current paper, in particular in Section \ref{subsec:kerA} to define the kernel of the cubic transformation. Related results in the Gross-Pitaevskii regime (and beyond) have been derived in \cite{ABS,BCOPS,BDS,BBCO,BCaS,BSS1,BSS2,B,CD,HST,NNRT,NT}. 
   
The proof of Theorem \ref{theorem:main} is based on a well-known localization technique; given a grand canonical trial state $\psi_{L}$ on a periodic box of size $L = \rho^{-\gamma}$, with approximately $\rho L^3$ particles and with approximately the correct energy per unit volume, we can construct a grand canonical trial state on the thermodynamic box $\Lambda_t$ with approximately the correct energy, first extending $\psi_{L}$ to a trial state $\psi_{L}^\text{dir}$ on a Dirichlet box of size slightly larger than $L$ and then combining several copies of $\psi_{L}^\text{dir}$ to cover the whole volume of $\Lambda_t$ (to avoid interactions among particles in different boxes, they must be separated by small corridors, broader than the support of $V$). At the end, (\ref{eq:main}) follows by the equivalence of ensembles. 

The main part of the paper is devoted to the construction of the grand canonical trial state on the periodic box $\Lambda_L$, of size $L = \rho^{-\gamma}$ for appropriate $\gamma > 0$ (to control the localization error produced by imposing Dirichlet boundary conditions, we will need $\gamma \geq 3/2$). In fact, we find it more convenient to rescale lengths, and work on the unit torus $\Lambda \equiv \Lambda_1 \simeq [-1/2 ;1/2]^3$. Denoting by $N = \rho L^3 = \rho^{1-3\gamma}$ the expected number of particles in $\Lambda_L$, we find $\rho = N^{-1/(3\gamma-1)}$ and $L = N^{\gamma / (3\gamma-1)} = N^{1-\kappa}$, where we introduced the parameter $\kappa = (2\gamma-1)/(3\gamma-1)$. For $\kappa \in [0;2/3)$, we consider the Hamilton operator 
\begin{equation}\label{eq:HN-fock}
\cH_N  = \sum_{p\in\Lambda^*} p^2 a_p^* a_p
    + \frac{N^\kappa}{2N} \sum_{p,q,r\in\Lambda^*} \hat{V}(r/N^{1-\kappa}) a_{p+r}^* a_{q-r}^* a_q a_p 
\end{equation}
acting on the bosonic Fock space $\cF (\Lambda) = \oplus_{n\geq 0}L^2_s(\Lambda^n)$. Here, $\Lambda^*= 2\pi\mathbb{Z}^3$ is the dual lattice and, for $p \in \Lambda^*$, $a_p^*, a_p$ are the usual creation and annihilation operators, satisfying canonical commutation relations. Moreover, we slightly abuse notation and denote $V$ and its periodic extension in the same way.
\begin{theorem}\label{theorem:Fockspaceresult}
    Under the same assumptions on $V$ as in Theorem \ref{theorem:main}, we have that for any $\kappa\in [1/2,2/3)$ and $N$ large enough there exists $\Psi_N\in\cF (\Lambda)$, $\norm{\Psi_N}=1$ such that
    \begin{equation}\label{eq:psiN-N}
        \la \Psi_N, \cN \Psi_N \ra \geq N, \quad \la \Psi_N, \cN^2 \Psi_N \ra \leq CN^2
    \end{equation}
    and
    \begin{equation}\label{eq:mainN} 
    \begin{split}
        \la \Psi_N, \cH_N \Psi_N \ra
        \leq \; &4\pi \mathfrak{a} N^{1+\kappa}
        +\frac{512\sqrt{\pi}}{15}\mathfrak{a}^{5/2}N^{5\kappa/2}\\
        &+   32\pi \left(\frac{4\pi}{3}-\sqrt{3}\right) \mathfrak{a}^4 N^{4\kappa-1} \log(N^{3\kappa-2}) 
        + C N^{4\kappa-1}.
    \end{split}
    \end{equation}
\end{theorem}
\textit{Remark.} The restriction to $\kappa \geq 1/2$ makes sure that finite size errors of order $N^{2\kappa}$ are negligible. 

With Theorem \ref{theorem:Fockspaceresult}, we can prove our main result. 
\begin{proof}[Proof of Theorem \ref{theorem:main}]
To translate the upper bound (\ref{eq:mainN}) into a bound for the ground state energy per unit volume in the thermodynamic limit we use \cite[Prop. 1.2]{BCS}, which states that, if $R < \ell < L$, with $R$ the radius of the support of the potential $V$, and if $\psi_L \in \cF (\Lambda_L)$ is normalized and such that 
\begin{equation}\label{eq:NNL} 
\langle \psi_L, \cN \psi_L \rangle \geq \rho (1+c \rho) (L+2\ell+R)^3, \qquad \langle \psi_L, \cN^2 \psi_L \rangle \leq C \rho^2 (L+2\ell +R)^6 \end{equation} 
then we have 
\begin{equation}\label{eq:prop12} e(\rho) \leq \frac{\langle \psi_L, \cH \psi_L \rangle}{L^3} + \frac{C}{L^4 \ell} \langle \psi_L, \cN \psi_L \rangle \end{equation} 
with the Hamilton operator 
\[ \cH = \sum_{p \in \Lambda_L^*} p^2 a_p^* a_p + \frac{1}{2} \sum_{p,q,r \in \Lambda^*_L} \hat{V} (r) a_{p+r}^* a_{q-r}^* a_q a_p \]
acting on $\cF (\Lambda_L)$ (where $\Lambda^*_L = 2\pi \bZ^3 / L$). 

Notice that the state $\psi_L \in \cF (\Lambda_L)$, whose energy appears on the r.h.s. of (\ref{eq:prop12}) and provides an upper bound to $e (\rho)$, must have, effectively, a slightly larger density than $\rho$ (to compensate the fact that, to construct the thermodynamic trial state, we first need to increase a bit the size of the box $\Lambda_L$ to impose Dirichlet boundary conditions and to make sure that particles in different boxes do not interact). 

Given a density $\rho > 0$, we thus choose $\wt{\rho} > \rho$ (to be specified later on) and we apply Theorem \ref{theorem:Fockspaceresult}. After rescaling lengths we find, for any $\gamma > 1$ (reflecting the condition $\kappa = (2\gamma -1)/(3\gamma-1) \in (1/2 ; 2/3)$), a normalized $\psi_L \in \cF (\Lambda_L)$ on the periodic box of size $L = \wt{\rho}^{-\gamma}$ (so that $N = \wt{\rho}^{1-3\gamma}$), with 
\[ \langle \psi_L, \cN \psi_L \rangle \geq \wt{\rho} L^3, \qquad  \langle \psi_L, \cN^2 \psi_L \rangle \leq C \wt{\rho}^2 L^6 \]
and (since $\cH_N$ is unitarily equivalent to $L^2 \cH$) with 
\[ \frac{\langle \psi_L , \cH \psi_L \rangle}{L^3} \leq 4 \pi \frak{a} \wt{\rho}^2 \left( 1 + \frac{128}{15\sqrt{\pi}} \sqrt{\wt{\rho} \frak{a}^3} + 8 \left( \frac{4\pi}{3} - \sqrt{3} \right) \wt{\rho} \frak{a}^3 \log (\wt{\rho} \frak{a}^3) + C \wt{\rho} \frak{a}^3 \right). \]

To make sure that (\ref{eq:NNL}) holds true, we set $\ell = L^\alpha$, for $0 < \alpha < 1$, and we apply the implicit function theorem to find $\wt{\rho} > \rho$ such that 
\[ \wt{\rho} = \rho (1 + c \rho) (1 + 2 \wt{\rho}^{\gamma (1-\alpha)} + R \wt{\rho}^{\gamma})^3. \]

From (\ref{eq:prop12}), we conclude therefore that 
\[ e (\rho) \leq 4 \pi \frak{a} \wt{\rho}^2 \left( 1 + \frac{128}{15\sqrt{\pi}} \sqrt{\wt{\rho} \frak{a}^3} + 8 \left( \frac{4\pi}{3} - \sqrt{3} \right) \wt{\rho} \frak{a}^3 \log (\wt{\rho} \frak{a}^3) + C \wt{\rho} \frak{a}^3 \right) + C \wt{\rho}^{1+\gamma (1+\alpha)}. \]
Since $\wt{\rho}^2 \leq \rho^2 (1 + C \rho + C \rho^{\gamma (1-\alpha)})$, we obtain  
\[ \begin{split}  e (\rho) \leq \; &4 \pi \frak{a} \rho^2 \left( 1 + \frac{128}{15\sqrt{\pi}} \sqrt{\rho \frak{a}^3} + 8 \left( \frac{4\pi}{3} - \sqrt{3} \right) \rho \frak{a}^3 \log (\rho \frak{a}^3) \right) \\ &+ C (\rho^3 + \rho^{2+ \gamma (1-\alpha)} + \rho^{1+ \gamma (1+\alpha)}). \end{split} \]
Choosing $\alpha = 1/(2\gamma)$, we find the error $\rho^{3/2+\gamma}$. Choosing $\gamma \geq  3/2$ (which corresponds to $\kappa \geq 4/7$, using the language of Theorem \ref{theorem:Fockspaceresult}), we arrive at the desired bound.
\end{proof} 

The rest of the paper is devoted to the proof of Theorem \ref{theorem:Fockspaceresult}. Let us briefly explain, on a heuristic level, the main ideas involved in the proof of (\ref{eq:main}). We start from Bogoliubov theory \cite{Bo}. First, we factor out the Bose-Einstein condensate, conjugating the Hamiltonian (\ref{eq:HN-fock}) with a Weyl operator generating a coherent state with, in average, $N_0$ particles with momentum $p=0$. 
%
In the second step, we conjugate the resulting excitation Hamiltonian 
with a Bogoliubov transformation $e^{B^* - B}$, with 
\begin{equation}\label{eq:B0} B = \frac{1}{2} \sum_{|p| \geq N^{\kappa/2-\beta/4}} \mu_p a_p a_{-p} \end{equation}
and with coefficients 
\[ \mu_p =  -\frac{1}{4} \log \Big( 1 + \frac{2N^\kappa \widehat{Vf} (p/N^{1-\kappa})}{p^2} \Big) \]
satisfying 
\begin{equation}\label{eq:mu2} \tanh (2 \mu_p) = - \widehat{Vf} (p/N^{1-\kappa}) / (p^2 + \widehat{Vf} (p/N^{1-\kappa})). \end{equation} 
This leads us to a new, unitarily equivalent, Hamiltonian 
\begin{equation}\label{eq:wtHN} \wt{\cH}_N \simeq C_N + \cK + \cC_N + \cC_N^* + \cV_N \end{equation} 
where the constant $C_N$ gives the correct ground state energy to leading order (and already contains some terms contributing to the Lee-Huang-Yang correction), $\cK$ is the kinetic energy operator, 
\[ \cC^*_N = \frac{1}{\sqrt{N}} \sum_{p,r \in \Lambda_+^*} \hat{V} (r/N^{1-\kappa}) \sigma_p a_{p+r}^* a_{-r}^* a_p^* \]
with $\sigma_p = \sinh \mu_p$ and   
\[ \cV_N = \frac{1}{2N} \sum_{r,p,q \in \Lambda^*} \hat{V} (r/N^{1-\kappa}) a_{p+r}^* a_q^* a_{q+r} a_p \, . \]
The two Bogoliubov transformations that have been used, for example, in \cite{BBCS1,BBCS2}, are combined here into (\ref{eq:B0}). For large momenta $|p| \gg N^{\kappa/2}$, $\mu_p \simeq - N^\kappa \widehat{Vf} (p/N^{1-\kappa})/ 2p^2$ approximates the solution of the zero-energy scattering equation (\ref{eq:0en}) and the action of (\ref{eq:B0}) helps renormalizing the interaction potential. For $|p| \lesssim N^{\kappa/2}$, on the other hand, the Bogoliubov transformation diagonalizes the quadratic terms in the excitation Hamiltonian (the identity (\ref{eq:mu2}) makes sure that $e^{B^* - B}$ diagonalizes exactly a quadratic Hamiltonian with potential $Vf$). In (\ref{eq:B0}), we restrict the sum to momenta $|p| \geq N^{\kappa/2-\beta/4}$, with the exponent $\beta = 2-3\kappa$; contributions arising from $|p| < N^{\kappa/2-\beta/4}$ do not change the energy, up to order $N^{4\kappa-1}$ (to resolve the Lee-Huang-Yang energy, it is enough to consider $|p| \gtrsim N^{\kappa/2}$).   
 

Wu \cite{Wu} proposed to use second order perturbation theory to obtain a better approximation for the ground state energy, essentially treating $\cK$ as the unperturbed Hamiltonian and $\cC_N + \cC_N^* +\cV_N$ as the perturbations (first order perturbation theory gives a vanishing contribution, because the ground state of $\cK$ is the vacuum and the vacuum expectation of $\cC_N, \cC_N^*, \cV_N$ is zero). This remark suggests that a good trial state for the renormalized Hamiltonian $\widetilde{\cH}_N$ could have the form 
\begin{equation}\label{eq:xi0} \begin{split} \xi 
&= \Omega - \frac{1}{\sqrt{N} \cK}  \sum_{|r|\geq N^{\kappa/2}, |v| \geq N^{\kappa/2-\beta/4}} \widehat{Vf} (r/N^{1-\kappa}) \sigma_v a_{r+v}^* a_{-r}^* a_{-v}^* \Omega 
\\ &= \Omega - \frac{1}{\sqrt{N}} \sum_{|r|\geq N^{\kappa/2}, |v|\geq N^{\kappa/2-\beta/4}} \frac{\widehat{Vf} (r/N^{1-\kappa}) \sigma_v}{(r+v)^2 + r^2 +v^2} a_{r+v}^* a_{-r}^* a_{-v}^* \Omega =: \Omega + A^* \Omega \end{split} \end{equation}  
where we used the renormalized potential to make sure that the energy of $\xi$ only depends on the scattering length and where we restricted the sum to $|r| \geq N^{\kappa/2}$ and $|v| \geq N^{\kappa/2-\beta/4}$, since contributions from smaller momenta do not affect the energy, to order $N^{4\kappa-1}$. It turns out that, with this choice of $\xi$,  
\begin{equation}\label{eq:en-pert}  \la \xi, (\cK + \cC_N + \cC_N^* + \cV_N) \xi \ra = \la \Omega, (A \cC_N^* + \cC_N A^* ) \Omega \ra + \la \Omega, A  (\cK + \cV_N )  A^* \Omega \ra  \end{equation} 
leads, together with the constant term $C_N$ in (\ref{eq:wtHN}), exactly to the correct expression for the ground state energy, including the third order term in (\ref{eq:main}). The problem with this approach is the fact that $\xi$ is not normalized. Since 
\[ \| \xi \|^2 = 1 + \| A^* \Omega \|^2  \simeq C N^{3\kappa-1} \]
the trial state $\xi / \| \xi \|$ only provides a good upper bound for the ground state energy, if $\kappa =0$. To avoid problems with the normalization, we will consider a trial state defined through a unitary map. We define, formally, $\xi = e^{A^* - A} \Omega$, with $A$ as in (\ref{eq:xi0}). Since  
\[ \la \Omega, [  \cC_N + \cC_N^* , A^* - A ] \Omega \ra = \la \Omega, ( \cC_N  A^* + A \cC_N^* ) \Omega \ra \]
and 
\[ \la \Omega, [ [ \cK + \cV_N , A^* - A] , A^* - A ] \Omega \ra = 2 \la \Omega, A (\cK + \cV) A^* \Omega \ra \]
the trial state $\xi$ does indeed generate the contributions (\ref{eq:en-pert}), leading to Wu's correction. To obtain a rigorous upper bound, though, we would have to show that $e^{A^* - A}$ is well-defined as a unitary operator and we would have to show convergence of its Duhamel expansion. 

On the one hand, we observe that the vacuum expectation of multiple commutators like $[ A , [ A^* , [ A, \dots [ A^*, O] \dots ]$ gets smaller, with increasing number of commutators. In fact, for every two additional commutators, we have four more momenta to sum up. When considering the vacuum expectation, however, we also have three more delta functions. Effectively, every new pair of commutators with $A,A^*$ produces therefore only one additional momentum sum, typically of the order $\| \sigma \|_2^2 \simeq N^{3\kappa/2}$. Since each factor $A,A^*$ carries a small constant $N^{-1/2}$, every two new commutators generate an additional small factor $N^{3\kappa/2 - 1} \ll 1$. 

On the other hand, to 
control error terms arising in the Duhamel expansion, it is very useful to introduce a cutoff on the local number of particles. For $x \in \Lambda$ and $\ell > 0$, we consider the operator  
\[ \cN_\ell (x) = \int_{|x-y| \leq \ell}  dy \, a_y^* a_y \]
measuring the number of particles in a ball of radius $\ell$ around $x$. Since we used an infrared cutoff in momentum space (in fact, we will work with smooth cutoffs), the kernel $\check{\sigma} (x)$ exhibits fast decay, for $|x| \gg \ell_\sigma := N^{-\kappa/2+\beta/4}$. Similarly, the kernel $\check{A} (y,x)$ associated with the cubic operator $A$ defined in (\ref{eq:xi0}), is very small for $|x| \gg  \ell_\sigma$ or for $|y| \gg \ell_\eta := N^{-\kappa/2}$. In fact, $\check{\sigma} (x)$ and $\check{A} (y,x)$ also decay for $|x| < \ell_\sigma$, with rate proportional to $|x|^{-5/2}$. This implies that their $L^2$-norm is essentially supported in the region $|x| \leq \ell_B = N^{-\kappa/2 + \beta/12} \ll \ell_\sigma$, in the sense that, if we introduce a lattice $\Lambda_B = \ell_B \bZ^3$ and, for every $u \in \Lambda_B \cap \Lambda$, we define $\sigma_u = \sigma \mathbbm{1}_{B_u}$, with $B_u$ the box with side length $\ell_B$ centered at $u$, then 
\begin{equation}\label{eq:sigmau-0}  \sum_{u \in \Lambda_B} \| \check{\sigma}_u \|_2 \leq C N^{3\kappa/4 + \epsilon} \end{equation} 
for any $\epsilon > 0$. Up to the additional factor $N^\epsilon$, this is comparable with $\| \check{\sigma} \|_2$. To take advantage from the decay of the kernels $\check{\sigma}$, $\check{A}$, we will therefore modify the trial state, defining $\xi = e^{A^* \Theta - \Theta A} \Omega$, where we introduced an operator $\Theta$, making sure that 
\begin{equation}\label{eq:Theta-cut} \Theta \cN_{\ell_B} (x) \Theta \leq N^\epsilon \end{equation} 
for any $\epsilon > 0$ and every $x \in \Lambda$, cutting off the local number of particles (more precisely, the local number of cubic excitations) in balls of radius $\ell_B = N^{-\kappa/2 + \beta/12}$. Since the expected (by Duhamel expansion) average number of particles generated by the cubic transformation in a ball of radius $\ell_B$ is of the order \[ \ell^3_B \| A^* \Omega \|^2 \simeq \ell^3_B N^{3\kappa-1} \simeq N^{-\beta/4} \ll 1, \] the cutoff does not substantially change the energy of the trial state (but it is nevertheless important to control error terms). 

To explain how we use the cutoff $\Theta$, let us try to estimate the number of particles generated by the unitary operator $e^{\Theta \wt{A}^* - \wt{A} \Theta}$, where we replaced $A$ by its simplified version 
\[ \wt{A} = \frac{1}{\sqrt{N}} \sum_{|p| \geq N^{\kappa/2}, |q| \geq N^{\kappa/2 - \beta/4}} \eta_p \sigma_q a_{p+q} a_{-p} a_{-q} \]
with $\eta_p = - N^\kappa \widehat{Vf} (p/N^{1-\kappa}) / 2p^2$. We have 
\begin{equation}\label{eq:duh0}  \la e^{\Theta \wt{A}^* - \wt{A}\Theta} \Omega, \cN e^{\Theta \wt{A}^*  -  \wt{A}\Theta} \Omega \ra = 2 \text{Re } \int_0^1 ds \la e^{s(\Theta  \wt{A}^*  - \wt{A} \Theta)} \Omega,  [\cN , \wt{A} ] \Theta e^{s(\Theta \wt{A}^* - \wt{A} \Theta)} \Omega \ra. \end{equation} 
Setting $\xi_s = e^{s(\Theta \wt{A}^* -\wt{A} \Theta)} \Omega$ and switching to position space, we find 
\[ \begin{split} 
\big| \la \xi_s, [ \cN , \wt{A} ] \Theta \xi_s \ra \big| &= 3 \int dx dy \, | \check{\eta} (x-y)| |\check{\sigma} (x-z)|  \| a_x a_y a_z \Theta \xi_s \| \| \xi_s \| \\ &\leq 3 \sum_{u \in \Lambda_B} \int dx dy \,  |\check{\eta} (x-y)| |\check{\sigma}_u (x-z)|  \| a_x a_y a_z \Theta \xi_s \| \| \xi_s \|. \end{split} \]
By Cauchy-Schwarz and using that $\check{\sigma}_u (x-z)$ is supported on $|x-z-u| \leq \ell_B$ and $\check{\eta} (x-y)$ is essentially supported on $|x-y| \leq \ell_\eta$, we conclude that 
\begin{equation} \label{eq:groNwtA} \begin{split} 
\big| \la \xi_s, [ \cN , \wt{A} ] \Theta \xi_s \ra \big| &\leq 3\| \check{\eta} \|_2 \sum_{u \in \Lambda_B} \| \check{\eta} \|_2  \| \check{\sigma}_u \|_2 \Big( \int dx \, \| \cN^{1/2}_{\ell_B} (x-u) \cN^{1/2}_{\ell_\eta} (x) a_x \Theta \xi_s \|^2 \Big)^{1/2}.  \end{split} \end{equation} 
With $\| \eta \|_2 \leq C N^{3\kappa/4}$ and using (\ref{eq:sigmau-0}) and (\ref{eq:Theta-cut}) (after passing $\Theta$ through $a_x$), we obtain 
\[ \big| \la \xi_s, [ \cN , \wt{A} ] \Theta \xi_s \ra \big|  \leq C N^{3\kappa/2-1/2+ \epsilon} \| \cN^{1/2} \xi_s \| \| \xi_s \| \leq C N^{3\kappa-1+\epsilon} \| \xi_s \|^2 + C \| \cN^{1/2} \xi_s \|^2. \]
Inserting in (\ref{eq:duh0}) and applying Gronwall's Lemma, we arrive at 
\[ \la e^{\wt{A}^* \Theta - \Theta \wt{A}} \Omega, \cN e^{\wt{A}^* \Theta - \Theta \wt{A}} \Omega \ra \leq C N^{3\kappa-1+\epsilon} \]
which is almost the optimal result (one can remove the additional $N^\epsilon$, expanding (\ref{eq:duh0}) to the next order; this will be shown for the full cubic operator $A$ in Section \ref{subsec:cubic1}). Deriving such an estimate bounding the r.h.s. of (\ref{eq:groNwtA}) with the global, rather than the local, number of particles operator would have been impossible. Also, localizing $\sigma$ on boxes of size $\ell_\sigma$, rather than $\ell_B \ll \ell_\sigma$, would have led to a worse, non-optimal, dependence on $N$.  

The plan of the paper is as follows. In the next section, we define our trial state $\psi_N = W_{N_0} e^{B^* - B} e^{A^* \Theta- \Theta A} \Omega$, where $W_{N_0}$ is the Weyl operator, generating the condensate in the zero momentum state, $e^{B^* - B}$ is the Bogoliubov transformation and $e^{A^* \Theta - \Theta A}$ is the cubic transformation. To this end, in Subsections \ref{subsec:bogo}, \ref{subsec:kerA} we define the coefficients of $B$ and $A$ and we study their properties, in momentum and in position space. In Subsection \ref{sec:cutoff}, we introduce the cutoff $\Theta$ for the local number of particles operator. In Section \ref{sec:energy}, we will then compute the energy of our trial state; in subsection \ref{subsec:quadratic} we first compute the action of the Bogoliubov transformation on the Hamilton operator. This reduces the problem to the computation of certain expectations, in the state $\xi = e^{A^* \Theta - \Theta A} \Omega$, which will be discussed in Subsections \ref{subsec:cubic1}, \ref{subsec:cubic2}. This will allow us to complete the proof of Theorem \ref{theorem:Fockspaceresult} in Subsection \ref{subsec:proof}. 

\medskip

\textit{Acknowledgements.} We gratefully acknowledge financial support from the Swiss National Science Foundation through the Grant “Dynamical and energetic properties of Bose-Einstein condensates” and from the European Research Council through the ERC-AdG CLaQS.

\section{The trial state}

This section is devoted to the construction of the trial state $\psi_N \in \cF (\Lambda )$, which will be later shown to satisfy (\ref{eq:psiN-N}) and (\ref{eq:mainN}). 

We will consider a trial state of the form 
\begin{equation}\label{eq:psi0} \psi_N = W_{N_0} e^{B-B^*} e^{A^* \Theta - \Theta A} \Omega. \end{equation} 
Here, $\Omega \in \cF (\Lambda )$ is the vacuum vector. Moreover,   
\begin{equation}\label{eq:weyl0} W_{N_0} = \exp (\sqrt{N_0} a_0^* - \text{h.c.} ) \end{equation} 
is a Weyl operator, satisfying
\begin{equation}\label{eq:weyl1} W_{N_0}^* a_p W_{N_0} = a_p + \delta_{p,0} \sqrt{N_0} \end{equation}
for all $p \in \Lambda^* = 2\pi \bZ^3$ and generating a coherent state with (in average) $N_0$ particles with momentum $p=0$ (the Bose-Einstein condensate).

On the other hand, $e^{B-B^*}$ and $e^{A^* \Theta - \Theta A}$ are unitary maps, with 
\begin{equation}\label{eq:BT0} B = \frac{1}{2} \sum_{p \in \Lambda^* \backslash \{ 0 \}} \mu_p a_p a_{-p} \end{equation} 
quadratic and 
\begin{equation}\label{eq:A-def} A = \sum_{p,q \in \Lambda^*} A_{p,q} \, a_{p+q} a_{-p} a_{-q} \end{equation} 
cubic in creation and annihilation operators. The operator $\Theta$, appearing in the cubic transformation, is a cutoff, controlling the local particle density. In the next subsections, we are going to define the kernels $\mu_p$, $A_{p,q}$ and the cutoff $\Theta$ precisely and we are going to study their most important properties.

\subsection{Coefficients of the Bogoliubov transformation} 
\label{subsec:bogo}

We use the solution $f$ of the zero-energy scattering equation (\ref{eq:0en}) to define the coefficients 
\begin{equation}\label{eq:defetainfty}
\eta_{\infty,p} = \frac{-N^\kappa \widehat{Vf}(p/N^{1-\kappa})}{2p^2}
\end{equation}
and
\begin{equation}\label{eq:defmuinfty}
\mu_{\infty,p} = -\frac{1}{4} \log (1-4\eta_{\infty,p})
\end{equation}
for $p\in \Lambda^*_+ = 2\pi \mathbb{Z}^3\backslash\{0\}$. We also define $\eta_{\infty,0}= \mu_{\infty,0}=0$. Notice that \[ \tanh (2\mu_{\infty,p})  = - \frac{
\widehat{Vf}(p/N^{1-\kappa})}{p^2 + \widehat{Vf}(p/N^{1-\kappa})} \, .\]
Hence, a Bogoliubov transformation with coefficients $\mu_{\infty, p}$ would exactly diagonalize a quadratic Hamiltonian with interaction $Vf$ instead of $V$. To describe the action of such a Bogoliubov transformation, we define the coefficients $\sigma_{\infty,p} = \sinh \mu_{\infty,p}$ and $\gamma_{\infty,p} = \cosh \mu_{\infty,p}$. We have
\begin{equation}
\begin{split}\label{eq:sigmainfty}
&\sigma_{\infty,p}^2 = \frac{p^2 +  N^\kappa \widehat{Vf}(p/N^{1-\kappa})-\sqrt{p^4 + 2p^2 N^\kappa \widehat{Vf}(p/N^{1-\kappa})}}{2\sqrt{p^4 + 2p^2 N^\kappa \widehat{Vf}(p/N^{1-\kappa})}}
=\frac{1-2\eta_{\infty,p}-\sqrt{1-4\eta_{\infty,p}}}{2\sqrt{1-4\eta_{\infty,p}}},\\
&\gamma_{\infty,p}^2 = \frac{p^2 +  N^\kappa \widehat{Vf}(p/N^{1-\kappa})+\sqrt{p^4 + 2p^2 N^\kappa \widehat{Vf}(p/N^{1-\kappa})}}{2\sqrt{p^4 + 2p^2 N^\kappa \widehat{Vf}(p/N^{1-\kappa})}}
=\frac{1-2\eta_{\infty,p}+\sqrt{1-4\eta_{\infty,p}}}{2\sqrt{1-4\eta_{\infty,p}}},\\
&\gamma_{\infty,p}\sigma_{\infty,p} = \frac{-N^\kappa \widehat{Vf}(p/N^{1-\kappa})}{2\sqrt{p^4 + 2p^2 N^\kappa \widehat{Vf}(p/N^{1-\kappa})}}
=\frac{\eta_{\infty,p}}{\sqrt{1-4\eta_{\infty,p}}}.
\end{split}
\end{equation} 

Some useful properties of the coefficients $\eta_{\infty,p}$, $\sigma_{\infty,p}$ and $\gamma_{\infty,p}$ are collected in the following lemma. 
 \begin{lemma}\label{lemma:scattering}
 The scattering equation \eqref{eq:0en} translates into the discrete version
\begin{equation}\label{eq:scatteringdiscrete}
p^2 \eta_{\infty,p} +\frac{N^\kappa}{2} \hat{V}(p/N^{1-\kappa}) +\frac{N^\kappa}{2N} \sum_{q\in \Lambda_+^*} \hat{V}((p-q)/N^{1-\kappa})\eta_{\infty,q}
= O(N^{2\kappa-1})
\end{equation} 
uniformly in $p \in \Lambda^*_+$. Furthermore, 
\begin{equation}\label{eq:etacorrection}
\sum_{p\in \Lambda_+^*} N^\kappa \hat{V}(p/N^{1-\kappa})\eta_{\infty,p} 
= (8 \pi \mathfrak{a}- \hat{V}(0)) N^{1+\kappa} + O(N^{2\kappa}).
\end{equation} 
 The coefficients $\gamma_{\infty,p}$ and $\sigma_{\infty,p}$ satisfy
 \begin{equation}\label{eq:sigmainftymomentum}
 \abs{\gamma_{\infty,p}-1}\leq \abs{\sigma_{\infty,p}}
 \leq C \abs{\widehat{Vf}(p/N^{1-\kappa})} \,  \min \Big( \frac{N^{\kappa/4}}{\abs{p}^{1/2}}, \frac{N^{\kappa}}{\abs{p}^{2}} \Big),
 \end{equation}
 and therefore, for $s \in (3/2, 6)$, 
 \begin{equation}\label{eq:sigmainftypowers}
 \sum_{p\in \Lambda_+^*} \abs{\gamma_{\infty,p}-1}^s
 \leq \sum_{p\in \Lambda_+^*} \abs{\sigma_{\infty,p}}^s 
 \leq C N^{3\kappa/2}
 \end{equation}
 as well as
 \begin{equation}\label{eq:sigmainftyL1}
  \sum_{p\in \Lambda_+^*} \abs{\gamma_{\infty,p}-1}
 \leq \sum_{p\in \Lambda_+^*} \abs{\sigma_{\infty,p}}
 \leq C N.
 \end{equation}
 We also have
 \begin{equation}\label{eq:gammasigmainfty-etainfty}
 \abs{\gamma_{\infty,p}\sigma_{\infty,p} - \eta_{\infty,p}} \leq C \min \Big( \frac{N^{\kappa}}{\abs{p}^2}, \frac{N^{2\kappa}}{\abs{p}^4} \Big)
 \end{equation}
 and therefore
 \begin{equation}\label{eq:gammasigmainfty-etainftysum}
  \sum_{p\in \Lambda_+^*} \abs{\gamma_{\infty,p}\sigma_{\infty,p} - \eta_{\infty,p}} \leq C N^{3\kappa/2}.
 \end{equation}
 \end{lemma}

Before proving Lemma \ref{lemma:scattering}, we establish some preliminary bounds. 
\begin{lemma}\label{lm:Fm} 
For $m \in \bN$, let 
    \begin{equation}\label{eq:defFm}
    F_{m}(\xi) =  \Big( 1+ \sum_{k \in \bN^3: |k| \leq m} \abs{\cF[ x^k Vf] (\xi) }]\Big)^m - 1.
\end{equation}
Then, for every $m \in \bN$ there is a constant $C_m > 0$ such that  
\begin{equation}\label{eq:inftynormFm}
    \norm{F_m}_\infty \leq C_m
\end{equation}
and 
\begin{equation}\label{eq:normFm}
\begin{split}
    \sum_{p\in\Lambda^*} F_{m}^2((p+\zeta)/N^{1-\kappa})
    \leq C_m N^{3-3\kappa}
\end{split}
\end{equation}
for all $N \in \bN$ and all $\zeta\in\mathbb{R}^3$. Here, $\cF$ denotes the Fourier transform on the torus but since $Vf$ has compact support, it also equals the one on full space. Moreover, for $k = (k_1, k_2, k_3) \in \bN^3$, we use the standard notation $|k| = k_1 + k_2 + k_3$ and $x^k = x_1^{k_1} x_2^{k_2} x_3^{k_3}$. 
\end{lemma}

\begin{proof}
    The bound \eqref{eq:inftynormFm} follows immediately from $\abs{\cF[ x^{k} Vf](\xi)}\leq C$. Also, using this inequality we have
    \begin{equation*}
\begin{split}
    \sum_{p\in\Lambda^*} F_{m}^2((p+\zeta)/N^{1-\kappa})
    &\leq C_m  \sum_{|k| \leq m} \sum_{p\in\Lambda^*} \abs{\cF[ x^{k} Vf]((p+\zeta)/N^{1-\kappa})}^2\\
    &\leq C_m   \sum_{|k| \leq m} \sum_{p\in\Lambda^*} \abs{N^{3-3\kappa}\cF[(N^{1-\kappa} \cdot)^{k} \exp(-i\zeta\cdot) Vf(N^{1-\kappa}\cdot)](p)}^2\\
    &\leq C_m N^{3-3\kappa}  \sum_{|k| \leq m} 
    \int_\Lambda dx N^{3-3\kappa} \abs{(N^{1-\kappa} x)^{k} \exp(-i\zeta x) Vf(N^{1-\kappa}x)}^2\\
    &\leq C_m N^{3-3\kappa}
\end{split}
\end{equation*}
showing \eqref{eq:normFm}.
\end{proof}

We are now ready to show Lemma \ref{lemma:scattering}. 
\begin{proof}[Proof of Lemma \ref{lemma:scattering}]
We define $w = 1- f = (-\Delta)^{-1} Vf / 2$ so that $\eta_{\infty,p} = - N^{3\kappa-2} \widehat{w} (p/N^{1-\kappa})$. For $p \not = 0$, we obtain 
\begin{equation}\label{eq:computationscattering}
\begin{split}
&p^2 \eta_{\infty,p} +\frac{N^\kappa}{2} \hat{V}(\frac{p}{N^{1-\kappa}}) +\frac{N^\kappa}{2N} \sum_{q\in \Lambda_+^*} \hat{V}(\frac{p-q}{N^{1-\kappa}})\eta_{\infty,q}
\\ = & \frac{N^\kappa}{2} \widehat{Vw}(\frac{p}{N^{1-\kappa}}) +\frac{N^\kappa}{2N} \sum_{q\in \Lambda_+^*} \hat{V}(\frac{p-q}{N^{1-\kappa}})\eta_{\infty,q}\\
=  &\frac{N^\kappa}{2(2\pi)^3} \int dq \hat{V}(\frac{p}{N^{1-\kappa}}-q) \widehat{w}(q)
 +\frac{N^\kappa}{2N} \sum_{q\in \Lambda_+^*} \hat{V}(\frac{p-q}{N^{1-\kappa}})\eta_{\infty,q}\\
= &\frac{N^\kappa}{4N(2\pi)^3} \int dq \hat{V}(\frac{p-q}{N^{1-\kappa}}) \frac{N^\kappa \widehat{Vf}(\frac{q}{N^{1-\kappa}})}{q^2}
 -\frac{N^\kappa}{4N} \sum_{q\in \Lambda_+^*} \hat{V}(\frac{p-q}{N^{1-\kappa}})\frac{N^\kappa \widehat{Vf}(\frac{q}{N^{1-\kappa}})}{q^2}\\
 =&\frac{N^{2\kappa-1}}{4(2\pi)^3} \sum_{q\in \Lambda_+^*} 
 \int_{[-\pi,\pi]^3} d\zeta 
  \left(\hat{V}(\frac{p-q+\zeta}{N^{1-\kappa}})\frac{\widehat{Vf}(\frac{q+\zeta}{N^{1-\kappa}})}{(q+\zeta)^2}
 -\hat{V}(\frac{p-q}{N^{1-\kappa}})\frac{\widehat{Vf}(\frac{q}{N^{1-\kappa}})}{q^2}
 \right)
 \\ &+O(N^{2\kappa-1})
 \end{split} 
 \end{equation} 
where the error $O (N^{2\kappa-1})$ arises from the contribution associated with $q=0$. Hence
 \[ \begin{split}
&p^2 \eta_{\infty,p} +\frac{N^\kappa}{2} \hat{V}(p/N^{1-\kappa}) +\frac{N^\kappa}{2N} \sum_{q\in \Lambda_+^*} \hat{V}((p-q)/N^{1-\kappa}) \eta_{\infty,q}  \\
  &=\frac{N^{2\kappa-1}}{4(2\pi)^3} \sum_{q\in \Lambda_+^*} 
 \int_{[-\pi,\pi]^3} d\zeta 
  \int_0^1 ds \frac{d}{ds}\left(\hat{V}(\frac{p-q+s\zeta}{N^{1-\kappa}})\frac{\widehat{Vf}(\frac{q+s\zeta}{N^{1-\kappa}})}{(q+s\zeta)^2}\right)
 +O(N^{2\kappa-1})
 \end{split} \]
 and, since the $s$-derivative evaluated at $s=0$ is odd in $\zeta$, 
 \begin{equation} \label{eq:sca2} \begin{split} 
 &p^2 \eta_{\infty,p} +\frac{N^\kappa}{2} \hat{V}(p/N^{1-\kappa}) +\frac{N^\kappa}{2N} \sum_{q\in \Lambda_+^*} \hat{V}((p-q)/N^{1-\kappa}) \eta_{\infty,q}  \\  &=\frac{N^{2\kappa-1}}{4(2\pi)^3} \sum_{q\in \Lambda_+^*} 
 \int_{[-\pi,\pi]^3} d\zeta 
  \int_0^1 ds \int_0^s dt
  \frac{d^2}{dt^2}\left(\hat{V}(\frac{p-q+t\zeta}{N^{1-\kappa}})\frac{\widehat{Vf}(\frac{q+t\zeta}{N^{1-\kappa}})}{(q+t\zeta)^2}\right)
 +O(N^{2\kappa-1}).
\end{split}
\end{equation} 
Using the function $F_2$, defined in (\ref{eq:defFm}), we can bound 
\begin{equation*}
    \begin{split}
    \sum_{q\in \Lambda_+^*} 
 \int_{[-\pi,\pi]^3} &d\zeta 
  \int_0^1 ds \int_0^s dt \; 
  \Big| \frac{d^2}{dt^2}\left(\hat{V}(\frac{p-q+t\zeta}{N^{1-\kappa}})\frac{\widehat{Vf}(\frac{q+t\zeta}{N^{1-\kappa}})}{(q+t\zeta)^2}\right) \Big| \\
\leq \; &C \sum_{q\in \Lambda_+^*} 
 \int_{[-\pi,\pi]^3} d\zeta 
  \int_0^1 ds \int_0^s dt \, F_2(\frac{q+t\zeta}{N^{1-\kappa}})
\sum_{n=0}^2 N^{-(1-\kappa)(2-n)}\abs{q+t\zeta}^{-2-n}\\
     \leq \; &C 
 \int_{[-\pi,\pi]^3} d\zeta 
  \int_0^1 ds \int_0^s dt 
  \sum_{\substack{q\in \Lambda_+^*\\ \abs{q+t\zeta}\leq N^{1-\kappa}}}
  \norm{F_2}_\infty
\sum_{n=0}^2 N^{-(1-\kappa)(2-n)}\abs{q+t\zeta}^{-2-n}\\
 &+C 
 \int_{[-\pi,\pi]^3} d\zeta 
  \int_0^1 ds \int_0^s dt 
  \sum_{\substack{q\in \Lambda_+^*\\ \abs{q+t\zeta} > N^{1-\kappa}}}
  F_2(\frac{q+t\zeta}{N^{1-\kappa}})
 N^{-2(1-\kappa)}\abs{q+t\zeta}^{-2}\\
 \leq \; &C.
    \end{split}
\end{equation*}
To bound the sum over $|q+t \zeta| > N^{1-\kappa}$, we used Cauchy-Schwarz and (\ref{eq:normFm}). 
From (\ref{eq:sca2}), we obtain (\ref{eq:scatteringdiscrete}), uniformly in $p \in \Lambda^*_+$. 

To prove (\ref{eq:etacorrection}) we observe that, starting from the second line,  (\ref{eq:computationscattering}) also holds for $p=0$. This implies that 
\[ N^{1+\kappa} \widehat{Vw} (0) + N^\kappa \sum_{q \in \Lambda^*_+} \hat{V} ((p-q/N^{1-\kappa}) \eta_{q} = O (N^{2\kappa}). \]
With $\widehat{Vw} (0) = \hat{V} (0) - \widehat{Vf} (0) = \hat{V} (0) - 8\pi \frak{a}$, we obtain  (\ref{eq:etacorrection}).
 
The bounds \eqref{eq:sigmainftymomentum} to \eqref{eq:gammasigmainfty-etainftysum} are easy to check using the explicit form in \eqref{eq:sigmainfty}.
\end{proof}

Let $\beta = 2-3\kappa > 0$. Noticing with (\ref{eq:sigmainftymomentum}) that 
\begin{equation}\label{eq:Nbeta4} N^\kappa \sum_{|p| < N^\alpha} \sigma_{\infty,p}^2 \leq C N^{3\kappa/2+2\alpha} \leq CN^{4\kappa-1} \end{equation} 
if $\alpha \leq \kappa/2 - \beta/4$, we conclude that excitations with momentum below $N^{\kappa/2-\beta/4}$ are negligible, to resolve the energy to the precision we are aiming for (the l.h.s. of (\ref{eq:Nbeta4}) arises from the quadratic term $N^\kappa \sum_p \hat{V} (p/N^{1-\kappa}) a_p^* a_p$, after conjugation with the Bogoliubov transformation; also the contribution of the other quadratic terms in the excitation Hamiltonian is negligible, for $|p| < N^{\kappa/2 - \beta/4}$). Thus, we can impose an infrared cutoff on the coefficients $\mu_p$ of the Bogoliubov transformation defined in (\ref{eq:BT0}). In order to guarantee decay in position space, it is important to use a smooth cutoff. We choose $\chi_l \in C^{\infty}(\mathbb{R})$ with $0\leq \chi_l \leq 1$, $\chi_l(x)=0$ if $x\leq 1$, $\chi_l(x)=1$ if $x\geq 2$, and we set $\chi (p) = \chi_l (\abs{p}/N^{\kappa/2-\beta/4})$. For $p \in \Lambda^*$, we define 
\begin{equation}\label{eq:defmu}
    \mu_p \coloneqq \mu_{\infty ,p} \chi (p)
\end{equation}
with $\mu_{\infty,p}$ as in (\ref{eq:defmuinfty}). With these coefficients we set  
\begin{equation*} B = \frac{1}{2} \sum_{p \in \Lambda^*} \mu_p a_p a_{-p} \, . \end{equation*} 
The trial state (\ref{eq:psi0}) is defined through the Bogoliubov transformation $e^{B-B^*}$, whose action on creation and annihilation operators is determined by 
\begin{equation}\label{eq:BT1} e^{B-B^*} a_p e^{B^*-B} = \gamma_p a_p + \sigma_p a_{-p}^*, \qquad e^{B-B^*} a^*_p e^{B^*-B} = \gamma_p a^*_p + \sigma_p a_{-p} \end{equation} 
where, similarly as in (\ref{eq:sigmainfty}), we introduced the notation $\sigma_p = \sinh \mu_p$, $\gamma_p = \cosh \mu_p$, for all $p \in \Lambda^*$. In position space, we set 
\[ \check{\sigma} (x) = \sum_{p \in \Lambda^*} \sigma_p e^{i p \cdot x} , \qquad \text{and} \qquad \widecheck{\gamma-1} (x) = \sum_{p\in\Lambda^*} (\gamma_p-1) e^{i p \cdot x}. \]
We will also make use of continuous versions of the coefficients $\mu_p, \sigma_p, \gamma_p$. Analogously to (\ref{eq:defetainfty}), (\ref{eq:defmuinfty}), we define therefore, for every $p \in \bR^3$, 
\[ \eta_\infty (p) = -N^\kappa \frac{\widehat{Vf} (p/ N^{1-\kappa})}{2 p^2} ,  \quad \mu_\infty (p) = -\frac{1}{4}  \log (1-4 \eta_\infty (p)),  \]
as well as $\sigma_\infty (p) = \sinh \mu_\infty (p)$, $\gamma_\infty (p) = \cosh \mu_\infty (p)$ and $\mu (p) = \mu_\infty (p) \chi (p)$, $\sigma (p) = \sinh \mu (p)$ and $\gamma (p) = \cosh \mu (p)$. Notice that 
\begin{equation}\label{eq:mu(p)-def} 
\mu(p) = \left\{ \begin{array}{ll} 0 &\textrm{ if } \abs{p}\leq  N^{\kappa/2-\beta/4}, \\ 
\mu_{\infty} (p) &\textrm{ if } \abs{p} \geq 2 N^{\kappa/2-\beta/4} \end{array} \right. \end{equation} 
and also that $|\mu (p)| \leq |\mu_\infty (p)|$ for all $p \in \bR^3$. This implies that $|\sigma (p)| \leq |\sigma_\infty (p)|$, $0 \leq \gamma (p) - 1 \leq \gamma_\infty (p) - 1$, $\sigma (p) = 0$ and $\gamma (p) = 1$ for $|p| < N^{\kappa/2-\beta/4}$.  

In the next lemma, we collect various properties of the coefficients $\sigma_p, \gamma_p$ and of the corresponding Fourier series. 
\begin{lemma}\label{lemma:propertiesquadratickernels}
In momentum space we have the bounds
\begin{equation}\label{eq:boundsigmamomentum}
    \abs{\gamma (p) -1}
    \leq\abs{\sigma (p) }
    \leq C \abs{\widehat{Vf}(p/ N^{1-\kappa})} \min(N^{\kappa/4}\abs{p}^{-1/2}, N^\kappa  \abs{p}^{-2})
\end{equation}
which imply
\begin{equation}\label{eq:Linftynormsigma}
    \norm{\sigma}_\infty, \norm{\gamma-1}_\infty\leq CN^{\beta/8},
\end{equation}
\begin{equation}\label{eq:L1normsigma}
    \sum_{p\in\Lambda^*} \abs{\sigma_p},  \sum_{p\in\Lambda^*} \abs{\gamma_p-1} 
    \leq C N
\end{equation}
and, for $s \in (3/2, 6)$,
 \begin{equation}\label{eq:sigmapowers}
 \sum_{p\in \Lambda^*} \abs{\gamma_{p}-1}^s
 \leq \sum_{p\in \Lambda^*} \abs{\sigma_{p}}^s 
 \leq C N^{3\kappa/2}.
 \end{equation}

Furthermore, we have the following bounds in position space:
\begin{equation}\label{eq:boundsigmaN}
\abs{\check\sigma (x)}, \abs{\widecheck{\gamma-1} (x)} \leq CN
\end{equation}

\begin{equation}\label{eq:boundsigmax}
\abs{\check\sigma (x)}, \abs{\widecheck{\gamma-1} (x)}
\leq C (\log N)^2 N^\kappa \abs{x}^{-1}
\end{equation}

\begin{equation}\label{eq:boundsigmax5/2}
\abs{\check\sigma (x)}, \abs{\widecheck{\gamma-1} (x)}
\leq C (\log N)^3 N^{\kappa/4} \abs{x}^{-5/2}
\end{equation}
and, for any $m\geq 2$, there exists $C_m > 0$ such that 
\begin{equation}\label{eq:boundsigmaxm}
\abs{\check\sigma (x)}, \abs{\widecheck{\gamma-1} (x)}
\leq C_m  N
\left(\frac{\ell_\sigma}{\abs{x}}\right)^m
\end{equation}
where $\ell_\sigma=N^{-\kappa/2+\beta/4}$.
\end{lemma}

\textit{Remark.} The decay $|x|^{-5/2}$ in (\ref{eq:boundsigmax5/2}) corresponds to the bound proportional to $|p|^{-1/2}$ in (\ref{eq:boundsigmamomentum}). 

\begin{proof}
The bounds in momentum space follow similarly to Lemma \ref{lemma:scattering}, since  $|\sigma (p)| \leq |\sigma_\infty (p)|$, $|\gamma (p) - 1| \leq |\gamma_\infty (p) - 1|$. 

To establish the estimates in position space, we study the derivatives of $\mu (p), \sigma (p), \gamma (p)$. First, we observe that $|\mu (p)| \leq C \log N$, for $|p| < N^{\kappa/2}$, and $|\mu (p)| \leq C N^\kappa / |p|^2$, for $|p| > N^{\kappa/2}$. Hence, we have 
\begin{equation}\label{eq:mup-bd}
    \abs{\mu(p)}  \leq C (\log N) \frac{N^\kappa}{N^\kappa+p^2}
\end{equation}
for all $p \in \bR^3$. Also,
\begin{equation}
\partial_{p_j}\mu_\infty(p) 
= 
-\frac{1}{4} \left( 1+ \frac{2N^\kappa \widehat{Vf}(\frac{p}{N^{1-\kappa}})}{p^2} \right)^{-1} 
\partial_{p_j}\left( \frac{2N^\kappa \widehat{Vf}(\frac{p}{N^{1-\kappa}})}{p^2} \right) \, .
\end{equation}
Thus, recalling the definition \eqref{eq:defFm}, we find, for any integer $m \geq 1$, 
\begin{equation*}
\begin{split}
\abs{\partial_{p_j}^m\mu_\infty(p)}
&\leq C_m
\sum_{k=1}^m
\Big| 1+ \frac{2N^\kappa \widehat{Vf}(\frac{p}{N^{1-\kappa}})}{p^2} \Big|^{-k} 
\sum_{\substack{m_1, \dots ,m_k \geq 1 : \\  \sum_{i=1}^k m_i = m}} 
\prod_{i=1}^{k}
\Big| \partial_{p_j}^{m_i}\left( \frac{N^\kappa \widehat{Vf}(\frac{p}{N^{1-\kappa}})}{p^2} \right) \Big| \\
&\leq
C_m\sum_{k=1}^m
\left(\frac{p^2}{p^2+N^\kappa}\right)^{k} 
F_{m}(p/N^{1-\kappa}) \\  &\hspace{4cm} \times 
\sum_{\substack{m_1, \dots ,m_k \geq 1 : \\  \sum_{i=1}^k m_i = m}} 
\prod_{i=1}^{k}
\sum_{n_i = 0}^{m_i} N^\kappa \abs{p}^{-2-n_i} N^{-(1-\kappa) (m_i-n_i)}.
\end{split} \end{equation*}
Setting $n = n_1 + \dots + n_k$ and absorbing combinatorial factors from the sums over $m_1, \dots , m_k \geq 1$ with $m_1 + \dots +m_k = m$ and over $n_1, \dots , n_k \geq 0$ with $n_1 + \dots + n_k = n \leq m$, we arrive at 
\begin{equation}\label{eq:partialmuinf} \begin{split} 
\abs{\partial_{p_j}^m\mu_\infty(p)} &\leq
C_m F_{m}(p/N^{1-\kappa}) \frac{N^\kappa}{p^2+N^\kappa}\sum_{n=0}^{m} \abs{p}^{-n} N^{-(1-\kappa) (m-n)}\\
&\leq
C_m F_{m}(p/N^{1-\kappa}) \frac{N^\kappa}{p^2+N^\kappa}\sum_{n=0}^{m} \abs{p}^{-n} N^{-(1-\kappa) (m-n)}.
\end{split}
\end{equation} 
Recalling that $\chi (p) = \chi_l (|p|/N^{\kappa/2-\beta/4})$, we find, for integer $m \geq 1$, 
\begin{equation}
\begin{split} 
    \abs{\partial_{p_j}^m \chi (p)} &\leq C_m  N^{-m (\kappa/2 - \beta/4)}\mathbbm{1}_{[N^{\kappa/2-\beta/4}, 2N^{\kappa/2-\beta/4}]}(\abs{p}) \\ &\leq C_m \abs{p}^{-m} \mathbbm{1}_{[N^{\kappa/2-\beta/4}, 2N^{\kappa/2-\beta/4}]}(\abs{p}) \, .
\end{split} 
\end{equation}
Combining this bound with (\ref{eq:partialmuinf}), we obtain, for $m \geq 0$,  
\begin{equation*}
    \abs{\partial_{p_j}^m\mu(p)} 
    \leq C_m F_{m}(p/N^{1-\kappa}) (\log N) \frac{N^\kappa}{N^\kappa+p^2} \mathbbm{1}(\abs{p}\geq N^{\kappa/2-\beta/4})
 \sum_{n=0}^{m} \abs{p}^{-n} N^{-(1-\kappa) (m-n)}
    \end{equation*}
where the factor $(\log N)$ comes from (\ref{eq:mup-bd}), to handle the contribution in which all derivatives hit $\chi$. We apply now the last bound to control derivatives of $\sigma (p), \gamma (p)$. For $m \geq 0$, we obtain 
\[ \begin{split}
    \abs{\partial_{p_j}^m\sigma(p)}
    \leq \; &
    C_m
    \sum_{k=1}^m \abs{\sinh^{(k)}(\mu (p))}
     \sum_{\substack{m_1,\dots,m_k\geq 1: \\ \sum_{i=1}^k  m_i =m}} \prod_{i=1}^k
    \abs{\partial_{p_j}^{m_i} \mu(p)
    }\\
    \leq \; & 
    C_m
    \Big( 1+\min \Big( \frac{N^{\kappa/4}}{\abs{p}^{1/2}} , \frac{N^\kappa}{\abs{p}^{2}} \Big) \Big)
    (\log N)^m\frac{N^\kappa}{N^\kappa+p^2} \, 
    F_{m}(p/N^{1-\kappa}) \\ &\hspace{4cm} \times 
    \sum_{n=0}^{m} \abs{p}^{-n} N^{-(1-\kappa) (m-n)}
    \mathbbm{1}(\abs{p}\geq N^{\kappa/2-\beta/4})
\end{split} \]  
where we used (\ref{eq:boundsigmamomentum}) to estimate $|\sinh^{(k)} (\mu (p))|$ and the bound $\prod_{j=1}^k F_{m_j} (p) \leq F_m (p)$, if $m_1 + \dots + m_k = m$ (several combinatorial factors, arising from sums over $m_1, \dots , m_k \geq 1$ with fixed $m = m_1 +\dots +m_k$ and also from sums over $n_1, \dots , n_k \geq 0$ with fixed $n_1 + \dots + n_k = n$ have been absorbed in the constant $C_m$). Distinguishing $|p| > N^{\kappa/2}$ and $|p| \leq N^{\kappa/2}$ and using the factor $N^\kappa / (N^\kappa + p^2)$ as well as using Young's inequality for the sum over $n$, we obtain 
\begin{equation}\label{eq:partialmsi} \begin{split}  
\abs{\partial_{p_j}^m\sigma(p)} \leq \; & 
    C_m  (\log N)^m
    \min \Big( \frac{N^{\kappa/4}}{\abs{p}^{1/2}}, \frac{N^\kappa}{\abs{p}^{2}} \Big)
    F_{m}(p/N^{1-\kappa}) \mathbbm{1}(\abs{p}\geq N^{\kappa/2-\beta/4}) 
  \\ &\hspace{7cm} \times   \big(|p|^{-m}+N^{-m(1-\kappa)}\big).
\end{split}
\end{equation} 
Analogously, we find 
\[ 
\begin{split}
    \abs{\partial_{p_j}^m\gamma(p)}
    &\leq 
    C_m  (\log N)^m 
    \min \Big( \frac{N^{\kappa/4}}{\abs{p}^{1/2}}, \frac{N^\kappa}{\abs{p}^{2}} \Big)
    F_{m}(p/N^{1-\kappa}) \mathbbm{1}(\abs{p}\geq N^{\kappa/2-\beta/4}) \\
 &\hspace{7cm} \times  \big(|p|^{-m}+N^{-m(1-\kappa)}\big) \, . 
\end{split}
\]
Now we apply these bounds to show the pointwise bounds in position space.
The bounds \eqref{eq:boundsigmaN} follow immediately from \eqref{eq:L1normsigma}.
Next, we observe that, for $x\in\Lambda  = [-1/2 ; 1/2]^3$ and any integer $m\geq 1$, we have
\begin{equation}
\begin{split}\label{eq:trickpowersx} 
c^m \abs{x_j}^m \abs{\check\sigma (x) }
    &\leq
    \abs{(e^{2\pi i x_j}-1)^m\check\sigma(x) }
    = \abs{(e^{2\pi i x_j}-1)^m \sum_{p\in\Lambda^*} \sigma(p)e^{ip\cdot x}}\\
    &= \abs{(e^{2\pi i x_j}-1)^{m-1} \sum_{p\in\Lambda^*} (\sigma (p-2\pi e_j)-\sigma (p))e^{ip\cdot x}}\\
    &= \Big| \int_0^{2\pi} ds \sum_{p\in\Lambda^*}  (e^{2\pi i x_j}-1)^{m-1} \partial_{p_j}\sigma (p- s e_j)e^{ip\cdot x} \Big| \\
&\leq  \int_0^{2\pi} ds_1\dots \int_0^{2\pi} ds_m \sum_{p\in\Lambda^*}  \Big| \partial_{p_j}^m\sigma(p-\sum_{i=1}^m s_i e_j) \Big| \\
&\leq (2\pi)^{m-1} \int_0^{2\pi m} ds \sum_{p\in\Lambda^*}  \abs{ \partial_{p_j}^m\sigma (p-s e_j)}.
\end{split}
\end{equation}
Taking $m=1$ and using (\ref{eq:partialmsi}), we find 
\[
\begin{split}
    \abs{x_j}\abs{\check\sigma (x) }
    &\leq C (\log N)
    \int_0^{2\pi} ds 
    \sum_{\substack{p\in\Lambda^*_+ : \\ |p-se_j| \geq N^{\kappa/2-\beta/4}}} 
    \frac{N^{\kappa}(N^{-(1-\kappa)}+\abs{p-s e_j}^{-1})}{\abs{p- s e_j}^2}  \\ &\hspace{7cm} \times  F_1((p-s e_j)/ N^{1-\kappa}) .
\end{split} \]    
We divide the sum into the two regions $N^{\kappa/2 - \beta/4} \leq |p-s e_j| \leq N^{1-\kappa}$ and $|p-se_j| > N^{1-\kappa}$. In the first region, we bound $F_1$ in the $\ell^\infty$-norm. In the second region, we use instead the $\ell^2$-norm of $F_1$. With \eqref{eq:normFm}, we obtain 
\[     \abs{x_j}\abs{\check\sigma (x) } \leq  C (\log N)^2 N^\kappa. \]
%
%

Combining the bounds with different choices of $j$, we obtain 
\begin{equation}
   \abs{\check\sigma (x) }  
   \leq C (\log N)^2 N^\kappa \abs{x}^{-1}.
\end{equation}
Since the estimate for $\gamma - 1$ can be shown similarly, this concludes the proof of \eqref{eq:boundsigmax}. 

Next, we prove (\ref{eq:boundsigmax5/2}). To this end, we can restrict our attention to $|x| \geq N^{-\kappa/2}$ (for $|x| < N^{-\kappa/2}$, (\ref{eq:boundsigmax}) is stronger). We split $\check\sigma (x) = \check\sigma_{1} (x) + \check\sigma_{2} (x)$, with 
\[ \check\sigma_{1} (x) = \sum_{p\in\Lambda^*} \sigma (p) (1-\chi_l (\abs{p}\abs{x})) e^{ip \cdot x},  \qquad \check\sigma_2 (x) = \sum_{p\in\Lambda^*}  \sigma(p) \chi_l (\abs{p}\abs{x})  e^{ip\cdot x} .\]  
Recalling that $\sigma (p) = 0$ for $|p| < N^{\kappa/2 - \beta/4}$, we have 
\[ 
\abs{\sigma (p)}(1-\chi_l(\abs{p} \abs{x}))
\leq \abs{\sigma (p)} \mathbbm{1}_{[N^{\kappa/2-\beta/4}, 2\abs{x}^{-1}]}(\abs{p}). \]
With the bound \eqref{eq:boundsigmamomentum}, we immediately obtain 
\begin{equation} \label{eq:sigma1-bd} 
\abs{\check\sigma_{1} (x)}
\leq C \sum_{\substack{p\in\Lambda_+^*: \\ N^{\kappa/2-\beta/4}\leq \abs{p} \leq 2\abs{x}^{-1}}} N^{\kappa/4}\abs{p}^{-1/2}
\leq C N^{\kappa/4} \abs{x}^{-5/2}.
\end{equation}

As for the second term, note that, for integer $m \geq 0$,   
\begin{equation}
{\partial_{p_j}^m \chi_l(\abs{p}\abs{x})}
\leq C_m \abs{p}^{-m}\mathbbm{1}(\abs{p}\geq \abs{x}^{-1}),
\end{equation}
and therefore 
\begin{equation}\label{eq:sigma2-bd} 
\begin{split}
    \abs{\partial_{p_j}^m(\sigma (p)\chi_l(\abs{p}\abs{x}))}
    &\leq 
    C_m  (\log N)^m
    \min \Big( \frac{N^{\kappa/4}}{\abs{p}^{1/2}}, \frac{N^\kappa}{\abs{p}^{2}} \Big) \\
    &\hspace{1cm} \times F_{m}(p/N^{1-\kappa}) \mathbbm{1}(\abs{p}\geq \abs{x}^{-1}) \big(|p|^{-m}+N^{-m(1-\kappa)}\big).
\end{split}
\end{equation}
Proceeding as in (\ref{eq:trickpowersx}), we find 
\begin{equation*}
\begin{split}
c^3 \abs{x_j}^3 \abs{\check\sigma_{2} (x) }
\leq \; &(2\pi)^2 \int_0^{6\pi} ds
\sum_{p\in\Lambda^*}  
\abs{\partial_{p_j}^3(\sigma \chi_l(\cdot \abs{x}))(p-  s e_j)}\\
\leq \; &C_3 (\log N)^3 \int_0^{6\pi} ds 
\sum_{\substack{p\in\Lambda^*\\\abs{p-s e_j}\geq \abs{x}^{-1}}} 
\frac{N^{\kappa/4}}{ \abs{p-s e_j}^{7/2}} F_{3} ((p-se_j)/N^{1-\kappa}) \\
&+C_3(\log N)^3\int_0^{6\pi}ds 
\sum_{\substack{p\in\Lambda^*\\ \abs{p- s e_j} \geq |x|^{-1}}} 
\frac{N^{\kappa}}{|p-se_j|^2} F_3 ((p-s e_j)/N^{1-\kappa}) N^{-3(1-\kappa)}. 
\end{split}
\end{equation*}
 Estimating $F_3$ with its $\ell^\infty$-norm (in the first term) and its $\ell^2$-norm (in the second term) and applying Lemma \ref{lm:Fm}, we obtain
\[ \abs{x_j}^3 \abs{\check\sigma_{2} (x) } \leq C (\log N)^3 N^{\kappa/4} |x|^{1/2} .\]
Combined with (\ref{eq:sigma1-bd}), this implies (\ref{eq:boundsigmax5/2}) (the bound for $1-\gamma$ can be shown similarly). 

Finally, we show \eqref{eq:boundsigmaxm}. Here, it is enough to consider $|x| > \ell_\sigma = N^{-\kappa/2 + \beta/4}$ and $m\geq 2$. Then, we have, 
\[
\begin{split}
c^m \abs{x_j}^m &\abs{\check\sigma (x) }
\\ \leq  \; & (2\pi)^m\int_0^{2\pi m} ds \sum_{p\in\Lambda^*}  \abs{(\partial_{p_j}^m\sigma (p- s e_j)}\\
\leq \; &C_m  (\log N)^m 
\int_0^{2\pi m} ds \sum_{\substack{p\in\Lambda^*\\ \abs{p-s e_j} \geq N^{\kappa/2-\beta/4}}}
\frac{N^\kappa}{\abs{p- s e_j}^{2}} F_m((p-s e_j)/N^{1-\kappa})
 N^{-(1-\kappa) m}\\
&+ C_m  (\log N)^m 
\int_0^{2\pi m} ds \sum_{\substack{p\in\Lambda^*\\ \abs{p- s e_j} \geq N^{\kappa/2-\beta/4}}}
\frac{N^\kappa}{ \abs{p- s e_j}^{m+2} } F_m ((p-se_j)/N^{1-\kappa}).
\end{split} 
\]
Estimating $F_m$ with its $\ell^2$-norm in the first term and with its $\ell^\infty$-norm in the second term, we conclude that 
\[ \begin{split} c^m \abs{x_j}^m \abs{\check\sigma (x) } &\leq C_m  (\log N)^{m} 
N N^{-m(\kappa/2-\beta/4)}.
\end{split}
\]
Bounds for $\gamma-1$ can be proven analogously; this implies (\ref{eq:boundsigmaxm}).
\end{proof}

In our analysis of the action of the cubic transformation $e^{A^* \Theta - \Theta A}$, it will often be convenient to localize kernels in small boxes, to take advantage of the presence of the cutoff $\Theta$, controlling the local number of particles. Recalling that $\beta = 2-3\kappa > 0$, we introduce the length scale $\ell_B =  N^{-\kappa/2 + \beta/12}$. We consider the lattice $\Lambda_B = \Lambda \cap 
 \ell_B \mathbb{Z}^3$ and, for $u\in \Lambda_B$, we define the box $B_u = \{x\in\Lambda, \inf_{v\in \Lambda_B} \abs{x-v} =\abs{x-u} \}$. We denote by $\check\sigma_u = \check\sigma \mathbbm{1}_{B_u}$ and $(\widecheck{\gamma-1})_u = (\widecheck{\gamma-1}) \mathbbm{1}_{B_u}$ the restrictions of $\check{\sigma}$ and $\widecheck{1-\gamma}$ on $B_u$. Clearly, $\check\sigma = \sum_{u \in\Lambda_B} \check\sigma_u$ and $\widecheck{\gamma-1} = \sum_{u\in\Lambda_B} (\widecheck{\gamma-1})_u$. 
 \begin{lemma} \label{lemma:sumlocalizedsigma}
For every $\epsilon > 0$, there is $C > 0$ such that    
\begin{equation}\label{eq:sigmau2}
        \sum_{u\in\Lambda_B} \norm{\check\sigma_u}_2,
        \sum_{u\in\Lambda_B} \norm{(\widecheck{\gamma-1})_u}_2 \leq C N^{3\kappa/4 +\epsilon}
    \end{equation}
  and 
    \begin{equation}\label{eq:sigmau1}
        \sum_{u\in\Lambda_B} \norm{\check\sigma_u}_1,
        \sum_{u\in\Lambda_B} \norm{(\widecheck{\gamma-1})_u}_1
        \leq C N^{\beta/8 + \epsilon}.
    \end{equation}
\end{lemma}

\textit{Remark.} Eq. \eqref{eq:sigmau2} implies that $\sum_{u \in\Lambda_B} \| \check{\sigma}_u \|_2$ is comparable (up to the additional factor $N^\epsilon$) with $\| \check{\sigma} \|_2$. In fact, $\ell_B$ is the smallest length scale with this property (this is the reason for fixing $\ell_B = N^{-\kappa/2 + \beta/12}$). 

\begin{proof}
    We only show the bounds for $\check\sigma$, the others follow analogously. For any $\delta > 0$, we have    \[
        \begin{split}
            &\sum_{u\in\Lambda_B} \norm{\check\sigma_u}_2
            = \norm{\check\sigma_0}_2
            + \sum_{0\neq u\in\Lambda_B, \abs{u}\leq \ell_\sigma N^\delta} \norm{\check\sigma_u}_2
            + \sum_{u\in\Lambda_B, \abs{u}> \ell_\sigma N^\delta} \norm{\check\sigma_u}_2.
      \end{split} \]      
The first term can be bounded by $\| \check{\sigma}_0 \|_2 \leq \| \check{\sigma} \|_2 \leq C N^{3\kappa/4}$, using (\ref{eq:sigmapowers}). To estimate the second and third terms, we use the pointwise bounds (\ref{eq:boundsigmax5/2}) and, respectively, (\ref{eq:boundsigmaxm}). We find, for any integer $m \geq 2$, 
\[ \begin{split} 
      \sum_{u\in\Lambda_B} \norm{\check\sigma_u}_2      
            \leq \; &C N^{3\kappa/4}
            + C (\log N)^3  \sum_{0\neq u\in\Lambda_B, \abs{u}\leq \ell_\sigma N^\delta} N^{\kappa/4} \ell_B^{3/2} \abs{u}^{-5/2}
            \\ &+ C \sum_{u\in\Lambda_B, \abs{u}> \ell_\sigma N^\delta}  N N^{-\delta m} \ell_B^{3/2}  
             \\
            \leq \; &C N^{3\kappa/4}
            + C (\log N)^3 N^{\kappa/4}  \ell_\sigma^{1/2} N^{\delta/2} \ell_B^{-3/2}
            + C N N^{-\delta m} \ell_B^{-3/2}\\
            \leq \; &C (\log N)^3 N^{3\kappa/4}  N^{\delta/2}
            + C N N^{-\delta m} \ell_B^{-3/2}.
        \end{split}
    \]
Choosing $\delta$ sufficiently small and then $m$ sufficiently large, we obtain (\ref{eq:sigmau2}). To prove (\ref{eq:sigmau1}), we observe that $\sum_{u \in \Lambda_B} \| \check{\sigma}_u \|_1 \leq \| \check{\sigma} \|_1$ and we use the pointwise bounds (\ref{eq:boundsigmax5/2}) for $|x| \leq N^{-\kappa/2 + \beta /4 + \delta}$ and (\ref{eq:boundsigmaxm}) for $|x| > N^{-\kappa/2 + \beta /4 + \delta}$.
\end{proof}

\subsection{Coefficients of the Cubic Transformation}
\label{subsec:kerA}

In this subsection, we are going to introduce the coefficients $A_{p,q}$, defining the cubic transformation (\ref{eq:A-def}). We proceed similarly as in \cite{COSS}. From (\ref{eq:defmu}), we recall the cutoff function $\chi_l \in C^{\infty}(\mathbb{R})$, satisfying $0\leq \chi_l \leq 1$, $\chi_l(x)=0$ if $x\leq 1$, $\chi_l(x)=1$ if $x\geq 2$. We define $\tilde\chi (p) = \chi_l (\abs{p}/N^{\kappa/2})$ and, for every $p, q \in \bR^3 \backslash \{ 0 \}$, 
\begin{equation}\label{eq:cubicfunctionmomentumspace}
A(p,q) = N^{-1/2}\tilde\chi(p) \frac{-N^\kappa \widehat{Vf}(\frac{p}{N^{1-\kappa}})\sigma (q)}{p^2 + q^2 + (p+q)^2}
= N^{-1/2}\frac{2p^2\eta(p)\sigma(q)}{p^2 + q^2 + (p+q)^2},
\end{equation}
with $\sigma (q)$ as defined above (\ref{eq:mu(p)-def}) and 
\begin{equation}\label{eq:defeta}
    \eta(p) = \frac{-N^\kappa \widehat{Vf}(\frac{p}{N^{1-\kappa}})}{2p^2}\tilde\chi(p) \, .
\end{equation}
Notice that the infrared cutoffs act differently in the two variables. In $p$, we restrict to momenta larger than $N^{\kappa/2}$. In $q$, on the other hand, we have to consider momenta larger than $N^{\kappa /2 - \beta/4}$. For $p,q \in \Lambda^*_+$, we also use the notation $\eta_p = \eta (p)$ and $A_{p,q} = A (p,q)$. For completeness, we set $A_{p,q} = 0$, if $p=0$ or $q=0$. 
In position space, we denote 
\begin{equation}\label{eq:defAxy}
    \check{A}(x,y) = \sum_{p,q \in \Lambda^*} A_{p,q} e^{ipx}e^{iqy}.
\end{equation} 
Several properties of the coefficients $A_{p,q}, \eta_p$ and of the function $\check{A}$ are collected in the next lemma.
\begin{lemma}\label{lemma:propertiesA}
First, we list bounds in momentum space. From (\ref{eq:boundsigmamomentum}), we recall 
\begin{equation}\label{eq:sigp} \abs{\sigma_p}\leq C \abs{\widehat{Vf}(p/ N^{1-\kappa})} \min(N^{\kappa/4}\abs{p}^{-1/2}, N^\kappa  \abs{p}^{-2}).\end{equation} 
With the trivial bound
\[ \abs{\eta_p} \leq  N^\kappa \abs{\widehat{Vf}(p/ N^{1-\kappa})}   \abs{p}^{-2} \]
and recalling $\eta_p = 0$ for $|p| < N^{\kappa/2}$, we find 
\begin{equation}\label{eq:normseta}
    \norm{\eta}_\infty\leq C, \quad \norm{\eta}_2^2\leq C N^{3\kappa/2}, \quad \norm{\eta}_1\leq CN.
\end{equation}
Moreover 
 \begin{equation}\label{eq:scatteringetacutoff}
        \frac{N^\kappa}{N} \sum_{q\in \Lambda_+^*} \hat{V}(\frac{p-q}{N^{1-\kappa}})\eta_{q}
        = - N^\kappa \widehat{Vw}(\frac{p}{N^{1-\kappa}})
        +O(N^{5\kappa/2-1})
    \end{equation}
    uniformly in $p$.    Furthermore, we have the bounds
    \begin{equation}\label{eq:eta-sigma}
        \sum_{p\in\Lambda^*}\abs{\eta_p-\sigma_p} \leq C N^{3\kappa/2}, \qquad 
        \sum_{p\in\Lambda^*}p^4\abs{\eta_p-\sigma_p}^2 \leq C N^{7\kappa/2}
    \end{equation}
    From 
\begin{equation}\label{eq:Amomentumbounds}
\abs{A_{p,q}}
\leq 2N^{-1/2}\abs{\eta_p}\abs{\sigma_q}  
\end{equation}
we conclude that 
\begin{equation}\label{eq:AL2norm}
\sum_{p,q\in\Lambda^*}\abs{A_{p,q}}^2 \leq C N^{3\kappa-1} 
\end{equation}
and that 
\begin{equation}\label{eq:Appq}   \sum_{p\in\Lambda^*} \abs{A_{p,-p-q}} \leq C N^{2\kappa-1/2} \abs{q}^{-1}. \end{equation}

In position space, there is $C > 0$ and, for all $m\geq 2$, $C_m > 0$ such that 
\begin{equation}\label{eq:positionspacedecayA}
\begin{split}
\abs{\check{A}(x,y)}
\leq 
&C N^{-1/2} \min \Big( N , (\log N) \frac{N^\kappa}{\abs{x}}, C_m N \Big(\frac{\ell_\eta}{\abs{x}}\Big)^m\Big) \\
&\times\min\left(N, (\log N)^2 \frac{N^\kappa}{\abs{y}}, (\log N)^3 \frac{N^{\kappa/4}}{\abs{y}^{5/2}}, C_m N \left(\frac{\ell_\sigma}{\abs{y}}\right)^m\right)
\end{split}
\end{equation}
where $\ell_\eta=N^{-\kappa/2}$ and $\ell_\sigma=N^{-\kappa/2+\beta/4}$.
\end{lemma}

\begin{proof}
The bounds (\ref{eq:normseta}), \eqref{eq:Amomentumbounds}, \eqref{eq:AL2norm} and \eqref{eq:Appq} are easy to check. Eq. (\ref{eq:scatteringetacutoff}) follows from  \eqref{eq:computationscattering} (whose second line is shown to be $O(N^{2\kappa-1})$), because 
\[ \frac{N^\kappa}{N} \sum_{q \in \Lambda^*_+} |\eta_q - \eta_{\infty, q}| \leq C\frac{N^\kappa}{N} \sum_{|q| < 2 N^{\kappa/2}} \frac{N^\kappa}{q^2} \leq C N^{5\kappa/2 -1}. \]
To prove (\ref{eq:eta-sigma}), we observe that, from Lemma \ref{lemma:scattering}, 
\[ \begin{split} |\eta_p - \sigma_p| &\leq |\eta_p - \eta_{\infty,p}| + |\sigma_p - \sigma_{\infty,p}| + |\eta_{\infty,p} - \gamma_{\infty,p} \sigma_{\infty,p}| + |\sigma_{\infty,p}| | \gamma_{\infty,p} - 1| \\ &\leq C \min \Big( \frac{N^\kappa}{p^2} , \frac{N^{2\kappa}}{|p|^4} \Big). \end{split} \]

As for the position space bound (\ref{eq:positionspacedecayA}), we proceed similarly as in the proof of  Lemma~\ref{lemma:propertiesquadratickernels}. Analogously to \eqref{eq:trickpowersx} we find
\begin{equation}\label{eq:trickpowersxyA}
c^m c^l \abs{x_j}^m  \abs{y_k}^l\abs{\check A (x,y) }
\leq C  \int_0^{2\pi m} ds \int_0^{2\pi l} dt \sum_{p,q\in\Lambda^*}  \abs{(\partial_{p_j}^m\partial_{q_k}^l A (p- s e_j, q- t e_k)}
\end{equation} 
which again allows to relate position space decay to derivatives in momentum space.
For the derivatives of the kernel we find
\[ 
\begin{split}
\abs{\partial_{p_j}^m\partial_{q_k}^l A (p, q)}
&\leq C_{m,l} N^{-1/2+\kappa}
\\ &\times \sum_{\substack{m_i, l_i\geq 0\\ m_1+m_2+m_3=m\\ l_1+l_2 =l}}
\Big| 
\partial_{p_j}^{m_1} \tilde\chi(p) \partial_{p_j}^{m_2} \widehat{Vf}(\frac{p}{N^{1-\kappa}})
\partial_{p_j}^{m_3} \partial_{q_k}^{l_1}\frac{1}{p^2 + q^2 + (p+q)^2}
\partial_{q_k}^{l_2}\sigma (q) \Big| 
\end{split} \]
Recalling the definition (\ref{eq:defFm}) and the bound (\ref{eq:partialmsi}) for the derivatives of $\sigma (q)$, we obtain 
\[ \begin{split}
\abs{\partial_{p_j}^m\partial_{q_k}^l A (p, q)}
&\leq C_{m,l} N^{-1/2+\kappa} \\ &\times 
\sum_{\substack{m_i\geq 0\\ m_1+m_2+m_3=m}}
\abs{p}^{-(m_1+ m_3 +2)} \mathbbm{1}(\abs{p}\geq N^{\kappa/2})
N^{-(1-\kappa) m_2} F_m (p/N^{1-\kappa}) \\
&\times\sum_{\substack{l_i\geq 0\\ l_1+l_2 =l}}
\abs{q}^{- l_1}
(\log N)^{l_2} 
    \min \Big( \frac{N^{\kappa/4}}{\abs{q}^{1/2}}, \frac{N^\kappa}{\abs{q}^{2}} \Big)
    F_{l_2}(q/N^{1-\kappa}) \\ &\hspace{4cm} \times 
    \sum_{l_3 =0}^{l_2} \abs{q}^{-l_3} N^{-(1-\kappa) (l_2-l_3)}
    \mathbbm{1}(\abs{q}\geq N^{\kappa/2-\beta/4}). 
 \end{split} \]
Setting $n = m_1+m_3$, $i= l_1+l_3$ (and absorbing appropriate combinatorial factors in the constant $C_{m,l}$), we arrive at
 \[    \begin{split} 
\big| \partial_{p_j}^m &\partial_{q_k}^l A (p, q) \big| \\ \leq &C_{m ,l} N^{-1/2}
\frac{N^\kappa}{\abs{p}^2} F_m(p/N^{1-\kappa})
\mathbbm{1}(\abs{p}\geq N^{\kappa/2}) \sum_{n=0}^m
\abs{p}^{-n} 
 N^{-(1-\kappa) (m-n)}
\\
&\times (\log N)^l 
\min\Big( \frac{N^{\kappa/4}}{\abs{q}^{1/2}}, \frac{N^\kappa}{\abs{q}^{2}} \Big)
F_{l}(q/N^{1-\kappa})
    \mathbbm{1}(\abs{q}\geq N^{\kappa/2-\beta/4})
 \sum_{i=0}^l
    \abs{q}^{-i} N^{-(1-\kappa) (l-i)}.
    \end{split}
\]
Since the r.h.s. is a product of a function of $p$ with a function of $q$, we can apply this bound to (\ref{eq:trickpowersxyA}) and proceed exactly as in the proof of Lemma \ref{lemma:propertiesquadratickernels}, decoupling the variables $x$ and $y$.  
\end{proof}

Similarly as discussed before Lemma \ref{lemma:sumlocalizedsigma} for $\check{\sigma}$, it will also be useful to decompose $\check A$ into a sum of functions $\check{A}_{u,v}$ supported in small boxes around the points $u,v \in \Lambda_B$ (this will allow us to take advantage of the cutoff $\Theta$ on the local number of particles). 
\begin{lemma}\label{lemma:sumlocalizedA}
Recall that $\ell_B = N^{-\kappa/2 + \beta/12}$, $\Lambda_B = \Lambda \cap \ell_B \mathbb{Z}^3$ and that, for $u\in \Lambda_B$, we defined $B_u = \{x\in\Lambda, \inf_{v\in \Lambda_B} \abs{x-v} =\abs{x-u} \}$. We denote by $\check A_{u,v} = \check A (\mathbbm{1}_{B_u}\otimes \mathbbm{1}_{B_v})$ the restriction of $\check{A}$ to $B_u \times B_v$, so that $\check A = \sum_{u, v \in\Lambda_B} \check A_{u,v}$. Then, for every $\epsilon > 0$, there exists $C > 0$ such that 
    \begin{equation}
        \sum_{u,v \in\Lambda_B} \norm{\check A_{u,v}}_2
        \leq C N^{(3\kappa-1)/2 +\epsilon}.
    \end{equation}
\end{lemma}
\begin{proof}
    We proceed as in the proof of Lemma \ref{lemma:sumlocalizedsigma}, decomposing 
   \[
        \begin{split}
            \sum_{u\in\Lambda_B} \norm{\check A_{u,v}}_2
            = \; &\norm{\check A_{0,0}}_2
            + \sum_{0\neq v\in\Lambda_B, \abs{v}\leq \ell_\sigma N^\delta} \norm{\check{A}_{0,v}}_2
            + \sum_{v\in\Lambda_B, \abs{v}> \ell_\sigma N^\delta} \norm{\check{A}_{0,v}}_2\\
          &+\sum_{0\neq u\in\Lambda_B} \norm{\check A_{u,0}}_2
            + \sum_{0\neq u,v \in\Lambda_B,
             \abs{v}\leq \ell_\sigma N^\delta} \norm{\check{A}_{u,v}}_2
            + \sum_{0\neq u, v\in\Lambda_B, \abs{v}> \ell_\sigma N^\delta} \norm{\check{A}_{u,v}}_2.
            \end{split} \]
  We estimate $\| \check{A}_{0,0} \|_2 \leq \| \check{A} \|_2 \leq C N^{(3\kappa-1)/2}$ by (\ref{eq:AL2norm}). For the other terms, we use the pointwise bound (\ref{eq:positionspacedecayA}). To estimate $\| \check{A}_{0,v} \|_2$, for $|v| \leq \ell_\sigma N^\delta$, we divide the integral over $x \in B_0$ in the two regions $|x| \leq \ell_\eta N^\delta$ (where we use the bound for $\check{A} (x,y)$ proportional to $N^\kappa / |x|$) and $|x| > \ell_\eta N^\delta$ (where we use the arbitrarily fast decay $(\ell_\eta / |x|)^m$). We find 
  \[ \begin{split}  \| \check{A}_{0,v} \|_2 &\leq C (\log N)^4 N^{-1/2 + 5\kappa/4} |v|^{-5/2} \ell_B^{3/2} \Big[ \int_{|x| \leq \ell_\eta N^\delta} |x|^{-2} dx +  N^{1-\delta m} \ell_B^3 \Big]^{1/2}  \\ &\leq  C (\log N)^4 N^{-1/2 +5\kappa/4} |v|^{-5/2} \ell_B^{3/2} \big[ \ell_\eta^{1/2} N^{\delta /2} + N^{(1-\delta m)/2} \ell_B^{3/2} \big] 
 \end{split}  \]
 for all $|v| \leq \ell_\sigma N^\delta$. Hence
 \[ \sum_{0\neq v \in \Lambda_B , |v| \leq \ell_\sigma N^\delta} \| \check{A}_{0,v} \|_2 \leq  C (\log N)^4 N^{-1/2 +5\kappa/4} \ell_\sigma^{1/2} \ell_B^{-3/2}  \ell_\eta^{1/2} N^{\delta} \leq C N^{(3\kappa-1)/2+ 2\delta}. \]
 For $|v| > \ell_\sigma N^\delta$ and also for all terms with $u \not = 0$, we can use the arbitrary fast decay of $\check{A} (x,y)$ (in $y$, for $|v| > \ell_\sigma N^\delta$, and in $x$, if $u \not = 0$).
\end{proof}

To describe commutators involving the cubic operator (\ref{eq:A-def}), it is useful to introduce the kernels 
\begin{equation} \label{eq:A3A6} \begin{split} \check{A}^{(3)} (x,y) &= \check{A} (x,y) +   \check{A} (y,x) +  \check{A} (-y , x-y) \\  \check{A}^{(6)} (x,y) &=  \check{A}^{(3)} (x,y) +  \check{A}^{(3)} (-x, -x+y). \end{split} \end{equation} 
It is sometimes useful to use the notation $\check{A}^{(1)} = \check{A}$. Observe that, by definition, 
\begin{equation}\label{eq:symmAxy} \check{A}^{(6)} (x,y) = \check{A}^{(6)} (y,x) = \check{A}^{(6)} (x-y, -y) = \check{A}^{(6)} (y-x, -x).  \end{equation}  

In Fourier space, we have 
\begin{equation}\label{eq:Ajpq} \begin{split} 
A^{(3)}_{p,q} &= A_{p,q} + A_{q,p} + A_{-p-q,p} , \\ A^{(6)}_{p,q} &= A^{(3)}_{p,q} + A^{(3)}_{-p-q,q} =  A_{p,q} + A_{q,p} + A_{-p-q,p} + A_{-p-q,q} + A_{q,-p-q} + A_{p,-p-q}.
\end{split} \end{equation} 
From (\ref{eq:symmAxy}), we find 
\begin{equation}\label{eq:symmApq} A^{(6)}_{p,q} = A^{(6)}_{q,p} = A^{(6)}_{-p-q, q} = A^{(6)}_{p, -p-q}.  \end{equation} 
Additionally, from (\ref{eq:cubicfunctionmomentumspace}), we also have $A^{(j)}_{-p, -q} = A^{(j)}_{p,q}$, for $j=1,3,6$.

We collect some properties of these kernels in the next lemma.
\begin{lemma} \label{lm:A3A6}
The Fourier coefficients $A^{(j)}_{p,q}$ satisfy 
\[ \sum_{p,q \in \Lambda^*} |A^{(j)}_{p,q}|^2 \leq C N^{3\kappa-1} \]
for $j=1,3,6$. Moreover 
\begin{equation} \label{eq:sumqA}
\begin{split}  
\sum_{p \in \Lambda^*} |A_{p,q}^{(j)}|^2 &\leq C N^{3\kappa/2 - 1} (1 + \sigma_q^2), \qquad 
 \sum_{q \in \Lambda^*} |A_{p,q}^{(j)}|^2 \leq C N^{3\kappa/2 - 1} (1 + \sigma_p^2) \\ 
 \sum_{p \in \Lambda^*} |A_{p,q}^{(j)}| &\leq C N^{1/2} (1 + |\sigma_q|), \quad \qquad \;
 \sum_{q \in \Lambda^*} |A_{p,q}^{(j)}| \leq C N^{1/2} (1 + |\sigma_p|)
 \end{split} \end{equation}
for $j=1,3,6$. Hence 
\begin{equation}\label{eq:sigmaA} \sum_{p,q \in \Lambda^*} (\sigma_p^2 + \sigma_q^2) |A^{(j)}_{p,q}|^2 \leq C N^{3\kappa-1} \end{equation} 
for $j=1,3,6$. 
 
As in Lemma \ref{lemma:sumlocalizedA}, we define $\check{A}^{(j)}_{u,v} = \check{A}^{(j)} \, (\mathbbm{1}_{B_u} \otimes \mathbbm{1}_{B_v})$, for $u,v \in \Lambda_B = \Lambda \cap \ell_B \bZ^3$ and for $j=1,3,6$, so that $\check{A}^{(j)} = \sum_{u,v \in \Lambda_B} \check{A}^{(j)}_{u,v}$. Then, we have  
\begin{equation}\label{eq:locAj} \sum_{u,v \in \Lambda_B} \| \check{A}^{(j)}_{u,v} \|_2 \leq C N^{(3\kappa-1)/2+ \epsilon} \end{equation} 
for any $\epsilon > 0$ and for $j=1,3,6$. 

It is sometimes also useful to take the supremum in one of the two variables of $\check{A}^{(j)}$. 
For $j=1,3,6$, we define $D^{(j)} (x) = \sup_s |\check{A}^{(j)} (s,x)|$, $\tilde{D}^{(j)} (x) = \sup_s |\check{A}^{(j)} (x,s)|$. Then, we have 
\[ \| \tilde{D}^{(j)} \|_2, \| D^{(j)} \|_2 \leq C N^{3\kappa/4 + 1/2+ \epsilon} \]
for any $\epsilon > 0$ and for $j=1,3,6$. Setting $D_u^{(j)} = D^{(j)} \, \mathbbm{1}_{B_u}$ for $u \in \Lambda_B$, we find  
\[ \sum_{u \in \Lambda^*} \| D^{(j)}_u \|_2 \leq C N^{3\kappa/4  + 1/2 +\epsilon} \]
for any $\epsilon > 0$. 
\end{lemma} 
\begin{proof} 
The estimates in Fourier space follow from $|A_{p,q}| \leq C N^{-1/2} |\eta_p| |\sigma_q|$ and from the estimates for $\eta_p, \sigma_q$ in Lemma \ref{lemma:propertiesquadratickernels} and Lemma \ref{lemma:propertiesA} (recalling also that $\eta_p= 0$, if $|p| < N^{\kappa/2}$, and that $\sigma_q = 0$ if $|q| \leq N^{\kappa/2 -\beta/4}$). To show (\ref{eq:sigmaA}), we use (\ref{eq:sigmapowers}). 

The bounds (\ref{eq:locAj}) can be proven proceeding similarly as in the proof of Lemma \ref{lemma:sumlocalizedA} (it is easy to check that the proof of Lemma \ref{lemma:sumlocalizedA} also extends to terms like $\check{A} (-y, x-y) \mathbbm{1} (x \in B_u) \mathbbm{1} (y \in B_v)$). Finally, the estimates for $D^{(j)}, \tilde{D}^{(j)}$ and their localized versions follow from the pointwise bound 
\begin{equation}\label{eq:point-Dj} \begin{split} 
|D^{(j)} (x)|, |\tilde{D}^{(j)} (x)| \leq \; & C N^{1/2} \min \big(N, (\log N)^2 \frac{N^\kappa}{|x|}, (\log N)^3 \frac{N^{\kappa/4}}{|x|^{5/2}}, N (\ell_\sigma / |x|)^m \big) ,\end{split} \end{equation} 
valid for $j=1,3,6$, which can be proven with (\ref{eq:positionspacedecayA}). To show decay of terms like $\check{A} (-y,x-y)$ in the $x$-direction, we use that either $|y| > |x|/2$ or $|x-y| > |x|/2$ is always true. 
\end{proof}

The commutator of the cubic transformation (\ref{eq:A-def}) with the kinetic energy operator is again a cubic operator, with the modified coefficients \[ A_\cK (p,q) = -N^{-1/2} N^\kappa \widehat{Vf} (p/N^{1-\kappa}) \tilde{\chi} (p) \sigma (q).\] To handle this term, it will be useful to remove the cutoff $\tilde{\chi}$. To this end, we will need to control the expectation of the cubic operator, with coefficients 
\begin{equation}\label{eq:kineticlowmomenta}
A_{\cK}^\text{low} (p,q) = -N^{-1/2} (1-\tilde\chi(p)) N^\kappa \widehat{Vf}(\frac{p}{N^{1-\kappa}})\sigma (q).
\end{equation}
satisfying $A_{\cK}^\text{low} (p,q) = 0$ if $|p| > 2 N^{\kappa/2}$. In the next lemma, we collect some bounds on (\ref{eq:kineticlowmomenta}), which will be used below, in the proof of Lemma \ref{lm:Kxi}.  
\begin{lemma} \label{lm:Klow} 
    We have
    \begin{equation}
        \sum_{p,q\in \Lambda^*} \abs{A_{\cK}^\text{low}(p,q)}^2 \leq C N^{5\kappa-1}.
    \end{equation}
    Let $\check A_{\cK}^\text{low}(x,y) = \sum_{p,q\in \Lambda^*_+} A_{\cK}^\text{low}(p,q) e^{ip\cdot x}e^{iq\cdot y}$. For every $m \in \bN$, we find $C_m > 0$ such that 
    \begin{equation}\label{eq:checkAlow}
        \abs{\check A_{\cK}^\text{low}(x,y)} \leq C_m N^{5\kappa/2 - 1/2}
        \left(\frac{\ell_\eta}{\abs{x}}\right)^m
        \abs{\check{\sigma}(y)}.
    \end{equation}
    Recall $\ell_B = N^{-\kappa/2 + \beta/12}$, $\Lambda_B = \Lambda \cap \ell_B \mathbb{Z}^3$.  For $u\in \Lambda_B$, let $B_u = \{x\in\Lambda, \inf_{v\in \Lambda_B} \abs{x-v} =\abs{x-u} \}$ and denote by $\check A_{\cK,u,v}^\text{low} = \check A_{\cK}^\text{low} (\mathbbm{1}_{B_u}\otimes \mathbbm{1}_{B_v})$ the restriction of $\check{A}_{\cK}^\text{low}$ to $B_u \times B_v$.  Then, for every $\epsilon > 0$ there is $C > 0$ such that     
    \begin{equation}\label{eq:AKlowuv}
        \sum_{u,v \in\Lambda_B} \norm{\check A_{\cK,u,v}^\text{low}}_2
        \leq C N^{(5\kappa-1)/2 +\epsilon}.
    \end{equation} 
\end{lemma}
\begin{proof}
From (\ref{eq:sigp}), we find 
\[ \sum_{p,q \in \Lambda^*} |A_{\cK}^\text{low} (p,q)|^2 \leq C N^{-1+2\kappa} \sum_{|p| \leq 2 N^{\kappa/2}, q\in \Lambda^*} |\sigma_q|^2 \leq  C N^{5\kappa-1}. \]
The bound (\ref{eq:checkAlow}) follows by noticing that $\check{A}_{\cK}^\text{low} (x,y)$ is the product of a function of $x$ (the Fourier series with coefficients $N^{-1/2+\kappa} (1-\tilde{\chi} (p)) \widehat{Vf} (p/N^{1-\kappa})$) with $\check{\sigma} (y)$ and by estimating 
\[  \big| \partial^m_{p_j} \big[ (1-\tilde{\chi} (p)) \widehat{Vf} (p/N^{1-\kappa}) \big] \big|\leq C_m N^{-m\kappa/2}  F_m (p/N^{1-\kappa})     \mathbbm{1} (|p| \leq C N^{\kappa/2}) \]
(where we used the fact that $N^{-1+\kappa} \ll N^{-\kappa/2}$). Thus, proceeding as in (\ref{eq:trickpowersx}) and recalling $\| F_m \|_\infty \leq C_m$ from Lemma \ref{lm:Fm}, we find 
\[ \begin{split} c^m |x_j|^m |\check{A}_{\cK}^\text{low} (x,y)| &\leq C_m N^{-1/2 + \kappa}  \int_0^{2\pi m} ds \sum_{p \in \Lambda^*} \big|\partial_{p_j}^m (1-\tilde{\chi} (p)) \widehat{Vf} (p/N^{1-\kappa}) \big| |\check{\sigma} (y)| \\ &\leq C_m N^{5\kappa/2-1/2} N^{-m\kappa/2} |\check{\sigma} (y)| \end{split} \] 
which implies (\ref{eq:checkAlow}). With (\ref{eq:checkAlow}), the bound (\ref{eq:AKlowuv}) follows, proceeding similarly as in the proof of Lemma \ref{lemma:sumlocalizedA}. The main contributions to the sum are $\| \check{A}_{\cK,0,0}^\text{low} \|_2 \leq \| \check{A}_{\cK}^\text{low} \|_2 \leq C N^{5\kappa-1}$ and the sum over $|v| \leq \ell_\sigma N^\delta$ of $\| \check{A}_{\cK,0,v}^\text{low} \|_2$. Here, we integrate separately over $|x| \leq \ell_\eta N^\delta$ (where we estimate $| \check{A}_{\cK}^\text{low} (x,y)| \leq C N^{(5\kappa-1)/2} |\check{\sigma} (y)|$) and over $|x| > \ell_\eta N^\delta$ (where we bound $| \check{A}_{\cK}^\text{low} (x,y)| \leq C N^{(5\kappa-1)/2-\delta m} |\check{\sigma} (y)|$). Using (\ref{eq:boundsigmax5/2}) to control $|\check{\sigma} (y)|$, we obtain 
\[ \begin{split}  \sum_{v \in \Lambda_B , |v| \leq \ell_\sigma N^\delta} &\|  \check{A}_{\cK,0,v}^\text{low} \|_2 \\ & \leq C (\log N)^3 N^{(5\kappa-1)/2} N^{\kappa/4} \ell_B^{3/2}  \left[ \ell_\eta^3 N^{3\delta} + N^{-\delta m} \ell_B^3 \right]^{1/2}  \sum_{v\in \Lambda_B , |v| \leq \ell_\sigma N^\delta}  |v|^{-5/2}  \\ &\leq C (\log N)^3 N^{(5\kappa-1)/2} N^{\kappa/4} \ell_B^{-3/2}  \ell_\eta^{3/2} \ell_\sigma^{1/2} N^{2\delta} \leq C N^{(5\kappa-1)/2+\epsilon}  \end{split} \]
for an arbitrarily small $\epsilon >0$ (where we first chose $\delta > 0$ small enough and then $m \in \bN$ sufficiently large). All other contributions can be bounded using the fast decay for $|x| > \ell_\eta N^\delta$ or for $|y| > \ell_\sigma N^\delta$. 
\end{proof}

\subsection{The Local Particle Number Cutoff}
\label{sec:cutoff}

We now define the cutoff $\Theta$ to be inserted in the cubic transformation appearing in (\ref{eq:psi0}) to control the local density of the excitations. We consider the length scale $\ell_{B} = N^{-\kappa/2+\beta/12}$, with $\beta= 2-3\kappa$ (the choice of $\ell_B$ is explained in the remark after Lemma~\ref{lemma:sumlocalizedsigma}). Through $\Theta$, we will require that balls of radius $\ell_B$ only contain $N^\epsilon$ particles, for a small $\epsilon > 0$. Note that the typical number of cubic excitations in a ball with radius $\ell_B$ is of the order $\ell_B^3 N^{3\kappa-1} \simeq N^{-\beta/4} \ll 1$ (we will show that the total number of cubic excitations is of the order $\| A \|_2^2 \simeq N^{3\kappa-1}$). For this reason, $\Theta$ will not substantially change the energy of our trial state.  


Pick a radial function $\chi_{\ell_B} \in C_c^\infty(\mathbb{R}^3)$ such that $0\leq \chi_{\ell_B}\leq 1$, $\chi_{\ell_B}(x)=1$ if $\abs{x}\leq \ell_B$ and $\chi_{\ell_B}(x)=0$ if $\abs{x}\geq 2\ell_B$. Define
\begin{equation}
    \cN_{\chi_{\ell_B}} (w) \coloneqq \int dv \, \chi_{\ell_B}(w-v) a_v^* a_v.
\end{equation}
For $\ell > 0$, we also define the operator 
\begin{equation}
    \cN_{\ell} (w) \coloneqq \int dv \mathbbm{1}(\abs{w-v}\leq \ell) \, a_v^* a_v 
\end{equation}
counting the number of particles in a ball of radius $\ell$ around $w$. We think of $\cN_{\chi_{\ell_B}} (w)$ as a smoothed out version of $\cN_{\ell_B} (w)$. For $n \in \bN$, we set 
\begin{equation} \label{eq:def-Theta} 
    \cL_n \coloneqq \int dw \, (\cN_{\chi_{\ell_B}} (w) +1)^n.
\end{equation}
Choosing $\Theta\in C^\infty(\mathbb{R})$ monotonically decreasing with $0\leq \Theta\leq1$, $\Theta(x)=1$ if $x\leq 1$ and $\Theta(x)=0$ if $x\geq 2$, and for $\alpha > 0$ to be specified later on, we define the cutoff
$
    \Theta(\cL_n/N^\alpha) \, .
$

Note that
$
    [\cN_{\chi_{\ell_B}} (w), a_x^*] = \chi_{\ell_B}(w-x)a_x^*.
$
Hence, for every $k \in \bN$, there is a $C_k > 0$ such that
\begin{equation}
    C_k^{-1}\norm{\prod_{i=1}^k a_{x_i} (\cN_{\chi_{\ell_B}} (w) +1)\zeta}
    \leq \norm{ (\cN_{\chi_{\ell_B}} (w) +1)\prod_{i=1}^k a_{x_i}\zeta}
    \leq \norm{\prod_{i=1}^k a_{x_i} (\cN_{\chi_{\ell_B}} (w) +1)\zeta}.
\end{equation}

With the cutoff $\Theta(\cL_n/N^\alpha)$, we introduce the unitary cubic transformation \begin{equation}\label{eq:defcubictransform}
    \exp \Big(\int dxdydz \check{A}(x-y,x-z) a_x^* a_y^* a_z^* \Theta(\cL_n/N^\alpha)-\hc\Big)
    \eqqcolon \exp(A^*\Theta-\hc)
\end{equation}
with the kernel $\check{A}$ defined in \eqref{eq:defAxy}. 
\begin{lemma}\label{lemma:propertiesL}
There is a constant $C > 0$ such that, for every $\epsilon >0$ and $\alpha > 0$, and $T>0$, we have 
\begin{equation}\label{eq:localNbound}
    \cN_{k \ell_B} (w)\Theta(T^{-n}\cL_n/N^\alpha)
    \leq 
   C T  k^3 N^{\epsilon} \Theta(T^{-n} \cL_n/N^\alpha)
\end{equation}
and also 
\begin{equation}\label{eq:localNboundexample}
    \norm{\cN_{k\ell_B}^{1/2}(w)\prod_{j=1}^m a_{x_j}\Theta(T^{-n}\cL_n/N^\alpha)\zeta}^2
    \leq C T k^3 N^{\epsilon}
    \norm{\prod_{j=1}^m a_{x_j}\Theta(T^{-n}\cL_n/N^\alpha)\zeta}^2 
\end{equation}
for all $n$ large enough, $w \in \Lambda$, $k\geq 1$, $\zeta \in \cF (\Lambda )$ and $m \in \bN$. 
Moreover, for every $\alpha > 0$, we have
    \begin{equation}\label{eq:expectationL}
    \la \exp(A^*\Theta-\hc)\Omega, \cL_n \exp(A^*\Theta-\hc)\Omega\ra \leq C.
\end{equation}
for all $n$ large enough. As a consequence
\begin{equation}\label{eq:norm1-Theta}
    \norm{(1-\Theta(\cL_n/N^\alpha))\exp(A^*\Theta-\hc)\Omega} \leq C N^{-\alpha/2}.
\end{equation}

Also, for every $\alpha > 0$ we have the (non-optimal)  bound 
\begin{equation}\label{eq:easyboundNm}
    \la \exp(s (A^*\Theta-\Theta A)) \zeta , \cN^m \exp (s (A^*\Theta-\Theta A)) \zeta \ra
    \leq C N^{m - m\beta/2} + \langle \zeta, \cN^m \zeta \rangle 
\end{equation}
for all $n$ large enough and all normalized $\zeta \in \cF (\Lambda )$ and $s\in[0,1]$. 

Finally, we observe that, for any $\nu > 0$, we have  
  \begin{equation}\label{eq:AThetaA}
    \la \exp(A^*\Theta-\hc)\Omega, A^* (1- \Theta (\cL_n / N^\alpha)) A \exp(A^*\Theta-\hc)\Omega\ra \leq C N^{-\nu}
\end{equation}
if $\alpha > 0$ and then $n \in \bN$ are chosen large enough. 
\end{lemma}

\textit{Remark.} As it is clear from their proof, the bounds (\ref{eq:norm1-Theta}) and (\ref{eq:AThetaA}) remain true if we replace $1-\Theta$ by a sharp cutoff, ie. 
\[ \begin{split}  \| \mathbbm{1}_{[k, \infty)} (\cL_n / N^\alpha) \exp (A^* \Theta -\hc ) \Omega \| &\leq C N^{-\alpha/2} \\
\la \exp (A^* \Theta -\hc ) \Omega, A^* \mathbbm{1}_{[k, \infty)} (\cL_n / N^\alpha) A \exp (A^* \Theta -\hc ) \Omega \ra &\leq N^{-\nu} \end{split} \]
for $k>0$.
In fact, the last bounds also holds, if we include factors of $(\cN+1)$, ie.  
\[ \la \exp (A^* \Theta -\hc ) \Omega, A^* (\cN+1)^m \mathbbm{1}_{[k, \infty)} (\cL_n / N^\alpha) A \exp (A^* \Theta -\hc ) \Omega \ra \leq N^{-\nu} \]
if $\alpha$ and $n$ are chosen large enough (depending on $m$).

\begin{proof}
We begin with \eqref{eq:localNbound}. For arbitrary $w$ and $y$ with $\abs{w-y}\leq \ell_B/2$ we have
\begin{equation*}
    (\cN_{\ell_B/2} (w)+1)^n
    \leq 
    (\cN_{\chi_{\ell_B}} (y) +1)^n.
\end{equation*}
Hence
\begin{equation*}
\begin{split} \cL_n   &= \int dy \, (\cN_{\chi_{\ell_B}} (y) +1)^n \geq \int_{|w-y| \leq \ell_B/2}  dy \, (\cN_{\chi_{\ell_B}} (y) +1)^n \\ &\geq \int_{|w-y| \leq \ell_B/2}  dy \, (\cN_{\ell_B/2} (w) +1)^n  \geq C \ell_B^3   (\cN_{\ell_B/2} (w) +1)^n.
\end{split}  
\end{equation*}
Taking the $n$-th root and writing $\cN_{\ell_B} (w)$ as the sum over the number of particles in finitely many balls of radius $\ell_B/2$, we obtain that, on the support of $\Theta(T^{-n}\cL_n/N^\alpha)$, 
\begin{equation*}
    \cN_{\ell_B} (w) 
    \leq C T\ell_B^{-3/n}N^{\alpha/n} 
    \leq C T N^{\epsilon}.
\end{equation*}
if $n$ is large enough. 
Decomposing balls of radius $k \ell_B$ into balls of radius $\ell_B$, we find 
\begin{equation*}
    \cN_{k\ell_B} (w) 
    \leq C T k^3 N^{\epsilon}.
\end{equation*}
if $n$ is large enough. Similarly, we can show \eqref{eq:localNboundexample}. In fact,
\begin{equation}
\begin{split}
    &\norm{\cN_{k \ell_B}^{1/2}(w)\prod_{j=1}^m a_{x_j}\Theta(T^{-n}\cL_n/N^\alpha)\zeta}^2\\
    &\leq 
    \norm{\cN_{k\ell_B}^{1/2}(w)\Theta\Big(T^{-n} \int dy \, \big(\cN_{\chi_{\ell_B}} (y) +1 +\sum_{j=1}^m \chi_{\ell_B}(y-x_m) \big)^n / N^\alpha \Big) \, \prod_{j=1}^m a_{x_j}\zeta}^2\\
    &\leq C T k^3 N^{\epsilon}
    \norm{\prod_{j=1}^m a_{x_j}\Theta(T^{-n}\cL_n/N^\alpha)\zeta}^2.
\end{split}
\end{equation}
%


To prove (\ref{eq:expectationL}), we set $\xi=\exp(A^*\Theta-\Theta A)\Omega$ and we observe that, splitting $\check{A}$ as we did in  Lemma \ref{lemma:sumlocalizedA},
\begin{equation*}
\begin{split}
    &\abs{\la\xi, [\cL_n, \int dx_1 dx_2 dx_3 \check{A}(x_1-x_2,x_1-x_3) a_{x_1}^* a_{x_2}^* a_{x_3}^* \Theta(\cL_n/N^\alpha)]\xi\ra}\\
    &\leq\int dx_1 dx_2 dx_3 dw \, | \check{A}(x_1-x_2,x_1-x_3)| \, \abs{\la\xi,[(\cN_{\chi_{\ell_B}} (w) +1)^n, a_{x_1}^* a_{x_2}^* a_{x_3}^*] \Theta(\cL_n/N^\alpha)\xi\ra}\\
    &\leq\sum_{k=0}^{n-1}\int dx_1 dx_2 dx_3 dw \, |\check{A}(x_1-x_2,x_1-x_3)| \\ &\hspace{.4cm} \times \big| \la\xi,(\cN_{\chi_{\ell_B}} (w) +1)^k 
    \sum_{i=1}^3\chi_{\ell_B} (x_i-w)  a_{x_1}^* a_{x_2}^* a_{x_3}^*
    (\cN_{\chi_{\ell_B}} (w) +1)^{n-k-1} \Theta(\cL_n/N^\alpha)\xi\ra \big| \\
    &\leq C
    \sum_{u,v\in\Lambda_B}
    \sum_{i=1}^3
    \int dx_1 dx_2 dx_3 dw \abs{\check{A}_{u,v}(x_1-x_2,x_1-x_3)}
    \chi_{\ell_B}(x_i-w)\\
    &\hspace{.4cm} \times\norm{a_{x_1} a_{x_2} a_{x_3} \Theta(4^{-n}\cL_n/N^\alpha)  (\cN_{\chi_{\ell_B}} (w) +1)^{(n-1)/2} \xi}
    \norm{(\cN_{\chi_{\ell_B}} (w) +1)^{(n-1)/2} \xi}.\\
\end{split}
\end{equation*}
Here we used the fact that $\Theta$ is monotonically decreasing
\begin{align*}
    \Theta(\cL_n/N^\alpha)\leq \Theta\Big(4^{-n}\int dw \, (\cN_{\chi_{\ell_B}} (w)+\sum_{i=1}^3\chi_{\ell_B} (x_i-w) +1)^n/N^\alpha\Big)
\end{align*}
and the pull-through formula
\begin{align*}
\int dw \,   (\cN_{\chi_{\ell_B}} (w)+\sum_{i=1}^3\chi_{\ell_B} (x_i-w) +1)^n a_{x_1} a_{x_2} a_{x_3}=a_{x_1} a_{x_2} a_{x_3}\cL_n.
\end{align*}
Cauchy-Schwarz yields
\[ \begin{split}
    &\abs{\la\xi, [\cL_n, \int dxdydz \check{A}(x_1-x_2,x_1-x_3) a_{x_1}^* a_{x_2}^* a_{x_3}^* \Theta(\cL_n/N^\alpha)]\xi\ra}\\
    &\leq C
    \sum_{u,v\in\Lambda_B}
    \sum_{i=1}^3
    \Big(\int dx_1 dx_2 dx_3 dw \, \chi_{\ell_B}(x_i-w)\mathbbm{1}(\abs{x_1-x_2-u},\abs{x_1-x_3-v}\leq \ell_B)\\
    &\hspace{3cm}\times\norm{a_{x_1} a_{x_2} a_{x_3} \Theta(4^{-n}\cL_n/N^\alpha)  (\cN_{\chi_{\ell_B}} (w) +1)^{\frac{(n-1)}{2}} \xi}^2 \Big)^{1/2}\\
    &\hspace{.3cm} \times  \Big(\int dx_1 dx_2 dx_3 dw \abs{\check{A}_{u,v}(x_1-x_2,x_1-x_3)}^2
    \chi_{\ell_B}(x_i-w)\norm{(\cN_{\chi_{\ell_B}} (w) +1)^{\frac{(n-1)}{2}} \xi}^2\Big)^{1/2}.
\end{split} \]    
Consider, for example, the term associated with $i=1$ (the other terms can be handled similarly). For every $\alpha, \epsilon > 0$ we can bound this contribution by 
\begin{equation}\label{eq:Lnbd} \begin{split} 
C &\sum_{u,v\in\Lambda_B} \norm{\check{A}_{u,v}}_2
    \Big(\int dx_1 dw \chi_{\ell_B}(x_1 - w)\\
    &\hspace{1cm}\times\norm{\cN_{\ell_B}^{1/2}(x_1 -u) \cN_{6\ell_B}^{1/2}(x_1- v) a_{x_1} \Theta(4^{-n}\cL_n/N^\alpha)  (\cN_{\chi_{\ell_B}} (w) +1)^{\frac{(n-1)}{2}} \xi}^2 \Big)^{1/2}\\
    &\hspace{4cm} \times \Big(\int dx_1 dw \, 
    \chi_{\ell_B}(x_1-w)\norm{(\cN_{\chi_{\ell_B}} (w) +1)^{\frac{(n-1)}{2}} \xi}^2\Big)^{1/2} \\
        &\leq C N^\epsilon  \ell_B^{3/2} 
    \sum_{u,v\in\Lambda_B} \norm{\check{A}_{u,v}}_2
    \Big(\int  dw \, \norm{(\cN_{\chi_{\ell_B}} (w) +1)^{\frac{n}{2}} \xi}^2 \Big)^{1/2} \\ &\hspace{4cm} \times 
    \Big(\int  dw \, \norm{(\cN_{\chi_{\ell_B}} (w) +1)^{\frac{(n-1)}{2}}\xi}^2\Big)^{1/2}\\
    &\leq C N^{- \beta/8 + \epsilon}  
    \la\xi, \cL_n\xi\ra
    \end{split} \end{equation} 
if $n$ is large enough. Here, we used (\ref{eq:localNboundexample}) twice, to control $\cN_{\ell_B}^{1/2}(x_1 -u)$ and $\cN_{\ell_B}^{1/2}(x_1- v)$, and we applied Lemma \ref{lemma:sumlocalizedA}. 
 Choosing $0 < \epsilon < \beta /8$, we conclude that 
 \[ \Big| \frac{d}{dt} \langle e^{t (A^* \Theta - \Theta A)} \Omega, \cL_n     e^{t (A^* \Theta - \Theta A)} \Omega\rangle \Big|  \leq C  \langle e^{t (A^* \Theta - \Theta A)} \Omega, \cL_n     e^{t (A^* \Theta - \Theta A)} \Omega\rangle  \]
 for every $t \in \bR$. 
Applying Gronwall's lemma in the interval $t \in [0;1]$, we obtain (\ref{eq:expectationL}). Eq. \eqref{eq:norm1-Theta} is an immediate consequence.

To show \eqref{eq:easyboundNm}, we first observe that
\[ \cN \leq C \ell_B^{-3} \int dw \, \cN_{\chi_{\ell_B}} (w) \leq C \ell_B^{-3} \cL^{1/n}_n. \]
Therefore, on the support of the cutoff $\Theta$, we have 
\begin{equation}\label{eq:particlenumbereasy}
    \cN^m \leq C \ell_B^{-3m} N^{m \eps} \leq C N^{m (3\kappa/2 - \beta/4 +\epsilon)} \end{equation} 
for any $\eps > 0$, if $n$ is large enough (depending on $\alpha, \epsilon$). Setting $\zeta_s = \exp (s (A^* \Theta - \Theta A)) \zeta$, we obtain 
\begin{equation}
\begin{split}
    &\Big| \la \exp (A^*\Theta-\hc)\zeta, \cN^m \exp(A^*\Theta-\hc)\Omega\ra - \langle \xi, \cN^m \zeta \rangle \Big| \\
    &\leq C \int_0^1 ds \, \Big| 
    \la \zeta_s, \big[\cN^m ,  A^* \big] \Theta (\cL_n / N^\alpha) \zeta_s \ra \Big| \\
    &\leq C \int_0^1 ds \, \int dx_1 dx_2 dx_3 \, \big| \check{A} (x_1 - x_2 , x_1 - x_3)  \big| \\ &\hspace{.5cm} \times  \| a_{x_1} a_{x_2} a_{x_3} \Theta (\cL_n / N^\alpha) (\cN + 1)^{(m-1)/2} \zeta_s \| \big|  \|  (\cN + 1)^{(m-1)/2} \Theta (\cL_n / N^\alpha) \zeta_s \|  \\
    &\leq C N^{(3\kappa-1)/2 + \epsilon} \int_0^1 ds \,  \|  (\cN + 1)^{(m-1)/2} \Theta (\cL_n / N^\alpha) \zeta_s \|  \|  (\cN + 1)^{m/2} \Theta (\cL_n / N^\alpha) \zeta_s \|
    \end{split} \end{equation} 
for any $\epsilon > 0$, if $n$ is large enough. Here we proceeded similarly as in (\ref{eq:Lnbd}), decomposing $\check{A}$ into a sum of $\check{A}_{u,v}$, for $u,v \in \Lambda_B$. With (\ref{eq:particlenumbereasy}), we conclude that 
\[ \begin{split} \la \exp (A^*\Theta-\hc) &\zeta, \cN^m \exp(A^*\Theta-\hc)\zeta\ra \\ &\leq \langle \zeta, \cN^m \zeta \rangle + C N^{(3\kappa-1)/2 + \epsilon} N^{(3\kappa/2-\beta/4 + \epsilon) (m-1/2)} \\ &\leq 
\langle \zeta, \cN^m \zeta \rangle + C N^{m (1 - 3\beta/4 + 2\epsilon)} \leq \langle \zeta, \cN^m \zeta \rangle + C N^{m (1-\beta/2)} \end{split} \]
if $n$ is large enough (depending on $\alpha > 0$). 

Finally, we prove (\ref{eq:AThetaA}). Since 
\[\int dw \, \big( \cN_{\chi_{\ell_B}} (w) + 1 \big)^n  a_x^* a_y^* a_z^* =   a_x^* a_y^* a_z^*  \int dw \, \big( \cN_{\chi_{\ell_B}} (w) +\chi_{\ell_B} (x) +\chi_{\ell_B} (y) +\chi_{\ell_B} (z)  + 1 \big)^n\]
we find 
\[ \begin{split} \langle e^{A^* \Theta - \Theta A} \Omega, A^* (1- &\Theta) A e^{A^* \Theta - \Theta A} \Omega \rangle \\= \; &\langle e^{A^* \Theta - \Theta A} \Omega, \mathbbm{1}_{[1,\infty)} (\cL_n / N^\alpha) A^* (1-\Theta) A e^{A^* \Theta - \Theta A} \Omega \rangle. \end{split}  \]
With $\| A^* (1-\Theta) A (\cN+1)^{-3} \| \leq C N^\tau$, for some $\tau > 0$, we conclude 
\[  \langle e^{A^* \Theta - \Theta A} \Omega, A^* (1- \Theta) A e^{A^* \Theta - \Theta A} \Omega \rangle \leq  C N^\tau \|  \mathbbm{1}_{[1,\infty)} (\cL_n / N^\alpha)  e^{A^* \Theta - \Theta A} \Omega \| \| \cN^3 e^{A^* \Theta - \Theta A} \Omega  \|. \]
From \eqref{eq:norm1-Theta} (the fact that here we have $\mathbbm{1}_{[1,\infty)}$ rather than $1-\Theta$ does not change the argument) and \eqref{eq:easyboundNm} we obtain 
\[ \langle e^{A^* \Theta - \Theta A} \Omega, A^* (1-\Theta) A e^{A^* \Theta - \Theta A} \Omega \rangle  \leq C N^{-\nu} \]
if $\alpha > 0$ and then $n \in \bN$ are large enough. 
\end{proof}

%

\subsection{The number of particles in the trial state} 

With the coefficients $\mu_p$ from (\ref{eq:defmu}), the coefficients $A_{p,q}$ from  (\ref{eq:cubicfunctionmomentumspace}) and the local particle number cutoff $\Theta = \Theta (\cL_n / N^\alpha)$ introduced in Section \ref{sec:cutoff}, we can define the quadratic operator $B$ and the cubic operator $A$ as in (\ref{eq:BT0}) and (\ref{eq:A-def}) and the trial state 
\begin{equation}\label{eq:psi1} \psi_N = W_{N_0} e^{B-B^*} e^{A^* \Theta - \Theta A} \Omega \end{equation} 
as in (\ref{eq:psi0}), with $W_{N_0}$ the Weyl operator (\ref{eq:weyl0}) generating the Bose-Einstein condensate. Here, we choose the condensate number of particles $N_0$ so that 
\[ N = N_0 + \sum_{p \in \Lambda_+^*} \sigma_p^2 \]
with $\sigma_p = \sinh \mu_p$ as defined in (\ref{eq:BT1}) (the norm $\| \sigma \|_2^2$ is the number of excitations generated by the Bogoliubov transformation). Since, from Lemma \ref{lemma:propertiesquadratickernels}, $\| \sigma \|_2^2 \leq C N^{3\kappa/2} \ll N$ for $\kappa < 2/3$, it is always possible to find such $N_0$. 

In the next lemma, we check that, with these definitions, our trial state has the correct number of particles. In the next section, we compute its energy. 
\begin{lemma}\label{lm:NN2}
Consider the trial state (\ref{eq:psi1}). Then, we have 
    \begin{equation}
        \la \Psi_N, \cN \Psi_N \ra \geq N, \quad \la \Psi_N, \cN^2 \Psi_N \ra \leq CN^2
    \end{equation}
 if $\alpha >0$ and then $n \in \bN$ are large enough.
\end{lemma}
\begin{proof}
 With (\ref{eq:weyl1}), we have
    \begin{equation*}
     \begin{split} 
        \la \Psi_N, \cN \Psi_N \ra
       & =
        \la e^{B^*-B} e^{A^*\Theta-\hc}\Omega, (\cN + N_0 + \sqrt{N_0} (a_0^* + a_0)) e^{B^*-B} e^{A^*\Theta-\hc}\Omega\ra  \\ &=
        \la e^{B^*-B} e^{A^*\Theta-\hc}\Omega, (\cN + N_0) e^{B^*-B} e^{A^*\Theta-\hc}\Omega\ra \end{split} 
        \end{equation*} 
because the expectation of $a_0, a_0^*$ is zero in the state $e^{B^*-B} e^{A^*\Theta-\hc}\Omega$. From (\ref{eq:BT1}), we find 
\begin{equation}\label{eq:eBN} e^{B-B^*} \cN e^{B^* - B} = \sum_{p \in \Lambda_+^*} (\gamma_p^2 + \sigma_p^2) a_p^* a_p + \sum_{p \in \Lambda^*_+} \gamma_p \sigma_p (a_p^* a_{-p}^* + a_p a_{-p}) + \sum_{p \in \Lambda_+^*} \sigma_p^2 .\end{equation}  
Since the expectation of $a_p^* a_{-p}^*$ vanishes in the state $e^{A^* \Theta- \hc} \Omega$, we immediately conclude that 
\[  \la \Psi_N, \cN \Psi_N \ra \geq N_0 + \sum_{p \in \Lambda^*_+} \sigma_p^2 = N. \]
On the other hand, recalling from Lemma \ref{lemma:propertiesquadratickernels} that $\| \sigma \|_\infty, \| \gamma \|_\infty \leq C N^{\beta /8}$, (\ref{eq:eBN}) implies that 
\[  e^{B-B^*} \cN e^{B^* - B} \leq C N^{\beta /4} \cN + C N \, . \] 
Similarly, 
\[ e^{B-B^*} \cN^2 e^{B^* - B} \leq C N^{\beta /2} \cN^2 + C N^2 .\]
Thus
 \[ \begin{split}
        \la \Psi_N, \cN^2 \Psi_N \ra
        &\leq 
        C \la  e^{B^*-B} e^{A^*\Theta-\hc}\Omega, 
        (\cN^2 + N^2_0) e^{B^*-B} e^{A^*\Theta-\hc}\Omega\ra \\
        &\leq C \la e^{A^*\Theta-\hc}\Omega, (N^{\beta/2} \cN^2 + N^2) e^{A^*\Theta-\hc}\Omega\ra \\
        &\leq C N^{\beta/2} N^{2 - \beta} + C N^2 \leq C N^2 
        \end{split} \]
            where in the last step we used the rough bound \eqref{eq:easyboundNm}.
\end{proof}

\section{Energy of the trial state}
\label{sec:energy} 

In this section, we are going to estimate the energy $\langle \psi_N, \cH_N \psi_N \rangle$ of the trial state $\psi_N$ defined in (\ref{eq:psi1}). Together with Lemma \ref{lm:NN2}, this will complete the proof of Theorem \ref{theorem:Fockspaceresult}. Throughout the section, we will use the notation $\xi = e^{A^* \Theta - \Theta A} \Omega$ for the action of the cubic transformation on the vacuum, so that $\psi_N = W_{N_0} e^{B^* - B} \xi$. 

\subsection{Action of Weyl and Bogoliubov transformations on $\cH_N$} 
\label{subsec:quadratic} 

In the next lemma, we compute the action of $W_{N_0}$ and of $e^{B^*-B}$ on the Hamiltonian 
\begin{equation}
    \cH_N = \cK + \cV_N
    = \sum_{p\in\Lambda^*} p^2 a_p^* a_p
    + \frac{N^\kappa}{2N} \sum_{p,q,r\in\Lambda^*} \hat{V}(r/N^{1-\kappa}) a_{p+r}^* a_{q-r}^* a_q a_p.
\end{equation} 
This will allow us to express the energy $\langle \psi_N, \cH_N \psi_N \rangle$ in terms of certain expectations in the state $\xi$, which will be estimated in the next subsections. 
\begin{lemma}\label{lemma:quadraticrenormalization}
     We have 
    \begin{equation}\label{eq:lmHN1}
    \begin{split}
        \la \Psi_N, &\cH_N \Psi_N \ra \\ 
        = \; &4\pi\aa N^{1+\kappa}\\
        &+ \frac{1}{2}\sum_{p\in\Lambda^*_+}
        \Big(
        \sqrt{\abs{p}^4+2p^2 N^\kappa\widehat{Vf}(\frac{p}{N^{1-\kappa}})}
        - p^2 - N^\kappa\widehat{Vf}(\frac{p}{N^{1-\kappa}})
        + \frac{N^{2\kappa}\widehat{Vf}(\frac{p}{N^{1-\kappa}})^2}{2\abs{p}^2}
        \Big)
        \\
        &- \frac{N^\kappa \norm{\sigma}_2^2}{N} \sum_{p\in\Lambda^*}
        \hat{V}(p/N^{1-\kappa}) \eta_{\infty,p}
        - N^{\kappa-1} \sum_{p,r\in\Lambda^*}
        \hat{V}(r/N^{1-\kappa})
        \eta_{\infty,p+r}\sigma_{p}^2\\
        &+ \la \xi, (\cK + \cC_N^* + \cC_N + \cV_N)\xi \ra
        + \la \xi, (\cQ + \tilde\cC_N^* + \tilde\cC_N + \tilde\cV_N)\xi \ra + O(N^{4\kappa-1} + N^{2\kappa})
    \end{split}
    \end{equation}
    where
    \begin{equation}\label{eq:CtCtV} \begin{split} 
        \cC_N^* &=
        \frac{N^\kappa N_0^{1/2}}{N} \sum_{p, r\in\Lambda^*} \hat{V}(r/N^{1-\kappa}) 
        \sigma_p 
        a_{p+r}^* a_{-r}^* a_{-p}^*, \\
        \tilde\cC_N^* &=
        \frac{N^\kappa N_0^{1/2}}{N} \sum_{p, r\in\Lambda^*} \hat{V}(r/N^{1-\kappa}) 
        ((\gamma_{p+r}\gamma_r-1)\sigma_p + \sigma_{p+r}\sigma_r\gamma_p)
        a_{p+r}^* a_{-r}^* a_{-p}^*, \\
        \tilde\cV_N
        &=\frac{N^\kappa}{2N} \sum_{p,q,r\in\Lambda^*} \hat{V}(r/N^{1-\kappa}) 
        (\gamma_{p+r} \gamma_{q-r}\gamma_{q}\gamma_{p}-1
        +\sigma_{p+r}\sigma_{q-r}\sigma_{q}\sigma_{p}
        +2\gamma_{p+r}\sigma_{q-r}\sigma_{q}\gamma_p
        ) \\ &\hspace{5cm} \times 
        a_{p+r}^* a_{q-r}^*
         a_{q}a_{p}
        \\
        &\hspace{.4cm} +\frac{N^\kappa}{N} \sum_{p,q,r\in\Lambda^*} \hat{V}(r/N^{1-\kappa}) 
        \gamma_{p+r}\sigma_{q-r}\gamma_{q}\sigma_{p}
        a_{p+r}^*a_{-p}^*  a_{-q+r} a_{q} 
\end{split}
\end{equation}
and $\cQ$ is a quadratic operator satisfying
\begin{equation} \label{eq:Qbd}
    \pm\cQ \leq C N^\kappa (\cN + \tilde \cN)
\end{equation}
with $\tilde \cN = \sum_{p\in\Lambda^*} \sigma_p^2 a_p^* a_p$.
\end{lemma}

\textit{Remark:} in the next subsection, we will show that the expectations of the operators $\cQ, \tilde\cC_N, \tilde\cV_N$ in the state $\xi = e^{A^* \Theta- \Theta A} \Omega$ are all negligible (they only contribute to order $N^{4\kappa-1}$ or lower). 

\begin{proof}
We proceed as in \cite[Section 3]{BCS}. With (\ref{eq:weyl1}), we find 
    \begin{equation}
        \la \Psi_N, \cH_N \Psi_N \ra
        = \sum_{n=0}^4 \la e^{B^*-B} \xi, \cW_n e^{B^*-B} \xi \ra
    \end{equation}
     where 
    \begin{equation}
        \begin{split}
            &\cW_0 = \frac{N^\kappa N_0^2}{2N} \hat{V}(0)\\
            &\cW_1 = \frac{N^\kappa N_0^{3/2}}{N} \hat{V}(0) a_0^*+\hc\\
            &\cW_2 = \cK 
                +\frac{N^\kappa N_0}{2N} \sum_{p\in\Lambda^*}
                \hat{V}(p/N^{1-\kappa}) \left( a_p^* a_{-p}^* + a_p a_{-p} + 2a_p^* a_{p} \right)
                +\frac{N^\kappa \hat{V}(0) N_0}{N} \sum_{p\in\Lambda^*}
                 a_p^* a_{p}\\
             &\cW_3 = \frac{N^\kappa N_0^{1/2}}{N} \sum_{p,r \in\Lambda^*} \hat{V}(r/N^{1-\kappa}) a_{p+r}^* a_{-r}^* a_p +\hc\\
            &\cW_4 = \cV_N.
        \end{split}
    \end{equation}
    Setting $\cG_n = e^{-B^* + B} \cW_n e^{B^*-B}$ for $n =0,\dots ,4$, we find 
      \begin{equation}
        \la \Psi_N, \cH_N \Psi_N \ra
        = \sum_{n=0}^4 \la \xi, \cG_n  \xi \ra.
    \end{equation}
Clearly, we have $\la \xi, \cG_0 \xi \ra = N^\kappa N_0^2 \hat{V}(0) / 2N$ and $ \la \xi, \cG_1 \xi \ra = 0$. Since $\langle \xi, a_p^* a_{-p}^* \xi \rangle = 0$ (because $A^*$ only creates triples), we obtain  
    \begin{equation}\label{eq:G2xi} 
    \begin{split}
        \la \xi, \cG_2 \xi \ra
        = \; & \la \xi, \cK \xi \ra
        + \sum_{p\in\Lambda^*} 2 p^2 \sigma_p^2 \la\xi, a_p^* a_p \xi \ra
        + \sum_{p\in\Lambda^*} p^2 \sigma_p^2\\
        &+ \frac{N^\kappa N_0}{N} \sum_{p\in\Lambda^*}
                \hat{V}(p/N^{1-\kappa}) (\gamma_p +\sigma_p)^2  \la\xi, a_p^* a_{p}\xi\ra
       \\ &\hspace{4cm} + \frac{N^\kappa N_0}{N} \sum_{p\in\Lambda^*}
                \hat{V}(p/N^{1-\kappa}) (\gamma_p\sigma_p + \sigma_p^2)\\
        &+ \frac{N^\kappa \hat{V}(0) N_0}{N} \sum_{p\in\Lambda^*}
                 (\gamma_p^2 + \sigma_p^2)\la\xi, a_p^* a_{p}\xi\ra
         + \frac{N^\kappa \hat{V}(0) N_0}{N} \| \sigma \|_2^2 \, .
    \end{split}
    \end{equation}
    Taking into account that the expectation of operators having the form $a^* a^* a$ and $a^* a a$ in the state $\xi$ is zero, we find 
    \begin{equation}\label{eq:G3xi} 
        \la\xi, \cG_3 \xi\ra = \frac{N^\kappa N_0^{1/2}}{N} \sum_{p, r\in\Lambda^*} \hat{V}(r/N^{1-\kappa}) 
        (\gamma_{p+r}\gamma_r\sigma_p + \sigma_{p+r}\sigma_r\gamma_p) 
        \la\xi, (a_{p+r}^* a_{-r}^* a_{-p}^* +\hc )\xi\ra.
    \end{equation}
Finally, since operators of the form $a^* a^* a^* a^*$, $a^* a^* a^* a$, $a^* a a a$, $a a a a$ have vanishing expectation in the state $\xi$, we obtain 
    \begin{equation}\label{eq:defG4}
    \begin{split}
        \la \xi, \cG_4 \xi \ra
        = \; &\frac{N^\kappa}{2N} \sum_{p,q,r\in\Lambda^*} \hat{V}(r/N^{1-\kappa}) 
        \big(\gamma_{p+r} \gamma_{q-r}\gamma_{q}\gamma_{p}
        +\sigma_{p+r}\sigma_{q-r}\sigma_{q}\sigma_{p} +2\gamma_{p+r}\sigma_{q-r}\sigma_{q}\gamma_p \big) \\ &\hspace{9cm} \times 
         \la \xi,  a_{p+r}^* a_{q-r}^* a_{q} a_{p}
        \xi\ra\\
        &+ \frac{N^\kappa}{N} \sum_{p,q,r\in\Lambda^*} \hat{V}(r/N^{1-\kappa}) 
        \gamma_{p+r}\sigma_{q-r}\gamma_{q}\sigma_{p}
         \la \xi, a_{p+r}^*a_{-p}^*  a_{-q+r} a_{q} 
        \xi \ra\\
        &+  \frac{N^\kappa}{N}\hat{V}(0) \norm{\sigma}_2^2\sum_{p\in\Lambda^*}  
        (\gamma_{p}^2 +\sigma_p^2)
         \la \xi,  a_{p}^* a_{p}
        \xi \ra\\
        &+  \frac{N^\kappa}{N} \sum_{p,r\in\Lambda^*} \hat{V}(r/N^{1-\kappa}) 
        \sigma_{p+r}^2 (\gamma_{p}^2 +\sigma_{p}^2)
         \la \xi,  a_{p}^*  a_{p}
        \xi \ra\\
        &+
        \frac{2N^\kappa}{N} \sum_{p,r\in\Lambda^*} \hat{V}(r/N^{1-\kappa})  
        \sigma_{p+r}\gamma_{p+r}
        \gamma_{p}\sigma_{p}  \la \xi,  a_{p}^*a_{p} 
        \xi \ra\\
        &+
         \frac{N^\kappa}{2N} \sum_{p,r\in\Lambda^*} \hat{V}(r/N^{1-\kappa})  
        (\sigma_{p+r}\gamma_{p+r}
        \gamma_{p}\sigma_{p}+
        \sigma_{p+r}^2\sigma_{p}^2)
        +
        \frac{N^\kappa}{2N} \hat{V}(0) 
        \norm{\sigma}_2^4.
    \end{split}
    \end{equation}
Going through the r.h.s. of (\ref{eq:G2xi}) and (\ref{eq:defG4}), we observe that all quadratic terms, with the exception of the kinetic energy operator $\cK$, can be bounded as in (\ref{eq:Qbd}) (here, we use the bounds from Lemma \ref{lemma:propertiesquadratickernels}, which imply, in particular, that $p^2 \sigma_p^2 \leq C N^\kappa$, for all $p \in \Lambda^*$). The cubic terms in (\ref{eq:G3xi}) correspond exactly to the sum of the two cubic terms defined in (\ref{eq:CtCtV}). Moreover, the quartic terms on the r.h.s. of (\ref{eq:defG4}) are clearly the sum of $\cV_N$ and $\tilde{\cV}_N$, with $\tilde{\cV}_N$ as defined in (\ref{eq:CtCtV}). To complete the proof of Lemma \ref{lemma:quadraticrenormalization}, we only have to check that the constant terms arising from $\cG_0, \cG_2, \cG_4$ match the r.h.s. of  (\ref{eq:lmHN1}). To this end we notice, first of all, that the contribution $\langle\xi,  \cG_0 \xi \rangle$ can be combined with the last term on the r.h.s. of (\ref{eq:G2xi}) and the last term on the r.h.s. of (\ref{eq:defG4}) to yield
\[ \frac{N^\kappa N_0^2}{2N} \hat{V} (0) + \frac{N^\kappa N_0}{N} \hat{V} (0) \| \sigma \|_2^2 + \frac{N^\kappa}{2N} \hat{V} (0) \| \sigma \|_2^4 = \frac{N^\kappa}{2N} \hat{V} (0) \big(N_0 + \| \sigma \|_2^2 \big)^2 = \frac{N^{1+\kappa}}{2} \hat{V} (0). \]
Next, we observe that the term 
\[ \Big| \frac{N^\kappa}{2N} \sum_{p,r} \hat{V} (r/ N^{1-\kappa}) \sigma_{p+r}^2 \sigma_p^2 \Big| \leq C N^{\kappa-1} \| \hat{V} \|_\infty \| \sigma \|_2^4 \leq C N^{4\kappa - 1} \]
appearing on the r.h.s. of (\ref{eq:defG4}) is negligible. The remaining constant term on the r.h.s. of (\ref{eq:defG4}) can be rewritten as 
\[ \begin{split} 
\frac{N^\kappa}{2N} \sum_{p,r} \hat{V} &(r/N^{1-\kappa}) \sigma_{p+r} \gamma_{p+r}\gamma_p \sigma_p \\ = \; & \frac{N^\kappa}{2N} \sum_{p,r} \hat{V} (r/N^{1-\kappa}) \sigma_{\infty, p+r} \gamma_{\infty, p+r}\gamma_{\infty,p} \sigma_{\infty,p} + O (N^{4\kappa-1}) \end{split} \]
because the contribution associated with low momenta $|p|, |p+r| \leq N^{\kappa/2 - \beta /4}$ is negligible. Using the bound for $\gamma_{\infty,p} \sigma_{\infty,p} - \eta_{\infty ,p}$ in Lemma \ref{lemma:scattering}, we obtain 
\[ \begin{split} \frac{N^\kappa}{2N} \sum_{p,r} \hat{V} (r/N^{1-\kappa}) \sigma_{p+r} &\gamma_{p+r}\gamma_p \sigma_p \\ = \; & \frac{N^\kappa}{2N} \sum_{p,r} \hat{V} (r/N^{1-\kappa}) \eta_{\infty,p+r} \eta_{\infty,p}  \\&+ \frac{N^\kappa}{N} \sum_{p,r} \hat{V} (r/N^{1-\kappa}) \eta_{\infty, p+r} \big( \gamma_{\infty,p} \sigma_{\infty,p} - \eta_{\infty,p} \big) + O (N^{4\kappa-1}) .\end{split} \]
As for the remaining constant terms in (\ref{eq:G2xi}), we write 
\[ \begin{split} &\frac{N^\kappa N_0}{N} \sum_p \hat{V} (p/N^{1-\kappa}) \gamma_p \sigma_p   \\ &=  \frac{N^\kappa N_0}{N} \sum_p \hat{V} (p/N^{1-\kappa}) \eta_{\infty,p} +  \frac{N^\kappa N_0}{N} \sum_p \hat{V} (p/N^{1-\kappa}) (\gamma_{\infty,p} \sigma_{\infty,p} - \eta_{\infty,p}) + O (N^{4\kappa-1})  \\
& = N^\kappa \sum_p \hat{V} (p/N^{1-\kappa}) \eta_{\infty,p} - \frac{N^\kappa \| \sigma \|_2^2}{N} \sum_p \hat{V} (p/N^{1-\kappa}) \eta_{\infty,p} \\ &\hspace{.4cm} +  N^\kappa \sum_p \hat{V} (p/N^{1-\kappa}) (\gamma_{\infty,p} \sigma_{\infty,p} - \eta_{\infty,p}) + O(N^{4\kappa-1}) \end{split} \]
%
and, similarly, 
\[ \begin{split} \frac{N^\kappa N_0}{N} &\sum_p \hat{V} (p/N^{1-\kappa}) \sigma_p^2  = N^\kappa \sum_p \hat{V} (p/N^{1-\kappa}) \sigma_{\infty,p}^2 + O (N^{4\kappa-1})
\end{split} \]
as well as
\[
\sum_p p^2 \sigma_p^2 = \sum_p p^2 \sigma_{\infty,p}^2 + O (N^{4\kappa-1}).
\]
Putting together all constant contributions, we find 
\begin{equation}\label{eq:const}
    \begin{split}
        &\frac{N^{1+\kappa}}{2} \hat{V} (0) + \sum_{p\in\Lambda^*} p^2 \sigma_{\infty,p}^2
        + N^\kappa\sum_{p\in\Lambda^*}
                \hat{V}(p/N^{1-\kappa}) \eta_{\infty,p} 
        - \frac{N^\kappa \norm{\sigma}_2^2}{N} \sum_{p\in\Lambda^*}
        \hat{V}(p/N^{1-\kappa}) \eta_{\infty,p}\\
        &+ N^\kappa \sum_{p\in\Lambda^*}
                \hat{V}(p/N^{1-\kappa}) (\gamma_{\infty,p}\sigma_{\infty,p}-\eta_{\infty,p})
        + N^\kappa \sum_{p\in\Lambda^*}
                \hat{V}(p/N^{1-\kappa})
                \sigma_{\infty,p}^2\\
        &+\frac{N^\kappa}{2N} \sum_{p,r\in\Lambda^*} \hat{V}(r/N^{1-\kappa})  
            \eta_{\infty,p+r}
            \eta_{\infty,p} \\              
            &+\frac{N^\kappa}{N} \sum_{p,r\in\Lambda^*} \hat{V}(r/N^{1-\kappa}) 
            \eta_{\infty,p+r}
            (\gamma_{\infty,p}\sigma_{\infty,p}-\eta_{\infty,p})
        + O(N^{4\kappa-1}).\\
    \end{split}
\end{equation}
With the scattering equation (\ref{eq:scatteringdiscrete}) and with (\ref{eq:etacorrection}), we combine 
\[ \begin{split} N^\kappa\sum_{p\in\Lambda^*}
                \hat{V}(p/N^{1-\kappa}) \eta_{\infty,p} &+\frac{N^\kappa}{2N} \sum_{p,r\in\Lambda^*} \hat{V}(r/N^{1-\kappa})  
            \eta_{\infty,p+r}
            \eta_{\infty,p} \\ = &- \sum_{p \in \Lambda^*} p^2 \eta_{\infty,p}^2 
       + \frac{N^\kappa}{2} \sum_{p \in \Lambda^*} \hat{V} (p/N^{1-\kappa}) \eta_{\infty,p} + O(N^{2\kappa})\\ 
       = &- \sum_{p \in \Lambda^*} p^2 \eta_{\infty,p}^2  + \frac{8\pi \frak{a} - \hat{V} (0)}{2} N^{1+\kappa} + O(N^{2\kappa})\end{split} \]
and, using also (\ref{eq:gammasigmainfty-etainftysum}),   
\[ \begin{split} 
 N^\kappa \sum_{p\in\Lambda^*}
                \hat{V}(p/N^{1-\kappa}) (\gamma_{\infty,p}\sigma_{\infty,p}-\eta_{\infty,p}) &+\frac{N^\kappa}{N} \sum_{p,r\in\Lambda^*} \hat{V}(r/N^{1-\kappa}) 
            \eta_{\infty,p+r}
            (\gamma_{\infty,p}\sigma_{\infty,p}-\eta_{\infty,p}) \\ &= -2 \sum_{p \in \Lambda^*} p^2 \eta_{\infty,p} ( \gamma_{\infty,p}\sigma_{\infty,p}-\eta_{\infty,p}) + O(N^{4\kappa-1}). \end{split} \]
Moreover, we write 
\[ \begin{split} N^\kappa \sum_{p\in \Lambda^*} \hat{V} (p/N^{1-\kappa}) \sigma_{\infty,p}^2 = &-2 \sum_{p\in \Lambda^*} p^2 \eta_{\infty,p} \sigma_{\infty,p}^2 \\ &- N^{\kappa-1} \sum_{p,r \in \Lambda^*} \hat{V} (r/N^{1-\kappa}) \eta_{\infty,p+r} \sigma_{p}^2 + O(N^{4\kappa-1}).\end{split} \]
Therefore, (\ref{eq:const}) can be rewritten as 
\[ \begin{split} 
&4\pi \frak{a} N^{1+\kappa} + \sum_{p\in\Lambda^*} p^2 (\sigma_{\infty,p}^2 + \eta_{\infty,p}^2 -2\eta_{\infty,p} \gamma_{\infty,p} \sigma_{\infty,p} -2 \eta_{\infty,p} \sigma_{\infty,p}^2)
       \\
        & - \frac{N^\kappa \norm{\sigma}_2^2}{N} \sum_{p\in\Lambda^*}
        \hat{V}(p/N^{1-\kappa}) \eta_{\infty,p} - N^{\kappa-1} \sum_{p,r\in \Lambda^*} \hat{V} (r/N^{1-\kappa}) \eta_{p+r} \sigma_{p}^2 + O(N^{4\kappa-1} + N^{2\kappa})
                \end{split} \]
and thus, inserting (\ref{eq:sigmainfty}), as
\[ \begin{split} 
&4\pi \frak{a} N^{1+\kappa}   + \frac{1}{2}\sum_{p\in\Lambda^*} p^2 
        \left(
        (\frac{1-2\eta_{\infty,p}}{\sqrt{1-4\eta_{\infty,p}}}-1) (1-2\eta_{\infty,p})
        + 2\eta_{\infty,p}^2
        - 4 \frac{\eta_{\infty,p}^2}{\sqrt{1-4\eta_{\infty,p}}}
        \right) \\
        &\hspace{.2cm}  - \frac{N^\kappa \norm{\sigma}_2^2}{N} \sum_{p\in\Lambda^*}
        \hat{V}(p/N^{1-\kappa}) \eta_{\infty,p} - N^{\kappa-1} \sum_{p,r\in \Lambda^*} \hat{V} (r/N^{1-\kappa}) \eta_{p+r} \sigma_{p}^2 + O(N^{4\kappa-1} + N^{2\kappa}) \\
        &= 4\pi \frak{a} N^{1+\kappa} + \frac{1}{2}\sum_{p\in\Lambda^*} p^2 
        \left(
        \sqrt{1-4\eta_{\infty,p}}
        - 1 + 2\eta_{\infty,p}
        + 2\eta_{\infty,p}^2
        \right) \\
          &\hspace{.2cm}  - \frac{N^\kappa \norm{\sigma}_2^2}{N} \sum_{p\in\Lambda^*}
        \hat{V}(p/N^{1-\kappa}) \eta_{\infty,p} - N^{\kappa-1} \sum_{p,r\in \Lambda^*} \hat{V} (r/N^{1-\kappa}) \eta_{p+r} \sigma_{p}^2 + O(N^{4\kappa-1} + N^{2\kappa}).
                \end{split} \]
With the definition (\ref{eq:defetainfty}) of $\eta_{\infty,p}$, we obtain exactly the constant term in (\ref{eq:lmHN1}).
\end{proof}

\subsection{Bounds for negligible cubic expectations} 
\label{subsec:cubic1} 

In this subsection, we estimate the expectations $\langle \xi, \cQ \xi \rangle, \langle \xi, ( \tilde{\cC}_N + \tilde{\cC}_N^* ) \xi \rangle , \langle \xi, \tilde{\cV}_N \xi \rangle$, appearing on the r.h.s. of (\ref{eq:lmHN1}), and we show that they are at most of order $N^{4\kappa-1}$ and therefore negligible for our purposes. We will make use of the following preliminary lemma to bound expectations of cubic operators.  
\begin{lemma}\label{lm:prelim1} 
Let
\begin{equation}\label{eq:cCK-def}  \cC_K  = \int dy_1 dy_2 dy_3 \, K (y_1-y_2,y_1-y_3) a_{y_1}a_{y_2}a_{y_3} + \hc \end{equation} 
with $K \in L^2 (\bR^3 \times \bR^3)$. We assume that $K (y,z) = \sum_{u,v \in \Lambda_B} K_{u,v} (y,z)$,  
where $K_{u,v} (y,z) = 0$, if $|y-u| > 10 \ell_B$ or $|z-v| > 10 \ell_B$. Then, for every $\epsilon > 0$ and $\nu > 0$, we have  
 \begin{equation}\label{eq:cCK-est} 
    \begin{split}
      \big|  \la &e^{t A^*\Theta-\hc}  \Omega, \cC_K e^{t A^*\Theta-\hc}\Omega\ra \big| 
       \\  &\leq |\langle \Omega, [\cC_K , A^* ] \Omega \rangle |  \\ &\hspace{.5cm} + C \ell_B^3  N^{(3\kappa-1)/2+\epsilon} \big( \sum_{u,v \in\Lambda_B} \norm{K_{u,v}}_2 \big)  
    \int_0^t ds \la e^{sA^*\Theta-\hc}\Omega, \cN e^{sA^*\Theta-\hc}\Omega\ra + C N^{-\nu} \end{split} 
    \end{equation}
   for all $t \in [0;1]$, if we choose $\alpha > 0$ and then $n$ large enough. Notice that, in (\ref{eq:cCK-est}), we can bound  \[ |\langle \Omega, [\cC_K , A^* ] \Omega \rangle | \leq C \| A \|_2  
\| K \|_2 \leq C N^{(3\kappa-1)/2} \| K \|_2. \] 
\end{lemma}

\begin{proof}
We take $t=1$ (the extension to $t \in [0;1]$ is straightforward). With Duhamel, we obtain 
    \begin{equation}\label{eq:CA*1}
    \begin{split}
        & \big| \la e^{A^*\Theta-\hc}\Omega, \cC_K e^{A^*\Theta-\hc}\Omega\ra \big| 
        \\ &\leq \int_0^1 ds \, \big| \la e^{sA^*\Theta-\hc}\Omega, [\cC_K , A^*\Theta]  
        e^{sA^*\Theta-\hc}\Omega\ra \big| \\
        &\leq \int_0^1 ds \, \big| \la e^{sA^*\Theta-\hc} \Omega,  ( \Theta [\cC_K, A^*] \Theta 
        +(1-\Theta )[\cC_K, A^*]\Theta 
        \\ &\hspace{4cm} + A^* \cC_K (\Theta -1)
        - A^* (\Theta-1) \cC_K ) e^{sA^*\Theta-\hc}\Omega\ra \big| \\ 
        &\leq \int_0^1 ds \, \big| \la e^{sA^*\Theta-\hc} \Omega,  \Theta [\cC_K, A^*] \Theta 
         e^{sA^*\Theta-\hc}\Omega\ra \big| + C N^{-\nu}
        \end{split} 
    \end{equation}
if $\alpha > 0$ and then $n \in \bN$ are chosen sufficiently large in \eqref{eq:defcubictransform}. Here we estimated all terms containing a factor $\Theta-1$ with Lemma \ref{lemma:propertiesL}, in particular (\ref{eq:norm1-Theta}), (\ref{eq:easyboundNm}) and (\ref{eq:AThetaA}), using also the bound $\| A^* (\cN+1)^{-3/2} \| \leq C \| A \|_2, \| \cC_K (\cN+1)^{-3/2} \| \leq  C \norm{K}_2$. 

We observe that
\begin{equation}\label{eq:commCKA} \begin{split}  [ \cC_K, A^*] = \int dx_1 dx_2 dx_3 &dy_1 dy_2 dy_3 \, K (y_1 - y_2, y_1 - y_3) \\ &\times \check{A} (x_1 - x_2 , x_1 - x_3) \big[ a_{y_1} a_{y_2} a_{y_3} , a_{x_1}^* a_{x_2}^* a_{x_3}^* \big] \end{split} \end{equation} 
where 
\begin{equation}\label{eq:comm-aaa} \begin{split} \big[ a_{y_1} a_{y_2} a_{y_3} , a_{x_1}^* a_{x_2}^* a_{x_3}^* \big]  = \; & \sum_{\pi \in S_3} \delta (y_1 - x_{\pi 1}) \delta (y_2 - x_{\pi 2}) \delta (y_3 - x_{\pi 3}) \\ &+ \sum_{i_1 < i_2}^3 \sum_{\pi \in S_3} \delta (y_{i_1} - x_{\pi i_1}) \delta (y_{i_2} - x_{\pi i_2}) \prod_{i_3 \not = i_1, i_2} a_{x_{\pi i_3}}^* a_{y_{i_3}}  \\ &+ \sum_{i,j=1}^3 \delta (y_i - x_j) \prod_{\ell \not = j} a^*_{x_\ell} \prod_{m \not = i} a_{y_m}  .\end{split} \end{equation}
We decompose the main term on the r.h.s. of (\ref{eq:CA*1}) in the sum  $T_4 + T_2 + T_0$, where $T_4$ are the contributions arising from the quartic terms on the r.h.s. of (\ref{eq:comm-aaa}), $T_2$ from the quadratic terms and $T_0 = \langle \Omega, [\cC_K , A^* ] \Omega \rangle$ from the constant terms. It is easy to check that $|T_0| \leq C \| K \|_2 \| A \|_2$. Let us focus on the term $T_2$. As an example, we consider the contribution $T_{2,0}$ associated with $i_1 = \pi i_1 = 1, i_2 = \pi i_2 = 2$ in (\ref{eq:comm-aaa}). Localizing the kernels $K, \check{A}$, we can bound
\[\begin{split}  |T_{2,0}| &\leq \sum_{u,v,\tilde{u} , \tilde{v} \in \Lambda_B}  \int_0^1 ds \int dx_1 dx_2 dx_3 dy_3  \, |K_{u,v} (x_1 - x_2 , x_1 - y_3)| |\check{A}_{\tilde{u}, \tilde{v}} (x_1 - x_2 , x_1 - x_3)| \\ &\hspace{7cm} \times  \big| \langle \Theta e^{s A^* - \text{h.c.}} \Omega, a_{x_3}^* a_{y_3} \Theta e^{s A^* \Theta- \text{h.c.}} \Omega \rangle \big|   \\
&\leq \sum_{u,v,\tilde{u} , \tilde{v} \in \Lambda_B} \int_0^1 ds  \Big( \int dx_1 dx_2 dx_3 dy_3 \, |K_{u,v} (x_1 - x_2 , x_1 - y_3)|^2 \\ &\hspace{2.5cm} \times  \mathbbm{1}(|x_1 - x_2 - \tilde{u}| , |x_1 - x_3 - \tilde{v}| \leq 10 \ell_B) \| a_{x_3} \Theta e^{sA^* \Theta - \hc} \Omega \|^2 \Big)^{1/2}  \\ &\hspace{.4cm}  \times \Big( \int dx_1 dx_2 dx_3 dy_3 \, |\check{A}_{\tilde{u},\tilde{v}} (x_1 - x_2 , x_1 - x_3)|^2 \\ &\hspace{2.5cm} \times \mathbbm{1}(|x_1 - x_2 - u| , |x_1 - y_3 - v| \leq 10 \ell_B) \| a_{y_3} \Theta e^{sA^* \Theta - \hc} \Omega \|^2 \Big)^{1/2} .
\end{split} \]
Using the characteristic functions to gain a volume factor $\ell_B^3$ (after switching to variables $x_2' = x_1 - x_2$, $y'_3= y_3 - x_1$, $x'_1 = x_1 - x_3$ and $x'_3 = x_3$ in the first integral on the r.h.s. and to variables $x'_2 = x_2 - x_1$, $x'_3 = x_3 - x_1$, $x'_1 = x_1 - y_3$, $y'_3 = y_3$ in the second) and applying Lemma \ref{lemma:sumlocalizedA}, we can estimate 
\[ |T_{2,0}| \leq C \ell_B^3 N^{(3\kappa-1)/2 + \epsilon} \sum_{u,v \in \Lambda_B} \| K_{u,v} \|_2
\int_0^1 ds \, \langle e^{sA^* \Theta - \hc} \Omega, \cN e^{sA^* \Theta - \hc} \Omega \rangle. \]
The other contributions to $T_2$ can be bounded similarly. Let us now consider $T_4$. Also here, we focus on an example, taking the term $T_{4,0}$ associated with $i = j = 1$ on the third line of (\ref{eq:comm-aaa}). We find  
\begin{equation}\label{eq:T4-bd} 
\begin{split}
|T_{4,0}|  &\leq \sum_{u,v, \tilde u, \tilde v\in\Lambda_B} 
    \int dx_1 dx_2 dx_3 dy_2 dy_3
    | K_{u,v} (x_1-y_2,x_1-y_3) | \, | \check{A}_{\tilde u, \tilde v}(x_1-x_2,x_1-x_3)| \\
    &\hspace{2cm}  \times\int_0^1 ds \, |\la \Theta e^{sA^*\Theta-\hc}\Omega, a^*_{x_2}a^*_{x_3} a_{y_{2}}a_{y_{3}}
     \Theta e^{sA^*\Theta-\hc}\Omega\ra | \\
     &\leq  \sum_{u,v, \tilde u, \tilde v\in\Lambda_B} \int_0^1 ds \Big( \int dx_1 dx_2 dx_3 dy_2 dy_3 \,    | K_{u,v} (x_1-y_2,x_1-y_3) |^2\\ &\hspace{2cm} \times \mathbbm{1}(|x_1 - x_2 - \tilde{u}| , |x_1 - x_3 - \tilde{v}| \leq 10 \ell_B)  \| a_{x_2} a_{x_3} \Theta e^{sA^*\Theta-\hc}\Omega \|^2 \Big)^{1/2} \\ &\hspace{.4cm} \times \Big( \int dx_1 dx_2 dx_3 dy_2 dy_3 \,    | \check{A}_{\tilde{u},\tilde{v}} (x_1-x_2,x_1-x_3) |^2 \\ &\hspace{2cm}  \times   \mathbbm{1}(|x_1 - y_2 - u| , |x_1 - y_3 - v| \leq 10 \ell_B) \| a_{y_2} a_{y_3} \Theta e^{sA^*\Theta-\hc}\Omega \|^2 \Big)^{1/2} .\end{split} \end{equation} 
We apply Lemma \ref{lemma:propertiesL} to estimate, with the notation $\xi_s = e^{sA^* \Theta - \hc} \Omega$,  
\[ \begin{split} \int  \mathbbm{1}(|x_1 - \tilde{v} - x_3| \leq 10 \ell_B) \| a_{x_3} a_{x_2} \Theta \xi_s \|^2  dx_3 = \| \cN_{10 \ell_B}^{1/2} (x_1 - \tilde{v}) a_{x_2} \Theta \xi_s \| \leq C N^\epsilon \| a_{x_2} \Theta \xi_s \| \\ 
 \int  \mathbbm{1}(|x_1 - v - y_3| \leq 10 \ell_B) \| a_{y_3}  a_{y_2} \Theta \xi_s \|^2  dy_3 = \| \cN_{10 \ell_B}^{1/2} (x_1 - v) a_{y_2} \Theta \xi_s \| \leq C N^\epsilon \| a_{y_2} \Theta \xi_s \| \end{split} \] 
if $n$ is large enough. With Lemma \ref{lemma:sumlocalizedA}, we conclude that 
\[ |T_{4,0}| \leq C \ell_B^3 N^{(3\kappa-1)/2+\epsilon} \sum_{u,v \in \Lambda_B} \| K_{u,v} \|_2 
\int_0^1 ds \, \langle e^{sA^* \Theta - \hc} \Omega, \cN e^{sA^* \Theta - \hc} \Omega \rangle. \]
Also the other contributions to $T_4$ can be handled similarly. 
\end{proof} 

An immediate application of Lemma \ref{lm:prelim1} is the following bound for the expectation of the number of excitations generated by the cubic transformation. 
\begin{lemma}\label{lm:boundN}
    We have
    \begin{equation}\label{eq:N-bd}
        \la e^{t A^*\Theta-\hc}\Omega, \cN e^{t A^*\Theta-\hc}\Omega\ra
        \leq C N^{3\kappa-1}
    \end{equation}
    for all $t \in [0;1]$, provided in \eqref{eq:defcubictransform} we choose $\alpha$ and then $n$ large enough.
\end{lemma}

\begin{proof}
With $[\cN, A^* ] = 3 A^*$, we find  
   \begin{equation*}
        \la e^{t A^*\Theta-\hc}\Omega, \cN e^{t A^*\Theta-\hc}\Omega\ra
        = 2 \text{Re }\int_0^t ds \la e^{sA^*\Theta-\hc}\Omega, (3A^* + 3A^*(\Theta-1)) e^{sA^*\Theta-\hc}\Omega\ra.
    \end{equation*}
The term proportional to $\Theta-1$ can be handled with Lemma \ref{lemma:propertiesL}; it can be made arbitrarily small, choosing $\alpha > 0$ and then $n \in \bN$ large enough in \eqref{eq:defcubictransform}. To estimate the expectation of $A^*$, we apply Lemma \ref{lm:prelim1}, with $K = A$. From Lemma \ref{lemma:propertiesA} and Lemma \ref{lemma:sumlocalizedA}, we obtain 
\begin{equation*}
\begin{split}
        \la e^{t A^*\Theta-\hc}&\Omega, \cN e^{t A^*\Theta-\hc}\Omega\ra \\
    &\leq C N^{3\kappa-1}
    + C N^{- \beta/4 + \epsilon} \int_0^t ds \la e^{sA^*\Theta-\hc}\Omega, \cN e^{sA^*\Theta-\hc}\Omega\ra\end{split}
\end{equation*}
for any $\epsilon > 0$, if $\alpha, n$ are large enough. Choosing $\epsilon > 0$ small enough, we can apply Gronwall's lemma to show (\ref{eq:N-bd}). 
\end{proof}
Inserting the bound (\ref{eq:N-bd}) back into (\ref{eq:cCK-est}), we obtain the following corollary.
\begin{cor} \label{cor:prelim-2} 
Suppose the operator $\cC_K$ is defined as in Lemma \ref{lm:prelim1}. For every $\epsilon > 0$, we have 
\begin{equation}   \begin{split}   \big| & \la e^{t A^*\Theta-\hc}  \Omega, \cC_K e^{t A^*\Theta-\hc}\Omega\ra \big| 
   \leq  C N^{(3\kappa-1)/2}  \Big( \norm{K}_2  + N^{-\beta/4 + \epsilon}  \sum_{u,v \in\Lambda_B} \norm{K_{u,v}}_2 \Big)    
        \end{split} 
    \end{equation} 
for all $t \in [0;1]$, provided $\alpha$ and then $n$ are large enough. 
\end{cor} 

We will also need bounds similar to the one in Lemma \ref{lm:prelim1}  for cubic operators with slightly more general kernels. 
\begin{cor} \label{cor:prelim3} 
For  $f \in L^1 (\Lambda)$ and a kernel $K$ as in Lemma \ref{lm:prelim1}, let 
\[ \cC_{K,f} = \int dy dy_1 dy_2 dy_3 \, f (y) K (y+ y_1 - y_2 , y+ y_1 - y_3) a_{y_1}^* a_{y_2}^* a_{y_3}^*  + \hc .\]
Then we have, for any fixed $\epsilon > 0$, 
\[ \begin{split} \big| \la e^{t A^*\Theta-\hc}  \Omega, &\cC_{K,f}   e^{t A^*\Theta-\hc}\Omega\ra \big| 
   \\ &\leq  |\langle \Omega, [\cC_{K,f} , A^* ] \Omega \rangle | + C \| f \|_1 N^{(3\kappa-1)/2 - \beta/4 + \epsilon}  \sum_{u,v \in\Lambda_B} \norm{K_{u,v}}_2  \end{split} \]    
for all $t \in [0;1]$, provided $\alpha$ and then $n$ are large enough. 
\end{cor} 

\begin{proof} 
We proceed as in the proof of Lemma \ref{lm:prelim1}. To bound contributions to $[ \cC_{K,f} , A^*]$ that are quadratic or quartic in creation and annihilation operators, we write $\cC_{K,f}$ as average with weight $f$ of operators that have exactly the form considered in Lemma \ref{lm:prelim1}. This observation allows us to derive the desired bound.
\end{proof} 

The last two corollaries allow us to control the expectation $\langle \xi, (\tilde{\cC}_N + \tilde{\cC}_N^*)\xi \rangle$ appearing on the r.h.s. of (\ref{eq:lmHN1}).  
\begin{lemma}\label{lm:tildeCN}
From (\ref{eq:lmHN1}) recall the definition 
\[ \tilde{\cC}_N =  N^{\kappa-1} N_0^{1/2} \sum_{p, r\in\Lambda^*} \hat{V}(r/N^{1-\kappa}) 
    ((\gamma_{p+r}\gamma_r-1)\sigma_p + \sigma_{p+r}\sigma_r\gamma_p)     a_{p+r} a_{-r} a_{-p}. \]
 We have
    \begin{equation}
      \big| \la e^{A^*\Theta-\hc}\Omega,  \tilde\cC_N e^{A^*\Theta-\hc}\Omega\ra \big| \leq C N^{4\kappa-1} 
    \end{equation}
for $\alpha$ and $n$ chosen large enough. 
\end{lemma}
\begin{proof}
We decompose $\tilde\cC_N = \sum_{j=1}^4 \tilde\cC^{(j)}_N$, where 
 \[ \tilde\cC^{(j)}_N = \sum_{p,r \in \Lambda^*} K^{(j)}_{r,p} a_{p+r} a_{-r} a_{-p} \]
with the coefficients 
\[ \begin{split} K^{(1)}_{r,p} &=  N^{\kappa-1} N_0^{1/2} \hat{V} (r/N^{1-\kappa}) (\gamma_r - 1) \sigma_p \\ K^{(2)}_{r,p} &= N^{\kappa-1} N_0^{1/2} \hat{V} (r/N^{1-\kappa})\gamma_r (\gamma_{p+r}-1) \sigma_p  \\
K^{(3)}_{r,p} &= N^{\kappa-1} N_0^{1/2} \hat{V} (r/N^{1-\kappa}) \sigma_{p+r} \sigma_r \\
K^{(4)}_{r,p} &= N^{\kappa-1} N_0^{1/2} \hat{V} (r/N^{1-\kappa}) \sigma_{p+r} \sigma_r (\gamma_p -1). \end{split} \]
With the momentum space bounds from Lemma \ref{lemma:propertiesquadratickernels}, we get $\| K^{(j)} \|_2 \leq C N^{(5\kappa -1)/2}$, for all $j=1, \dots ,4$. Switching to position space, we write, for $j=1,\dots ,4$,  
\[ \tilde\cC^{(j)}_N =  \int dx dy dz \, \check{K}^{(j)} (x-y , x-z) a_x a_y a_z \]
with 
\[ \begin{split} \check{K}^{(1)} (y,z) &= N^{\kappa-1}N_0^{1/2} [N^{3-3\kappa}V(N^{1-\kappa}\cdot)* (\widecheck{\gamma-1})](y) \check\sigma(z) \\
\check{K}^{(2)} (y,z) &= N^{\kappa-1}N_0^{1/2} \int dw \, [N^{3-3\kappa}V(N^{1-\kappa}\cdot) * \check\gamma](w) (\widecheck{\gamma-1})(w+y)\check\sigma(w+z)  \\ 
\check{K}^{(3)} (y,z) &= N^{\kappa-1}N_0^{1/2}  [N^{3-3\kappa}V(N^{1-\kappa}\cdot ) * \check\sigma](y) \check\sigma (z)\\ 
\check{K}^{(4)} (y,z) &= N^{\kappa-1}N_0^{1/2} \int dw \, [N^{3-3\kappa}V(N^{1-\kappa}\cdot)*  \check\sigma](w) \check\sigma (w+y)(\widecheck{\gamma-1})(w+z) .\end{split} \]
For $u \in \Lambda_B= \Lambda \cap \ell_B \bZ^3$, we recall $\check{\sigma}_u =  \check\sigma \mathbbm{1}_{B_u}$ and $\widecheck{(\gamma-1)}_u = \widecheck{(\gamma -1)} \mathbbm{1}_{B_u}$, with $B_u = \{x\in\Lambda, \inf_{v\in \Lambda_B} \abs{x-v} =\abs{x-u} \}$. For $u,v \in \Lambda_B$, we define the localized kernel 
\[ \begin{split} K^{(1)}_{u,v} (y,z) &= N^{\kappa-1} N_0^{1/2} [N^{3-3\kappa} V (N^{1-\kappa} \cdot) * \widecheck{(\gamma-1)}_u] (y) \check{\sigma}_v (z) \end{split} \] 
and we notice that $\| K^{(1)}_{u,v} \|_2 \leq C N^{\kappa-1/2} \| \widecheck{(\gamma-1)}_u \|_2 \| \check{\sigma}_v \|_2$. With Lemma \ref{lemma:sumlocalizedsigma} we find 
\[ \sum_{u,v \in \Lambda_B} \| K^{(1)}_{u,v} \|_2 \leq C N^{(5\kappa-1)/2 + \epsilon} \]
for any $\epsilon > 0$, if $\alpha$ and then $n$ are large enough. Recalling that $\| K^{(1)} \|_2 \leq C N^{(5\kappa-1)/2}$, we obtain from Corollary \ref{cor:prelim-2} that 
\[ \big| \la e^{A^*\Theta-\hc}\Omega,  \tilde\cC^{(1)}_N e^{A^*\Theta-\hc}\Omega\ra \big| \leq C N^{4\kappa-1} (1+ N^{ -\beta/4+\epsilon}) \leq C N^{4\kappa-1}\]
if $\epsilon > 0$ is chosen small enough (and $\alpha$ and then $n$ are large enough). 
The expectation of $\tilde{\cC}^{(3)}_N$ can be bounded analogously. To control the contribution of $\tilde{\cC}^{(2)}_N$, we observe that 
\[ \check{K}^{(2)} (y,z) =  \int dw \, f (w) \tilde{K}^{(2)} (w + y , w+z) \]
with $f (w) = [N^{3-3\kappa} V (N^{1-\kappa} \cdot ) * \check{\gamma}] (w)$ and $\tilde{K}^{(2)} (y,z) = N^{\kappa-1}N_0^{1/2} \widecheck{(\gamma-1)} (y) \check{\sigma} (z)$. From Lemma \ref{lemma:sumlocalizedsigma}, we have 
\[ \sum_{u,v \in\Lambda_B} \| \tilde{K}^{(2)}_{u,v} \|_2 \leq N^{5\kappa/2-1/2+\epsilon} .\]
Since moreover \[ | \langle \Omega, [ \tilde{\cC}_N^{(2)} , A^* ] \Omega \rangle |  \leq C \| K^{(2)} \|_2 \| A \|_2 \leq C N^{4\kappa-1} \] 
and $\| f \|_1 \leq C N^{\beta/8+\epsilon}$ (with the bounds in Lemma \ref{lemma:propertiesquadratickernels}, we find $\| \widecheck{\gamma-1} \|_1 \leq C N^{\beta/8+\epsilon}$, for any $\epsilon > 0$), Corollary \ref{cor:prelim3} implies that 
\[ \big| \la e^{A^*\Theta-\hc}\Omega,  \tilde\cC^{(2)}_N e^{A^*\Theta-\hc}\Omega\ra \big| \leq C N^{4\kappa-1} (1 + N^{-\beta/8 + \epsilon} ) \leq C N^{4\kappa-1} \]
if $\epsilon > 0$ is chosen sufficiently small. The expectation of $\tilde\cC^{(4)}_N$ can be handled analogously. 
\end{proof}

In the next lemma, we show a bound for the expectation of the operator $\tilde{\cN} = \sum_{p \in \Lambda^*} \sigma_p^2 a_p^* a_p$. Since it will be useful in the next sections, we actually prove a more general estimate, controlling the expectation of general non-negative quadratic operators, diagonal in momentum space. For $k \in L^1 (\Lambda )$, we introduce the notation   
\[ \cN_k = \int dz dz' \, k (z-z') a_z^* a_{z'} \, .\]
\begin{lemma}
\label{lemma:N_g}
Let $h \in L^2 (\Lambda )$, real-valued. We assume that $h (x) = \sum_{u \in \Lambda_B} h_u (x)$, where $h_u \in L^2 (\Lambda )$ for all $u \in \Lambda_B$, with $h_u (x) = 0$, if $|x-u| > 10 \ell_B$. 
%
%
Then, for every $\epsilon > 0$ there is a constant $C > 0$ such that  
\begin{equation}\label{eq:Nhh-claim} \begin{split}  \langle e^{t A^* \Theta - \hc} \Omega , &\cN_{h \star h} e^{t A^* \Theta- \hc} \Omega \rangle \\ &\leq  C  \langle \Omega, \big[ A , \big[ \cN_{h\star h} , A^* \big] \big] \Omega \rangle + C N^{-5\beta/8 + \epsilon} \Big(\sum_{u \in \Lambda_B} \| h_u \|_2 \Big)^2 \end{split} \end{equation} 
for all $t \in [0;1]$, if $\alpha$ and then $n \in \bN$ are large enough. Here, we used the notation \[ h \star h (x) = \int h (y-x) h (y) dy \] ($h \star h$ is the usual convolution, if $h (-y) = h(y)$). In particular, recalling that $\tilde{\cN} = \cN_{\check{\sigma} \star \check{\sigma}}$, we find $C > 0$ such that, for all $t \in [0;1]$, 
\[ \langle e^{t A^* \Theta - \hc} \Omega , \tilde{\cN} e^{t A^* \Theta- \hc} \Omega \rangle \leq C N^{3\kappa-1} .\]
\end{lemma} 

%

{\it Remark.} For $k \in L^1 (\Lambda )$, we also define the operator 
\begin{equation}\label{eq:dotN} \ddot{\cN}_k = \int dz dz' \, k(z-z') [a_{z'}, A^* ] [ a_z^*, -A] \end{equation} 
contributing to the second derivative of $\cN_k$ (hence the notation). From the proof of Lemma \ref{lemma:N_g}, we also obtain  
\[ \big| \la e^{s A^* \Theta -\hc} \Omega, \ddot{\cN}_{h \star h} e^{s A^* \Theta- \hc} \Omega \ra \big| \leq C N^{-5\beta/8 + \epsilon} \big( \sum_{u \in \Lambda_B} \| h_u \| \big)^2 \]
for all $s \in [0;1]$ and for any $\epsilon > 0$, if $\alpha$ and then $n$ are chosen large enough in the definition of the cutoff $\Theta$. This bound will be used in the next subsection. 

\begin{proof}
Throughout the proof, we will use the notation $g = h \star h$. By Cauchy-Schwarz, we have 
\[ \| h \|_1 = \int |h (x)| dx \leq \sum_{u \in \Lambda_B} \int |h_{u} (x)| dx \leq  C \ell_B^{3/2}  \sum_{u \in \Lambda_B} \| h_{u} \|_2  \]
and thus 
\[ \begin{split} \| \hat{h} \|_\infty \leq C \ell_B^{3/2} \sum_{u \in \Lambda_B} \| h_{u} \|_2 , \qquad 
\| \hat{g} \|_\infty \leq \| g \|_1 \leq C \ell_B^{3} \big(\sum_{u \in \Lambda_B} \| h_{u} \|_2 \big)^2 .\end{split} \]

Setting $\xi_t = e^{t A^* \Theta - \hc} \Omega$ for $t \in [0;1]$, we find 
    \begin{equation*}
    \begin{split}
        &\la \xi_t , \cN_g \xi_t \ra = 2 \text{Re } 
        \int_0^t ds \la \xi_s , 
        \big( [\cN_g, A^*]
        +[\cN_g, A^*](\Theta-1)
        +A^* [\cN_g,(\Theta-1)] \big)
        \xi_s \ra.
    \end{split}
    \end{equation*}
All terms containing a factor $(\Theta-1)$ can be handled with Lemma \ref{lemma:propertiesL}, in particular with \eqref{eq:norm1-Theta}, \eqref{eq:easyboundNm}, \eqref{eq:AThetaA}. Since $\cN_g \leq \| \hat{g} \|_\infty \cN$, $\| A (\cN+1)^{-3/2}\| \leq C N^{(3\kappa-1)/2}$, it is easy to check that they can be included in the error proportional to $(\sum_u \| h_u \|_2)^2$ on the r.h.s. of (\ref{eq:Nhh-claim}), if we choose $\alpha >0$ and then $n \in \bN$ large enough in (\ref{eq:defcubictransform}). Also in the rest of the proof we will ignore these terms. With the definition (\ref{eq:A3A6}), we compute 
\begin{equation}\label{eq:NgAstar} 
    [\cN_g, A^*]
    = \int dzdz'dx_1dx_2  \, 
      g(z-z')\check A^{(3)} (x_1-x_2,x_1-z')a^*_za^*_{x_1}a^*_{x_2}.  
\end{equation}
Expanding once more, we arrive at
\[  \langle \xi_t , \cN_g \xi_t \rangle \leq 2 \text{Re } \int_0^t ds_1 \int_0^{s_1} ds_2 \, \langle \xi_{s_2} , [[ \cN_g , A^* ] , -A ] \xi_{s_2} \rangle + C N^{-5\beta/8 + \epsilon} \big(\sum_{u \in \Lambda_B} \| h_u \|_2 \big)^2 \]
where
\begin{equation}\label{eq:decoNg} 
[[\cN_g, A^*],-A]=\la \Omega, [A,[\cN_g,A^*]]\Omega \ra+T_2+T_{4,1}+T_{4,2}
\end{equation}
with 
\[ \begin{split} 
T_{2}&=-\int dzdz'dx_1dx_2dx_3 \, g(z-z')\check A^{(6)} (x_1-x_2,x_1-z')\check A^{(6)} (x_3-x_2,x_3-z)a^*_{x_1}a_{x_3}\\&\hspace{0.4cm} + \int dzdz'dx_1dx_2dx_3 \, g(z-z')\check A^{(3)} (x_1-x_2,x_1-z')\check A^{(6)} (x_3-x_2,x_3-x_1)a^*_{z}a_{x_3}
    \end{split}\]
    and 
\begin{equation*}
    \begin{split}
T_{4,1}&=\int dzdz'dx_1dx_2dx_3dx_4 \, g(z-z') \\ &\hspace{3cm} \times \check A^{(6)} (x_1-x_2,x_1-z')\check A^{(3)} (x_3-x_4,x_3-x_1)a^*_{z}a^*_{x_2}a_{x_3}a_{x_4}\\T_{4,2}&=\int dzdz'dx_1dx_2dx_3dx_4 \, g(z-z')  \\ &\hspace{3cm} \times \check A^{(3)} (x_1-x_2,x_1-z')\check A^{(3)} (x_3-x_4,x_3-z)a^*_{x_1}a^*_{x_2}a_{x_3}a_{x_4} 
   \end{split}
\end{equation*}
where we recall the kernels $\check{A}^{(3)}, \check{A}^{(6)}$ defined in (\ref{eq:A3A6}). 
Using (\ref{eq:symmApq}), we pass to momentum space, writing
\[ T_2 = \sum_{p,q \in \Lambda^*} \hat{g}_p \, A^{(6)}_{q,p} \, A^{(3)}_{q,p} \big( a_p^* a_p + a_q^* a_q + a_{-p-q}^* a_{-p-q} \big). \]
With (\ref{eq:symmApq}) and with Lemma \ref{lm:A3A6}, we estimate
\[\begin{split}  \big| \la \xi_{s_2},  T_2 \xi_{s_2} \ra \big| &\leq C N^{3\kappa/2-1} \| \hat{g} \|_\infty \la \xi_{s_2},  (\cN + \tilde{\cN}) \xi_{s_2} \ra \\ &\leq C N^{\beta/4-1} \big(\sum_{u \in \Lambda_B} \| h_u \|_2 \big)^2  \la \xi_{s_2} , (\cN + \tilde{\cN}) \xi_{s_2} \ra .\end{split} \]

 We now consider the term $T_{4,1}$. Using Lemma \ref{lemma:propertiesL} to insert the cutoff $\Theta$ (choosing $\alpha > 0$ and $n \in \bN$ large enough), it is enough to estimate 
 \begin{equation*}
    \begin{split}
        &\abs{\la \xi_{s_2}, \Theta T_{4,1} \Theta \xi_{s_2} \ra}\\
        &\leq 
        \int dy dz' dx_1 dx_2 dx_3 dx_4   
    \abs{h (y-z')}
    \abs{\check A^{(6)} (x_1-x_2,x_1-z')} \abs{\check A^{(3)} (x_3-x_4,x_3-x_1)} 
     \\
    &\hspace{2cm} \times \Big| \int dz\, h (y-z)\la \Theta \xi_{s_2} ,
    a^*_z a_{x_2}^*
     a_{x_3} a_{x_4}
    \Theta \xi_{s_2} \ra \Big|. \end{split} 
    \end{equation*}
Next, we localize $h$ and the kernels $\check{A}^{(3)}, \check{A}^{(6)}$. With Cauchy-Schwarz, we obtain 
    \[ \begin{split} 
    &\abs{\la \xi_{s_2}, \Theta T_{4,1}\Theta  \xi_{s_2} \ra}\\  
        &\leq \sum_{u,v,\tilde u, \tilde v, w \in \Lambda_B}
    \Big(\int dy dz' dx_1 dx_2 dx_3 dx_4 \, |h_{w} (y-z')|^2 \abs{\check A^{(3)}_{\tilde u,\tilde v} (x_3-x_4,x_3-x_1)}^2
   \\
    &\hspace{2cm}
\times        \mathbbm{1}(\abs{x_1-x_2-u}, \abs{x_1-z'-v}\leq 10 \ell_B)\Big\| 
    \int dz\, h (y-z) a_{x_2} a_z 
    \Theta \xi_{s_2} \Big\|^2\Big)^{1/2}
    \\ &\hspace{1cm}  
    \times\Big(\int dy dz' dx_1 dx_2 dx_3 dx_4   
    \abs{\check A^{(6)}_{u,v}(x_1-x_2,x_1-z')}^2
     \\
    &\hspace{2cm} \times \mathbbm{1}(\abs{x_3-x_4-\tilde{u}},\abs{x_3-x_1-\tilde{v}},\abs{y-z'-w} \leq 10 \ell_B)
    \norm{
     a_{x_3} a_{x_4}
    \Theta \xi_{s_2} }^2\Big)^{1/2}.\end{split}\]
Applying Lemma \ref{lemma:propertiesL} and then Lemma \ref{lm:A3A6} we find, for any $\epsilon > 0$, 
 \[ \begin{split}  \abs{\la \xi_{s_2} , &\Theta T_{4,1} \Theta \xi_{s_2} \ra} 
     \\ &\leq C \ell_B^{9/2}  N^{\epsilon/2} \big(\sum_{ w\in \Lambda_B} \norm{h_{w}}_2\big)
    \big(\sum_{u,v\in \Lambda_B} \norm{\check A^{(6)}_{u,v}}_2\big)\big(\sum_{\tilde{u},\tilde{v}\in \Lambda_B} \norm{\check A^{(3)}_{\tilde{u},\tilde{v}}}_2\big) 
    \\&\hspace{3cm}\times \| \cN^{1/2} \Theta \xi_{s_2} \| \Big( \int dy\; \Big\|  \int dz\, h (y-z)
     a_z 
    \Theta \xi_{s_2} \Big\|^2 \Big)^{1/2} \\  
    &\leq C N^{\beta/4-1+\epsilon} \big( \sum_u \| h_u \|_2 \big)^2 \langle \xi_{s_2} , \cN \xi_{s_2} \rangle + C \langle \xi_{s_2}, \cN_g \xi_{s_2} \ra + C N^{-5\beta/8+\epsilon}  \big( \sum_u \| h_u \|_2 \big)^2
    \end{split}
\]
if the parameters $\alpha > 0$ and $n$ in the definition of $\Theta$ are chosen sufficiently large. 

To bound $T_{4,2}$, we expand it once again, estimating 
\[ \begin{split} \text{Re } \int_0^t &ds_1 \int_0^{s_1} ds_2 \,  \langle \xi_{s_2} , T_{4,2} \xi_{s_2} \rangle \\ &\leq \text{Re }  \int_0^t ds_1 \int_0^{s_1} ds_2 \int_0^{s_2} ds_3 \, \la \xi_{s_3}, \Theta [T_{4,2}, A^* ] \Theta \xi_{s_3} \ra + C N^{-5\beta/8+\epsilon}  \big( \sum_u \| h_u \|_2 \big)^2 \end{split} \]
for $\alpha >0$ and $n\in \bN$ large enough (using again Lemma \ref{lemma:propertiesL} to insert $\Theta$). We compute  
\begin{equation*}
 [T_{4,2},A^*]= T_{4,2,3} + T_{4,2,5}
\end{equation*}
with 
\begin{equation*}
    \begin{split}
T_{4,2,3}&= \int dzdz'dx_1dx_2dx_3dx_4dx_5 \, g(z-z')\check A^{(3)} (x_1-x_2,x_1-z')\\&\hspace{3cm}\times \check A^{(3)} (x_3-x_4,x_3-z)\check A^{(6)} (x_5-x_3,x_5-x_4)a^*_{x_1}a^*_{x_2}a^*_{x_5}\\
T_{4,2,5}&= \int dzdz'dx_1dx_2dx_3dx_4dx_5dx_6 \, g(z-z')\check A^{(3)} (x_1-x_2,x_1-z')
\\&\hspace{3cm} \times \check A^{(6)} (x_3-x_4,x_3-z)A^{(3)} (x_5-x_6,x_5-x_4)a^*_{x_1}a^*_{x_2}a^*_{x_5}a^*_{x_6}a_{x_3}.
    \end{split}
\end{equation*}
For the quintic term, we find
\begin{equation}
\label{eq:t425}
    \begin{split}
        \abs{\la &\xi_{s_3}, \Theta T_{4,2,5} \Theta \xi_{s_3} \ra}\\ &\leq 
        \int dz dz' dx_1 dx_2 dx_3  dx_4 dx_5 dx_6 \,  
    \abs{g(z-z')} \abs{\check A^{(3)} (x_1-x_2,x_1-z')} \\ 
     &\hspace{.4cm} \times\abs{\check A^{(6)} (x_3-x_4,x_3-z)} \abs{\check A^{(3)} (x_5-x_6,x_5-x_4)}
     \norm{a_{x_{1}} a_{x_{2}} a_{x_{5}}a_{x_{6}}\Theta\xi_{s_3}}   \norm{a_{x_3}\Theta\xi_{s_3}}.
     \end{split} \end{equation} 
 Localizing the kernels $\check{A}^{(3)}, \check{A}^{(6)}$ and applying Cauchy-Schwarz, Lemma \ref{lemma:propertiesL} and Lemma~\ref{lm:A3A6}, we obtain  
  \[  \begin{split}   \abs{\la &\xi_{s_3}, \Theta T_{4,2,5} \Theta \xi_{s_3} \ra}\\ &\leq 
    \sum_{u_i ,v_i \in\Lambda_B : \, i=1,2,3}
        \Big(\int dzdz'dx_1 dx_2 dx_3  dx_4 dx_5 dx_6   \, 
    \abs{g(z-z')}\abs{\check A^{(6)}_{u_2,v_2}(x_3-x_4,x_3-z)}^2  \\ & \hspace{.5cm} \times
     \mathbbm{1}(|x_1-x_2-u_1|,|x_1-z'-v_1|,|x_5-x_6-u_3|,|x_5-x_4-v_3|\leq 10\ell_B)\\& \hspace{.5cm}\times \norm{a_{x_{1}} a_{x_{2}} a_{x_{5}}a_{x_{6}}\Theta\xi_{s_3}}^2\Big)^{1/2}\\
     &\times \Big(\int dzdz'dx_1 dx_2 dx_3  dx_4 dx_5 dx_6   \, 
    \abs{g(z-z')} \abs{\check A^{(3)}_{u_1,v_1}(x_1-x_2,x_1-z')}^2 \\ 
     &\hspace{.5cm} \times \abs{\check A^{(3)}_{u_3,v_3}(x_5-x_6,x_5-x_4)}^2    
       \mathbbm{1}(|x_3-x_4-u_2|,|x_3-z-v_2|\leq 10\ell_B) \norm{a_{x_{3}}\Theta\xi_{s_3}}^2\Big)^{1/2}\\
    &\leq C N^{\epsilon} \Big(\sum_{u,v\in\Lambda_B} \ell_B^{3/2}\norm{ \check A^{(3)}_{u,v}}\Big)^{2}\Big(\sum_{u,v\in\Lambda_B} \ell_B^{3/2}\norm{ \check A^{(6)}_{u,v}}\Big)
    \norm{g}_1 
    \la \Theta \xi_{s_3} , \cN \Theta \xi_{s_3} \ra\\
    &\leq C N^{3\beta/8 -1 + \epsilon} \big( \sum_{u \in \Lambda_B} \| h_u \| \big)^2 \, \la \xi_{s_3} , \cN \xi_{s_3} \ra + C N^{-5\beta/8+\epsilon}  \big( \sum_u \| h_u \|_2 \big)^2 
    \end{split}\]
for $n \in \bN$ large enough. As for the cubic terms $T_{4,2,3}$, we need to expand it a last time, writing  
\[\begin{split} &\text{Re }  \int_0^1 ds_1 \int_0^{s_1} ds_2 \int_0^{s_2} d{s_3} \, \la \xi_{s_3}, T_{4,2,3}  \xi_{s_3} \ra \\ &\leq \text{Re } \int_0^1 ds_1 \int_0^{s_1} ds_2 \int_0^{s_2} d{s_3}   \int_0^{s_3} ds_4  \, \la \xi_{s_4} , \Theta [ T_{4,2,3} , A^* ] \Theta  \xi_{s_4} \ra  + C N^{-5\beta/8+\epsilon}  \big( \sum_u \| h_u \|_2 \big)^2. \end{split} \]
A long but straightforward computation shows that  
\begin{equation*}
 [T_{4,2,3},A^*]= T_{4,2,3,0} + T_{4,2,3,2}+T_{4,2,3,4}
\end{equation*}
with 
\begin{equation*}
    \begin{split}
T_{4,2,3,0}&= \int dzdz'dx_1 \dots dx_5\;g(z-z')\check A^{(3)} (x_1-x_2,x_1-z')\\&\hspace{1cm}\times \check A^{(3)} (x_3-x_4,x_3-z)\check A^{(6)} (x_5-x_3,x_5-x_4)\check A^{(6)} (x_1-x_2,x_1-x_5)\\
T_{4,2,3,2}&= \int dzdz'dx_1 \dots dx_6 \;g(z-z')\check A^{(3)} (x_1-x_2,x_1-z')\check A^{(3)} (x_3-x_4,x_3-z) \\&\hspace{1cm}\times \check A^{(6)} (x_5-x_3,x_5-x_4)\check A^{(6)} (x_6-x_1,x_6-x_2)a^*_{x_5}a_{x_6}\\&\hspace{0.4cm}+\int dzdz'dx_1 \dots dx_6 \;g(z-z')\check A^{(6)} (x_1-x_2,x_1-z')\check A^{(3)} (x_3-x_4,x_3-z) \\&\hspace{1cm}\times  \check A^{(6)} (x_5-x_3,x_5-x_4)\check A^{(6)} (x_6-x_5,x_6-x_2)a^*_{x_1}a_{x_6}\\ T_{4,2,3,4}&= \int dzdz'dx_1 \dots dx_7\;g(z-z')\check A^{(6)} (x_1-x_2,x_1-z')\check A^{(3)} (x_3-x_4,x_3-z) \\&\hspace{1cm}\times \check A^{(6)} (x_5-x_3,x_5-x_4)\check A^{(3)} (x_6-x_7,x_6-x_1)a^*_{x_5}a^*_{x_2}a_{x_6}a_{x_7}\\&\hspace{0.4cm}+\int dzdz'dx_1 \dots dx_7\;g(z-z')\check A^{(3)} (x_1-x_2,x_1-z') \check A^{(3)} (x_3-x_4,x_3-z) 
\\&\hspace{1cm} \times \check A^{(6)} (x_5-x_3,x_5-x_4)\check A^{(3)} (x_6-x_7,x_6-x_5)a^*_{x_1}a^*_{x_2}a_{x_6}a_{x_7}.\\
    \end{split}
\end{equation*}
For $T_{4,2,3,0}$ and $T_{4,2,3,2}$ we switch to momentum space. With \eqref{eq:A3A6} and Lemma \ref{lm:A3A6}, we obtain   
\begin{equation*}
    \begin{split}
        \abs{T_{4,2,3,0}}
        &\leq 
        \sum_{p,q,r \in \Lambda^*} |\hat{g}_p| \abs{A^{(6)}_{r,p}}\abs{A^{(3)}_{r,p}}\abs{A^{(6)}_{q,p}}\abs{A^{(3)}_{q,p}} \\ & \leq C \| \hat{g} \|_\infty \sum_{p,q,r \in \Lambda^*} |A^{(6)}_{r,p} |^2 |A^{(3)}_{q,p}|^2 \\  & \leq C N^{-\beta/2} \|\hat{g}\|_\infty \sum_{p,r \in \Lambda^*}(1+\sigma_p^2)\abs{A^{(6)}_{r,p}}^2 \leq C N^{-3\beta/4+\epsilon} \big(\sum_u \| h_u \|_2 \big)^2  
    \end{split}
\end{equation*}
for any $\epsilon > 0$. Similarly,
\begin{equation*}
    \begin{split}
       | \langle \xi_{s_4} &, T_{4,2,3,2} \xi_{s_4} \ra | \\
        &\leq C
        \sum_{p,q,r \in \Lambda^*} |\hat{g}_p| \abs{A^{(6)}_{r,p}}\abs{A^{(3)}_{r,p}}\abs{A^{(6)}_{q,p}}\abs{A^{(3)}_{q,p}} \la \xi_{s_4} , (a_{p}^* a_{p} + a_{q}^* a_{q} + a_{p+q}^* a_{p+q}) \xi_{s_4} \ra 
    \\
    &\leq C N^{-\beta/2} \norm{\hat{g}}_\infty (1+ \| \sigma \|_\infty^2) 
        \sum_{j=3,6} \sum_{p,q \in \Lambda^*}  
    \abs{A^{(j)}_{q,p}}^2 \la \xi_{s_4} , (a_{p}^* a_{p} + a_{q}^* a_{q} + a_{p+q}^* a_{p+q}) \xi_{s_4} \ra \\
    &\leq C N^{-1+\epsilon} \big( \sum_u \| h_u \|_2 \big)^2 \la \xi_{s_4}, (\cN + \tilde\cN) \xi_{s_4} \ra .
    \end{split}
\end{equation*}
As for the quartic contribution, we find that, for any $\epsilon > 0$,  
\begin{equation*}
    \begin{split}
        \abs{ \la \Theta\xi_{s_4} , T_{4,2,3,4} \Theta\xi_{s_4} \ra}
     &\leq C N^{\epsilon} \Big(\big(\sum_{u,v \in \Lambda_B} \ell_B^{3/2}\norm{ \check A^{(3)}_{u,v}}\big)^{4} + \big(\sum_{u,v \in \Lambda_B} \ell_B^{3/2}\norm{ \check A^{(6)}_{u,v}}\big)^{4} \Big)
    \norm{g}_1 
    \la \xi_{s_4} , \cN \xi_{s_4} \ra \\ &\leq C N^{\beta/4 -1 + \epsilon}\big( \sum_u \| h_u \|_2 \big)^2  \la \xi_{s_4}, \cN \xi_{s_4}\ra,
    \end{split}
\end{equation*}
if $n \in \bN$ in the definition of $\Theta$ is large enough. This bound can be shown proceeding similarly as we did for $T_{4,2,5}$ in \eqref{eq:t425}; we omit the details.

Combining all the bounds we proved, we obtain 
\[ \begin{split} \la \xi_t , \cN_g \, \xi_t \rangle \leq \; &\text{Re } \langle \Omega, [ A, [ \cN_g, A^* ]]  \Omega \rangle  +  C N^{3\beta/8 -1 + \epsilon} \big( \sum_{u \in \Lambda_B} \| h_u \| \big)^2 \, \int_0^t ds \la \xi_s , (\cN + \tilde{\cN}) \xi_s \rangle \\ &+ C \int_0^t ds \la \xi_s , \cN_g \xi_s \ra + C N^{-5\beta/8+\epsilon} \big(\sum_{u \in \Lambda_B} \| h_u \|_2 \big)^2  \end{split}  \]
for all $t \in [0;1]$. With Lemma \ref{lm:boundN}, we arrive at
\begin{equation}\label{eq:boot} \begin{split}  \la \xi_t, \cN_g  \xi_t \rangle \leq \; &\text{Re } \langle \Omega, [ A, [ \cN_g, A^* ]]  \Omega \rangle + C N^{3\beta/8 -1 + \epsilon} \big( \sum_{u \in \Lambda_B} \| h_u \| \big)^2 \, \int_0^t ds \la \xi_s ,  \tilde{\cN} \xi_s \rangle \\ &+ C \int_0^t ds \la \xi_s , \cN_g \xi_s \ra + C N^{-5\beta/8+\epsilon} \big(\sum_{u \in \Lambda_B} \| h_u \|_2 \big)^2.  \end{split}  \end{equation}

Consider first the choice $h = \check{\sigma}$. By Lemma \ref{lemma:sumlocalizedsigma}, we have $\sum_u \| \check{\sigma}_u \|_2 \leq C N^{3\kappa/4+\epsilon}$, for any $\epsilon > 0$. Since $\tilde{\cN} = \cN_{\check{\sigma} * \check{\sigma}}$, (\ref{eq:boot}) implies that
\[ \begin{split} \la \xi_t, \tilde{\cN} \xi_t \rangle \leq \; &\text{Re } \langle \Omega, [ A, [ \tilde{\cN}, A^* ]]  \Omega \rangle + C (1 + N^{-\beta/8 + \epsilon}) \, \int_0^t ds \la \xi_s ,  \tilde{\cN} \xi_s \rangle  + C N^{3\kappa-1-\beta/8+\epsilon}  \end{split}  \]
Here, we have 
\[  \langle \Omega, [ A, [ \tilde{\cN}, A^* ]]  \Omega \rangle = \int dz dz' dx_1 dx_2 \, (\check{\sigma} * \check{\sigma} ) (z-z') A^{(3)} (x_1 - x_2, x_1 - z') A^{(6)} (z-x_1 , z-x_2) \]
and thus, passing to momentum space and using Lemma \ref{lm:A3A6}, 
\[ \big| \langle \Omega, [ A, [ \tilde{\cN}, A^* ]]  \Omega \rangle \big|  = \Big| \sum_{p,q \in \Lambda^*} \sigma_q^2 A^{(3)}_{p,q} A^{(6)}_{p,q} \Big| \leq \sum_{p,q \in \Lambda^*} \sigma_q^2 |A^{(3)}_{p,q}|  | A^{(6)}_{p,q}| \leq C N^{3\kappa-1} .\] 
Hence, choosing $\epsilon > 0$ small enough, we find 
 \[ \begin{split} \la \xi_t, \tilde{\cN} \xi_t \rangle \leq \; &C N^{3\kappa-1}+ C \int_0^t ds \la \xi_s ,  \tilde{\cN} \xi_s \rangle .\end{split}  \]
By Gronwall, we conclude that $\la \xi_t , \tilde{\cN} \xi_t \rangle \leq C N^{3\kappa-1}$ 
for all $t \in [0;1]$. 

Inserting this bound back on the r.h.s. of (\ref{eq:boot}) we find, now for a general $h \in L^2 (\Lambda )$,  
\[  \la \xi_t, \cN_{h \star h}  \xi_t \rangle \leq \; \text{Re } \langle \Omega, [ A, [ \cN_{h\star h} , A^* ]]  \Omega \rangle +  C \int_0^t ds \la \xi_s , \cN_g \xi_s \ra + C N^{-5\beta/8+\epsilon} \big(\sum_{u \in \Lambda_B} \| h_u \|_2 \big)^2.  \] 
Thus, again by Gronwall's lemma, we conclude that 
\[ \la \xi_t , \cN_{h*h} \xi_t \ra \leq C | \langle \Omega, [ A, [ \cN_{h*h} , A^* ]]  \Omega \rangle | + C N^{-5\beta/8+\epsilon} \big(\sum_{u \in \Lambda_B} \| h_u \|_2 \big)^2 \]
for all $t \in [0;1]$. 

The remark after the lemma follows by the remark that the operator $\ddot{\cN}_{h \star h}$ coincides exactly with the term $T_{4,2}$, introduced in (\ref{eq:decoNg}). 
\end{proof}

Finally, we estimate the expectation $\la \xi , \tilde{\cV}_N \xi \ra$, with $\tilde{\cV}_N$ as defined in (\ref{eq:CtCtV}). 
\begin{lemma} \label{lm:tildeVN}
We have 
\[ \la e^{A^* \Theta - \hc} \Omega, \tilde{\cV}_N e^{A^* \Theta - \hc} \Omega \ra \leq C N^{4\kappa-1} \]
if the parameters $\alpha$ and then $n$ in the definition (\ref{eq:def-Theta}) of the cutoff $\Theta$ are large enough. 
\end{lemma} 

\begin{proof}
We are going to bound the expectation of 
\[ \tilde{\cV}_N^{(0)} = \frac{N^\kappa}{2N} \sum_{p,q,r\in\Lambda^*} \hat{V}(r/N^{1-\kappa}) 
        (\gamma_{p+r} \gamma_{q-r}\gamma_{q}\gamma_{p}-1) a_{p+r}^* a_{q-r}^* a_q a_p . \]
All other contributions to $\tilde\cV_N$ can be controlled analogously, replacing factors of $\gamma-1$ by factors of $\sigma$. To estimate $\la \xi, \tilde{\cV}_N^{(0)} \xi \rangle$, we observe that, setting $\theta_p = \gamma_p -1$, we have 
\[ \begin{split} \gamma_{p+r} \gamma_{q-r} \gamma_q \gamma_p - 1 = \; &\theta_{p+r} \theta_{q-r} \theta_q \theta_p + \theta_{p+r} \theta_{q-r}  \theta_q + \theta_{p+r} \theta_{q-r} \theta_p  + \theta_{p+r} \theta_q \theta_p + \theta_{q-r} \theta_q \theta_p  
\\ &+ \theta_{p+r} \theta_{q-r}+ \theta_{p+r} \theta_q + \theta_{p+r} \theta_p + \theta_{q-r} \theta_q+ \theta_{q-r} \theta_p + \theta_{q} \theta_p  
\\ &+ \theta_{p+r} +  \theta_{q-r} +  \theta_{q} +  \theta_{p}. \end{split} \] 
Next, we switch to position space. Setting $V_N (x) = N^{2-2\kappa} V (N^{1-\kappa}x)$, we find
 \begin{equation}
\label{tildeV_N}
\begin{split}
\tilde\cV^{(0)}_N= \; &\frac{1}{2}\int dxdy \, V_N(x-y) \, a^* (\check{\theta}_x) a^* (\check{\theta}_y) a (\check{\theta}_x) a (\check{\theta}_y) \\&+ \int dxdy \, V_N(x-y) \, a^*_{x} a^* (\check{\theta}_y) a (\check{\theta}_x) a (\check{\theta}_y) + \text{h.c.}\\&+ \int dxdy \, V_N(x-y) \, a^*_{x}a^* (\check{\theta}_y) a_{x} a (\check{\theta}_y) \\&+ \int dxdy \, V_N(x-y) \, a^*_{x}a^* (\check{\theta}_y) a (\check{\theta}_x) a_{y} \\&+ \frac{1}{2}\int dxdy \, V_N(x-y) a^*_{x}a^*_{y}a (\check{\theta}_x) a (\check{\theta}_y) + \text{h.c.}\\&+\int dxdy \, V_N(x-y) \, a^*_{x}a^*_{y}a_{x}a (\check{\theta}_y) + \text{h.c.}
\end{split}
\end{equation}
where we used the notation $\check{\theta}_x (s) = \check{\theta} (x-s)$, for $x \in \Lambda $. When estimating the expectation of all these terms, we can insert, on both sides of the inner products, the cutoff $\Theta$. For the first four terms (the ones without $V_N(x-y)a_x a_y$ factor), this is done as before. As an example, we now show how to insert $\Theta$ for the last term in \eqref{tildeV_N} which has to be handled differently (even though we need to use Duhamel for showing that it is negligible).
\begin{equation}\label{eq:insTheta} \begin{split} \int &dx dy \, V_N (x-y) \la \xi , a_x^* a_y^* a_x a (\check{\theta}_y) \xi \ra \\ =\; & 
 \int dx dy \, V_N (x-y) \la \Theta \xi , a_x^* a_y^* a_x a (\check{\theta}_y) \Theta \xi \ra +
  \int dx dy \, V_N (x-y) \la \xi , a_x^* a_y^* a_x a (\check{\theta}_y) (1-\Theta) \xi \ra \\ &+
   \int dx dy \, V_N (x-y) \la (1-\Theta) \xi , a_x^* a_y^* a_x a (\check{\theta}_y) \Theta \xi \ra. \end{split} \end{equation} 
The second (and analogously the third) term on the r.h.s. can be bounded by 
\[ \begin{split} \Big|   \int &dx dy \, V_N (x-y) \la \xi , a_x^* a_y^* a_x a (\check{\theta}_y) (1-\Theta) \xi \ra \Big| \\ \leq\; & \Big( \int dx dy \| a_x a_y (\cN+1) \xi \|^2 \Big)^{1/2} \Big( \int dx dy \, V_N^2 (x-y) \| a_x a (\check{\theta}_y) (\cN+1)^{-1} (1-\Theta) \xi \|^2 \Big)^{1/2} \\ \leq \; &\sup_y \| \check{\theta}_y \|_2 \| V_N \|_2 \| (\cN+1)^2 \xi \|  \| (1-\Theta) \xi \|. \end{split} \]
With Lemma \ref{lemma:propertiesL}, we conclude that, for any $\nu > 0$,  
\[ \Big|   \int dx dy \, V_N (x-y) \la \xi , a_x^* a_y^* a_x a (\check{\theta}_y) (1-\Theta) \xi \ra \Big| \leq C N^{-\nu} \]
if $\alpha > 0$ and then $n \in \bN$ are chosen large enough in \eqref{eq:defcubictransform}. 
Hence, we are left only with the first term on the r.h.s. of (\ref{eq:insTheta}), where $\xi$ has been replaced by $\Theta \xi$. 

In Lemma \ref{intermediate_lemma} below, we are going to show that  
\begin{equation}\label{eq:interm1} \begin{split} \int dx dy \, |V_N (x-y)| \| a (\check{\theta}_x) a (\check{\theta}_y) \Theta \xi \|^2 \leq C N^{4\kappa-1} N^{-\beta/4 + \epsilon} \\
\int dx dy\, | V_N (x-y)| \| a_x a (\check{\theta}_y) \Theta \xi \|^2 \leq C N^{4\kappa-1} N^{-\beta/2 + \epsilon}. \end{split} \end{equation}
With these bounds and with Cauchy-Schwarz, we can estimate the expectation of the first 4 terms on the r.h.s. of (\ref{tildeV_N}). 

Let us consider the fifth term on the r.h.s. of (\ref{tildeV_N}). Recalling the definition (\ref{eq:A3A6}) of the kernels $\check{A}^{(3)}, \check{A}^{(6)}$, we find 
\[ [a^*_xa^*_y,-A]=\cT_1(x,y)+\cT_3(x,y) \]
with 
\begin{equation}\label{eq:T3T1}
\begin{split}
\cT_3(x,y)&= \int dx_1dx_2 \, \Big( \check{A}^{(3)} (x_1-x_2,x_1-x)a^*_ya_{x_1}a_{x_2}+\check{A}^{(3)} (x_1-x_2,x_1-y)a^*_xa_{x_1}a_{x_2} \Big) ,\\
\cT_1(x,y)&= \int dz \check{A}^{(6)} (z-y,z-x)a_z \, .
\end{split}
\end{equation}

Expanding $e^{-A^* \Theta + \hc} a_x^* a_y^* e^{A^* \Theta - \hc}$ we find 
\[
\begin{split}
\int dx dy \,&V_N(x-y)\;\la \xi,a^*_xa^*_ya(\check{\theta}_x) a(\check{\theta}_y) \xi\ra \\ 
= \; &\int_0^1 ds \int dxdy \, V_N(x-y) \;\la \Theta\xi_s, \cT_3 (x,y)e^{-(1-s) A^*\Theta-\hc} a(\check{\theta}_x) a(\check{\theta}_y) \Theta\xi \ra  \\&+\int_0^1 ds \int dxdy \, V_N(x-y) \; \la  \xi_s, \cT_1 (x,y) e^{-(1-s)A^*\Theta-\hc} a(\check{\theta}_x) a(\check{\theta}_y) \xi \ra + \cE_\Theta
\end{split} \]
with an arbitrarily small error $\cE_\Theta$. In the second term, we expand the action of $e^{(1-s)A^* \Theta - \hc}$ on $\cT_1 (x,y)$. We find 
\[ \begin{split} 
\int dx &dy \, V_N(x-y)\;\la \xi,a^*_xa^*_ya(\check{\theta}_x) a(\check{\theta}_y) \xi\ra \\ 
= \; &\int_0^1 ds \int dxdy \, V_N(x-y) \;\la \Theta \xi_s, \cT_3 (x,y)e^{-(1-s) A^*\Theta-\hc} a(\check{\theta}_x) a(\check{\theta}_y) \Theta\xi \ra \\ &+  \int dxdy \, V_N(x-y) \;\la  \xi, \cT_1 (x,y) a(\check{\theta}_x) a(\check{\theta}_y) \xi \ra \\ &+ \int_0^1 ds \int_0^{1-s} dt 
\int dx dy \, V_N (x-y) \la \Theta \xi_{1-t},  [ \cT_1 (x,y) , A^* ] e^{-t A^* \Theta - \hc} a (\check{\theta}_x) a (\check{\theta}_y) \Theta \xi \rangle \\ &+ \cE_\Theta.
\end{split} \]
Finally, in the second term we expand the action of $e^{A^* \Theta - \hc}$on $\cT_1 (x,y) a (\check{\theta}_x) a (\check{\theta}_y)$. We arrive at 
\begin{equation}\label{eq:VN5-1}  \begin{split} 
\int dx &dy \, V_N(x-y)\;\la \xi,a^*_xa^*_ya(\check{\theta}_x) a(\check{\theta}_y) \xi\ra \\ 
= \; &\int_0^1 ds \int dxdy \, V_N(x-y) \;\la \Theta \xi_s, \cT_3 (x,y)e^{-(1-s) A^*\Theta-\hc} a(\check{\theta}_x) a(\check{\theta}_y) \Theta\xi \ra  \\ &+ \int_0^1 ds \int_0^{1-s} dt 
\int dx dy \, V_N (x-y) \la \Theta \xi_{1-t},  [ \cT_1 (x,y) , A^* ] e^{-t A^* \Theta - \hc} a (\check{\theta}_x) a (\check{\theta}_y) \Theta \xi \rangle \\ &+  \int_0^1 ds \int dxdy \, V_N(x-y) \;\la \Theta \xi_s , [ \cT_1 (x,y) , A^*] a(\check{\theta}_x) a(\check{\theta}_y) \Theta\xi_s \ra \\ &+  \int_0^1 ds \int dxdy \, V_N(x-y) \;\la \Theta \xi_s ,  [ a(\check{\theta}_x) a(\check{\theta}_y) , A^*]   \cT_1 (x,y)  \Theta\xi_s \ra\\ &+  \int_0^1 ds \int dxdy \, V_N(x-y) \;\la \Theta \xi_s , \big[ \cT_1 (x,y) ,  [ a(\check{\theta}_x) a(\check{\theta}_y) , A^*] \big]  \Theta\xi_s \ra  + \cE_\Theta.
\end{split} \end{equation} 
In Lemma \ref{intermediate_lemma2}, we will prove that 
\begin{equation}\label{eq:interm2} \begin{split} 
  \int dxdy \, | V_N(x-y)| \|\cT_3^*(x,y) \Theta \xi_s\|^2 &\leq CN^{\epsilon} N^{4\kappa-1}N^{-\beta} \\
    \int dxdy \, |V_N(x-y)| \|[a^*(\check{\theta}_x) a^*(\check{\theta}_y) ,A]\Theta \xi_s\|^2 &\leq CN^{\epsilon} N^{4\kappa-1}N^{-\beta/2} \\
     \int dxdy\, | V_N(x-y)| \|[\cT_1^*(x,y),A] \Theta \xi_s\|^2 &\leq CN^{\epsilon} N^{4\kappa-1}\\
     \int dxdy \, |V_N(x-y)| \|\cT_1(x,y)\Theta\xi_s\|^2 &\leq CN^{\epsilon} N^{4\kappa-1+\beta/4} .\end{split} \end{equation} 
Together with (\ref{eq:interm1}), these bounds can be used to control the first 4 terms on the r.h.s. of (\ref{eq:VN5-1}). In Lemma \ref{lm:interm3}, we will also show that 
\begin{equation}\label{eq:cT1-cc}  \Big| \int dx dy \, V_N (x-y) \, \la \Theta\xi_s,  \big[ \cT_1 (x,y) , [ a (\check{\theta}_x) a (\check{\theta}_y) , A^*] \big] \Theta \xi_s \rangle \Big| \leq C N^{4\kappa-1} \, .\end{equation} 
By (\ref{eq:VN5-1}), this implies that 
\begin{equation}\label{eq:lastVN5-1}  \Big| \int dx dy V_N (x-y) \la \xi_s a_x^* a_y^* a (\check{\theta}_x) a (\check{\theta}_y) \xi_s \rangle \Big| \leq C N^{4\kappa-1} .\end{equation}
Proceeding very similarly as in the proof of (\ref{eq:lastVN5-1}), we can also show 
\[ \Big| \int dx dy V_N (x-y) \la \xi a_x^* a_y^* a _x a (\check{\theta}_y) \xi \rangle \Big| \leq C N^{4\kappa-1}. \] 
The only term for which the details of the analysis needs some additional adjustments is
\[
\begin{split}
    \int_0^1 ds \int dxdy \, V_N(x-y) \;\la \Theta \xi_s ,  [ a_x a(\check{\theta}_y) , A^*]   \cT_1 (x,y)  \Theta\xi_s \ra,
\end{split}
\]
which is of order $N^{4\kappa-1 +\epsilon -\beta/4}$.

From \eqref{tildeV_N}, we obtain $\la \xi , \tilde{\cV}_N^{(0)} \xi \rangle \leq C N^{4\kappa-1}$, as claimed.
\end{proof}

In the next lemma, we show the bounds (\ref{eq:interm1}). 
\begin{lemma}
\label{intermediate_lemma}
As in the proof of Lemma \ref{lm:tildeVN}, we let $\xi_s = e^{A^*\Theta-\hc}\Omega$, for all $s \in [0;1]$, $\theta = \gamma-1$, $V_N (x) = N^{2-2\kappa} V (N^{1-\kappa}x)$. For every $\epsilon > 0$, there is $C > 0$ such that 
    \begin{equation}\label{eq:inter-lm}
    \begin{split}
  T_0 &:=  \int dxdy \, |V_N(x-y)| \| a(\check{\theta}_x) a(\check{\theta}_y) \Theta\xi_s \|^2 \leq C N^{4\kappa-1} N^{-\beta/4+\epsilon} \\
   T_1 &:= \int dxdy \, |V_N(x-y)| \|a_x a(\check{\theta}_y)  \Theta\xi_s \|^2 \leq C N^{4\kappa-1} N^{-\beta/2+\epsilon}
      \end{split}  
    \end{equation}
    for all $s \in [0;1]$, if $\alpha$ and then $n$ are chosen sufficiently large in the definition (\ref{eq:def-Theta}) of the cutoff $\Theta$. 
\end{lemma}

\begin{proof}
We write 
\begin{equation}
    \begin{split}
    T_0& =\int dxdydrdr' dt dt' \, \check{\theta} (r) \check{\theta} (r')\check{\theta} (t) \check{\theta} (t') |V_N (x-y)| \la \Theta \xi_s , a^*_{y-t}a^*_{x-r}a_{y-t'}a_{x-r'} \Theta \xi_s\rangle.
    \end{split}
\end{equation}
Next, we localize $\check{\theta} = \sum_{u \in \Lambda_B} \check{\theta}_u$, where we recall that $\check{\theta}_u = \check{\theta} \mathbbm{1}_{B_u}$ and  $B_u = \{x\in\Lambda, \inf_{v\in \Lambda_B} \abs{x-v} =\abs{x-u} \}$ with $\Lambda_B = \Lambda \cap \ell_B \mathbb{Z}^3$. With Cauchy-Schwarz, we obtain 
\begin{equation*}
    \begin{split}
    |T_0| &\leq \sum_{u, u',v,v'}\Big(\int dxdydrdr'dtdt' \; |V_N(x-y)|   \mathbbm{1}_{B_u}(r) \mathbbm{1}_{B_v} (t) \\ &\hspace{7cm} \times \check{\theta}_{u'}(r')^2\check{\theta}_{v'}(t')^2\; \| a_{y-t}a_{x-r} \Theta \xi_s\|^2\Big)^{1/2}\\&\hspace{.1cm} \times\Big(\int dxdydrdr'dtdt' \; |V_N(x-y)| \mathbbm{1}_{B_{u'}}(r')\mathbbm{1}_{B_{v'}} (t')  \check{\theta}_{u}(r)^2\check{\theta}_{v}(t)^2  \|a_{y-{t'}}a_{x-{r'}}\Theta \xi_s\|^2\Big)^{1/2}
    \\ &\leq CN^{\epsilon} \|V_N\| _1 \ell_{B}^3  \Big(\sum_{u  \in \Lambda_B} \norm{\check{\theta}_u }_2 \Big)^4 \|\cN^{1/2}\xi_s\|^2 .
    \end{split}
\end{equation*}
With Lemma \ref{lemma:sumlocalizedsigma}, Lemma \ref{lm:boundN} and since $\| V_N \|_1 \leq C N^{-1+\kappa}$, we conclude that 
\[ |T_0| \leq C N^{4\kappa-1} N^{-\beta/4+\epsilon} . \]
The second bound in (\ref{eq:inter-lm}) can be shown similarly. In this case, we can remove two factors of $\check{\theta}$ and the integrals over $r,r'$; as a result, we obtain 
\[ |T_1| \leq C N^\epsilon \| V_N \|_1 \Big(\sum_{u  \in \Lambda_B} \norm{\check{\theta}_u }_2 \Big)^2 \la \xi_s, \cN\xi_s\ra \leq C N^{4\kappa-1} N^{-\beta/2+\epsilon}. \]
\end{proof}

Next, we show the estimates (\ref{eq:interm2}). 
\begin{lemma}
\label{intermediate_lemma2}
As in the proof of Lemma \ref{lm:tildeVN}, we let $\xi_s = e^{A^*\Theta-\hc}\Omega$, for all $s \in [0;1]$, $\theta = 1-\gamma$, $V_N (x) = N^{2-2\kappa} V (N^{1-\kappa}x)$. For any $\epsilon>0$, we have 
\begin{equation}
\label{intermediate_bounds1}
    \int dxdy \, |V_N(x-y)| \|\cT_3^*(x,y) \Theta \xi_s\|^2\leq CN^{\epsilon} N^{4\kappa-1}N^{-\beta}
    \end{equation}
    \begin{equation}
    \label{intermediate_bounds2}
    \int dxdy \, |V_N(x-y)| \|[a^*(\check{\theta}_x) a^*(\check{\theta}_y) ,A]\Theta \xi_s\|^2\leq CN^{\epsilon} N^{4\kappa-1}N^{-\beta/2}
    \end{equation}
    \begin{equation}
    \label{intermediate_bounds3}
   \int dxdy \, |V_N(x-y)| \|[\cT_1^*(x,y),A] \Theta \xi_s\|^2\leq CN^{\epsilon} N^{4\kappa-1}
   \end{equation}
        \begin{equation}
    \label{intermediate_bounds4}
   \int dxdy \, |V_N(x-y)| \|\cT_1(x,y)\Theta\xi_s\|^2\leq CN^{\epsilon} N^{4\kappa-1+\beta/4}
   \end{equation}
   if $\alpha$ and then $n$ are chosen large enough. 
   \end{lemma}
   \begin{proof} 
From \eqref{eq:T3T1}, we find 
\[ \cT_3 (x,y) \cT_3^* (x,y)  = \cT_{3,2} + \cT_{3,4} + \cT_{3,6} \]
with 
\[ \begin{split} 
\cT_{3,2} (x,y) &= \sum \int dx_1 dx_2 \, \check{A}^{(3)} (x_1 - x_2 , x_1 - z_1) \check{A}^{(6)} (x_1 - x_2 , x_1 - w_1) a_{z_2}^* a_{w_2} \\
\cT_{3,4} (x,y) &= \sum \int dx_1 dx_2 dx_3  \check{A}^{(6)} (x_1 - x_2 , x_1 - z_1) \check{A}^{(6)} (x_3 - x_2 , x_3 - w_1)  a_{x_3}^* a_{z_2}^* a_{w_2} a_{x_1} \\ 
\cT_{3,6} (x,y) &= \sum  \int dx_1 dx_2 dx_3 dx_4 \check{A}^{(3)} (x_1 - x_2 , x_1 - z_1) \check{A}^{(3)} (x_3 - x_4 , x_3 - w_1)  \\ &\hspace{8cm} \times  a_{z_2}^* a_{x_3}^* a_{x_4}^* a_{x_1} a_{x_2} a_{w_2} \end{split} \]
where the sums run over the four possible choices of $(z_1, z_2), (w_1, w_2) \in \{ (x,y) , (y,x) \}$ (i.e. we either have $z_1 = x, z_2 = y$ or $z_1 = y, z_2 = x$ and similarly for $(w_1 , w_2)$). Let us denote by $\cT_{3,2}^{(0)} (x,y)$ the contribution to $\cT_{3,2} (x,y)$ associated with $z_1= w_1 =x$, $z_2 = w_2 = y$. With Cauchy-Schwarz, we can estimate 
\[ \begin{split} \Big|  \int dx dy &| V_N (x-y)| \langle \xi, \cT_{3,2}^{(0)} (x,y)  \xi \rangle \Big|  \\ \leq \; & \Big( \int dx dy dx_1 dx_2 \, |V_N (x-y)| \check{A}^{(3)} (x_1 - x_2 , x_1 - x)^2 \| a_y \xi \|^2 \Big)^{1/2} \\ &\times \Big( \int dx dy dx_1 dx_2 \, |V_N (x-y)| \check{A}^{(6)} (x_1 - x_2 , x_1 - x)^2 \| a_y \xi \|^2 \Big)^{1/2}  \\ \leq \; &C \| V_N \|_1 \| \cN^{1/2} \xi \|^2 \| A^{(6)}  \|_2 \| A^{(3)} \|_2 \leq C N^{4\kappa-1 - \beta} \end{split} \]
where we used $\| V_N \|_1 \leq C N^{-1 + \kappa}$, the bounds in Lemma \ref{lm:A3A6} (in particular, $\| A^{(3)} \|_2 , A^{(6)} \|_2  \leq C N^{(3\kappa-1)/2}$) and the estimate $\| \cN^{1/2} \xi \|^2 \leq C N^{3\kappa-1}$ from Lemma \ref{lm:boundN}. The other contributions arising from $\cT_{3,2}$ can be bounded similarly. 

Next, let us denote by $\cT_{3,4}^{(0)} (x,y)$ the contribution to $\cT_{3,4} (x,y)$ associated with $z_1 = w_1 = x$ and $z_2 = w_2 = y$. We can estimate
\[ \begin{split} \Big| \int dx dy &|V_N (x-y)| \la \Theta \xi , \cT_{3,4}^{(0)} (x,y) \Theta \xi \ra \Big| \\ \leq \; &\sum_{u, u',v,v'\in \Lambda_B}
        \int dx dy dx_1 dx_2 dx_3 
|V_N(x-y)| 
       \\
        &\hspace{.3cm}\times\abs{\check{A}^{(6)}_{u,v}(x_1-x_2,x_1-x)} \abs{\check{A}^{(6)}_{u',v'}(x_3-x_2,x_3-x)}\|a_{x_1}a_y \Theta \xi\|\|a_{x_3}a_y \Theta \xi\|
     \\ \leq \; &\sum_{u, u',v,v'  \in \Lambda_B} 
     \Big( \int dxdy dx_1dx_2dx_3 \, |V_N(x-y)| \mathbbm{1}(|x_3-x_2-u' | \leq \ell_{B}) \\ &\hspace{3cm} \times |\check{A}^{(6)}_{u,v} (x_1-x_2,x_1-x)|^2 \|a_{x_3}a_{y} \Theta \xi\|^2 \Big)^{1/2}
     \\& \hspace{1.9cm}\times \Big( \int dxdy dx_1dx_2dx_3 \, |V_N(x-y)| \mathbbm{1}(|x_1-x_2-u|\leq \ell_{B}) \\ &\hspace{3cm} \times |\check{A}^{(6)}_{u',v'}(x_3-x_2,x_3-x)|^2\|a_{x_1}a_{y}\Theta \xi\|^2 \Big)^{1/2}
     \\ \leq \; &CN^{\epsilon} \|V_N\| _1 \|\cN^{1/2}\xi\|^2\Big(\sum_{u  \in \Lambda_B} \norm{\check{A}^{(6)}_{u,v}}_2 \Big)^2\leq C N^{4\kappa-1-\beta+\epsilon}. 
    \end{split}
\] 
The other contributions arising from $\cT_{3,4}$ can be handled similarly. Also the contributions from $\cT_{3,6}$ can be estimated analogously. 

Next, we consider the bound \eqref{intermediate_bounds2}. Similarly to (\ref{eq:T3T1}), we have 
\begin{equation}\label{eq:inter-comm}  [a^* (\check{\theta}_x) a^* (\check{\theta}_y),A]= - \int dr dt \;\check{\theta} (r) \check{\theta} (t) \big(\cT_1 (x-r,y-t)+ \cT_3 (x-r,y-t)\big). \end{equation} 
Therefore, we find 
\begin{equation}\label{eq:tT13} \begin{split}  \int dx dy \, |V_N &(x-y)| \| \big[ a^* (\check{\theta}_x) a^* (\check{\theta}_y) , A \big] \Theta\xi_s \|^2 \\ &\leq \sum_{i=1,3} \int dx dy \, |V_N (x-y)|  \Big\|  \int dr dt \, \check{\theta} (r) \check{\theta} (t) \cT_i (x-r, y-t) \Theta \xi_s \Big\|^2 .\end{split} \end{equation} 
To bound the term with $i=3$, we observe that 
\[ \begin{split} 
 \Big\| &\int dr dt \, \check{\theta} (x-r) \check{\theta} (y-t) \cT_3 (r,t) \Theta \xi_s \Big\|^2 \\ = \; & \int 
  dr dt dr' dt' \check{\theta} (x-r) \check{\theta} (y-t) \check{\theta} (x-r') \check{\theta} (y-t')  \sum_{\substack{(z_1,z_2)\in \{(r,t),(t,r)\},\\ (w_1,w_2)\in \{(r',t'),(t',r')\}}} 
  \int dx_1dx_2dx_3dx_4 \,  \\ & \times \check{A}^{(3)} (x_1-x_2,x_1-z_1)\check{A}^{(3)} (x_3-x_4,x_3-w_1) \la \Theta \xi_s ,a^*_{z_2}a^*_{x_3}a^*_{x_4}a_{x_1}a_{x_2}a_{w_2} \Theta \xi_s \ra 
   \\ &+ \int dr dt dr' \check{\theta} (x-r) \check{\theta} (y-t) \sum_{\substack{(z_1,z_2)\in \{(r,t),(t,r)\}}}
 \big(\check{\theta} (x-r') \check{\theta} (y-z_2) +\check{\theta} (x-z_2) \check{\theta} (y-r') \big)
 \\ &\times \int dx_1 dx_2 dx_3 dx_4 \check{A}^{(3)} (x_1-x_2,x_1-z_1)\check{A}^{(3)} (x_3-x_4,x_3-r' ) \la \Theta \xi_s ,a^*_{x_3} a^*_{x_4}a_{x_1}a_{x_2} \Theta \xi_s \ra  \\ = \; &\la \Theta \xi_s , \tilde{\cT}_{3,6}(x,y) \Theta \xi_s \ra + \la \Theta \xi_s , \tilde{\cT}_{3,4} (x,y) \Theta \xi_s \ra.
  \end{split} \]
Let us denote by $\tilde{\cT}_{3,6}^{(0)} (x,y)$ the contribution to $\tilde{\cT}_{3,6} (x,y)$ associated with the choice $z_1 = r, w_1 = r', z_2 = t, w_2 = t'$. Localizing the kernels $\check{\theta}$ and $\check{A}^{(3)}$ and applying Cauchy-Schwarz, we can estimate 
\[ \begin{split} 
\Big| \int dx dy \, &|V_N (x-y)| \la \Theta \xi_s , \tilde{\cT}^{(0)}_{3,6} (x,y) \Theta \xi_s \rangle \Big| \\ \leq \; &\sum_{u, u',v,v',w,w',m,m' \in \Lambda_B}\Big( \int dr dt dr' dt' dx_1dx_2dx_3dx_4 dxdy \, \, |V_N(x-y)|  \\ &\hspace{1.3cm}\times\check{\theta}_{u}(x-r)^2 \check{\theta}_{u'} (y-t)^2 \, \check{A}^{(3)}_{m',w'} (x_3-x_4,x_3-r')^2 \\&\hspace{1.3cm}\times  
 \mathbbm{1}_{B_v} (x-r') \mathbbm{1}_{B_{v'}}(y-t') \mathbbm{1}_{B_{m}} (x_1 - x_2)  \mathbbm{1}_{B_{w}}(x_1-r)  \, \|a_{x_1}a_{x_2}a_{t'} \Theta \xi_s\|^2\Big)^{1/2}
    \\&\hspace{.5cm}\times \Big( \int dr dt dr' dt' dx_1dx_2dx_3dx_4 dxdy\;  |V_N(x-y)| 
     \\&\hspace{1.3cm}\times  \check{\theta}_{v} (x-r')^2\check{\theta}_{v'} (y-t')^2 \check{A}^{(3)}_{m,w} (x_1-x_2,x_1-r)^2   
   \\&\hspace{1.3cm}\times \mathbbm{1}_{B_u}(x-r)\mathbbm{1}_{B_{u'}}(y-t) \mathbbm{1}_{B_{m'}}(x_3-x_4)\mathbbm{1}_{B_{w'}} (x_3-r') \|a_{x_4}a_{x_3}a_{t} \Theta \xi_s \|^2\Big)^{1/2}
    \\\leq \; &CN^{\epsilon} \ell_{B}^6 \|V_N\| _1 \|\cN^{1/2}\xi_s\|^2 \Big(\sum_{u,v  \in \Lambda_B} \norm{\check{A}^{(3)}_{u,v}}_2 \Big)^2 \Big(\sum_{u  \in \Lambda_B} \norm{ \check{\theta}_u}_2 \Big)^4 \leq C N^{4\kappa - 1 - \beta/2 + \epsilon}.
    \end{split}
\]
where we used $\| V_N \|_1 \leq C N^{-1 + \kappa}$, $\sum_u \| \check{\theta}_u \|_2 \leq N^{3\kappa/4+\epsilon}$ from Lemma \ref{lemma:sumlocalizedsigma}, $\sum_{u,v} \| A^{(3)}_{u,v} \|_2 \leq C N^{(3\kappa-1)/2}$ from Lemma \ref{lm:A3A6} and the estimate $\| \cN^{1/2} \xi \|^2 \leq C N^{3\kappa-1}$ from Lemma \ref{lm:boundN}. The other contributions arising from $\tilde{\cT}_{3,6} (x,y)$ can be controlled analogously. 

As for the term $\tilde{\cT}_{3,4} (x,y)$, we consider the contribution $\tilde{\cT}^{(0)}_{3,4} (x,y)$, associated with $z_1 =r$, $z_2 = t$, proportional to $\check{\theta} (x-r) \check{\theta} (x-r') \check{\theta} (y-t)^2$. Localizing all kernels $\check{\theta}, \check{A}^{(3)}$, we find 
\[ \begin{split}\Big| \int &dx dy |V_N (x-y)| \langle \Theta \xi_s, \tilde{\cT}^{(0)}_{3,4} (x,y) \Theta \xi_s \rangle \Big| \\ \leq \; &\sum_{u, u',v,v',w,w',m  \in \Lambda_B} 
     \Big( \int dxdy dx_1dx_2dx_3 dr dt dr' \, | V_N(x-y) | \check{\theta}_m (y-t)^2\check{\theta}_w (x-r)^2 \\ &\hspace{.5cm} \times \check{A}^{(3)}_{u,v} (x_3-x_4,x_3-r')^2 \mathbbm{1}_{B_{w'}}(x-r')     \mathbbm{1}_{B_{v'}}(x_1-r) \mathbbm{1}_{B_{u'}}(x_1-x_2) \| a_{x_1} a_{x_2} \Theta \xi\|^2 \Big)^{1/2}
     \\&\hspace{1.5cm}\times \Big( \int dxdy dx_1dx_2dx_3dr dt dr' \, |V_N (x-y)|    \check{\theta}_m (y-t)^2 \check{\theta}_{w'} (x-r')^2\\ &\hspace{.5cm} \times  \check{A}^{(3)}_{u',v'} (x_1-x_2,x_1-r)^2  \mathbbm{1}_{B_{w}}(x-r)     \mathbbm{1}_{B_{v}}(x_3-r') \mathbbm{1}_{B_{u}}(x_3-x_4) \|a_{x_3} a_{x_4} \Theta \xi\|^2 \Big)^{1/2}
          \\&\leq CN^{\epsilon} \ell_{B}^6 \|V_N\| _1 \|\cN^{1/2}\xi\|^2\Big(\sum_{u,v  \in \Lambda_B} \norm{ \check{A}^{(3)}_{u,v}}_2 \Big)^2 \Big(\sum_{u  \in \Lambda_B} \norm{ \check{\theta}_{u}}_2 \Big)^4\leq C N^{4\kappa-1-\beta/2+\epsilon} .
  \end{split} \]
The other contributions arising from $\tilde{\cT}_{3,4} (x,y)$ can be handled similarly. As for the term with $i=1$ on the r.h.s of (\ref{eq:tT13}), we compute 
\[ \begin{split} \Big\| \int dr dt \; &\check{\theta} (x-r) \check{\theta} (y-t) \cT_1 (r,t) \Theta \xi_s \Big\|^2  = \\ =&\int dr dt dr' dt' dx_1dx_2 \, \check{\theta} (x-r) \check{\theta} (y-t) \check{\theta} (x-r') \check{\theta} (y-t') \\ &\hspace{4cm} \times \check{A}^{(6)} (x_1-t,x_1-r)\check{A}^{(6)} (x_2-t',x_2-r') \la \xi ,a^*_{x_1}a_{x_2}\xi\ra.
    \end{split} \]
Hence, proceeding as we did above to control the terms $\tilde{\cT}_{3,6}^{(0)}, \tilde{\cT}_{3,4}^{(0)}$, we find  
\[ \begin{split} 
\int & dx dy |V_N (x-y)| \Big\| \int dr dt \; \check{\theta} (x-r) \check{\theta} (y-t) \cT_1 (r,t) \Theta \xi_s \Big\|^2 \\ 
\leq \; &C \ell_B^6 \| V_N \|_1 \| \cN^{1/2} \xi_s \|^2 \Big( \sum_{u,v \in \Lambda_B} \| A^{(6)}_{u,v} \|_2 \Big)^2 \Big( \sum_{w \in \Lambda_B} \| \check{\theta}_w \|_2 \Big)^4 \leq C N^{4\kappa-1-\beta/2+\epsilon} .
\end{split} \] 

Next, we show \eqref{intermediate_bounds3}. From (\ref{eq:T3T1}), we find 
\begin{equation}
\label{intermediate_expression1}
    [\cT^*_1(x,y),A] =  \int dz dx_1 dx_2 \, \check{A}^{(6)} (z-y,z-x)\check{A}^{(3)} (x_1-x_2,x_2-z) a_{x_1}a_{x_2}.
\end{equation} 
Hence 
\begin{equation}
\begin{split}
    \|[\cT_1^* &(x,y) ,A] \Theta \xi_s\|^2\\ =&\; \int dz dz' dx_1 dx_2 dx_3 dx_4 \, \check{A}^{(6)} (z-y,z-x) \check{A}^{(3)} (x_1-x_2,x_2-z) \\&\hspace{1cm}\times \check{A}^{(6)} (z'-y,z'-x) \check{A}^{(3)} (x_3-x_4,x_4-z') \la \Theta \xi, a^*_{x_3}a^*_{x_4}a_{x_1}a_{x_2} \Theta \xi\ra.
    \end{split}
\end{equation}
Therefore, setting $D^{(6)} (x) = \sup_s |A^{(6)} (x,s)|$ and $D_u^{(6)} (x) = D^{(6)} (x) \mathbbm{1}_{B_u} (x)$ for $u \in \Lambda_B$, we obtain 
\[ \begin{split} 
\int &dx dy \, |V_N (x-y)|  \|[\cT_1^* (x,y) ,A] \Theta \xi_s\|^2 \\ \leq \; &\sum_{w,u,v,w',u',v' \in \Lambda_B} \Big( \int dx dy dz dz' dx_1 dx_2 dx_3 dx_4 \, |V_N (x-y)| D_w^{(6)} (z-y)^2 \\ &\hspace{.3cm} \times  \check{A}^{(3)}_{u',v'} (x_3 - x_4 , x_4 - z')^2 \mathbbm{1}_{B_{w'}} (z'-y) \mathbbm{1}_{B_u} (x_1 - x_2) \mathbbm{1}_{B_v} (x_2 -z)  \| a_{x_1} a_{x_2} \Theta \xi_s \|^2 \Big)^{1/2} \\  
&\times \Big( \int dx dy dz dz' dx_1 dx_2 dx_3 dx_4 \, |V_N (x-y)| D_{w'}^{(6)} (z'-y)^2 \\ &\hspace{.3cm} \times \check{A}^{(3)}_{u,v} (x_1 - x_2 , x_2 - z)^2 \mathbbm{1}_{B_{w}} (z-y) \mathbbm{1}_{B_{u'}} (x_3 - x_4) \mathbbm{1}_{B_{v'}} (x_4 -z')  \| a_{x_3} a_{x_4} \Theta \xi_s \|^2 \Big)^{1/2} \\
\leq \; &CN^\epsilon \ell_B^6 \| V_N \|_1 \| \cN^{1/2} \xi_s \|^2 \, \Big( \sum_{u,v} \| A_{u,v}^{(3)} \|_2 \Big)^2 \Big(\sum_w \| D^{(6)}_w \|_2 \Big)^2 \leq C N^{4\kappa-1 + \epsilon} \end{split} \] 
where we used the bounds in Lemma \ref{lm:A3A6} and Lemma \ref{lm:boundN} to estimate $\| \cN^{1/2} \xi_s \|^2 \leq C N^{3\kappa-1}$. 

Finally, we consider \eqref{intermediate_bounds4}. From (\ref{eq:T3T1}), setting again $D^{(6)} (x) = \sup_s |A^{(6)} (s,x)|$, we find 
\[ \begin{split} \int dx dy \, & |V_N (x-y)|  \| \cT_1 (x,y) \Theta \xi_s \|^2 \\ &= \int dx dy dz dz' \, |V_N (x-y)| A^{(6)} (z-y, z-x) A^{(6)} (z'-y, z' -x) \langle a_z \Theta \xi_s , a_{z'} \Theta \xi_s \rangle \\  &\leq \int dx dy dz dz' \, |V_N (x-y)| D^{(6)} (z-x) D^{(6)} (z' -x)  \| a_z \Theta \xi_s \| \| a_{z'} \Theta \xi_s \| .  \end{split} \] 
Localizing $D^{(6)}$ and applying Cauchy-Schwarz, we arrive at
\[ \begin{split} \int dx &dy \, |V_N (x-y)|  \| \cT_1 (x,y) \Theta \xi_s \|^2  \\ \leq \; & \sum_{u,u' \in \Lambda_B} \Big( \int dx dy dz dz' \, |V_N (x-y)| |D_u^{(6)} (z-x)|^2 \mathbbm{1}_{B_{u'}} (z'-x) \| a_{z'} \Theta \xi_s \|^2 \Big)^{1/2} \\ &\hspace{.5cm} \times  \Big( \int dx dy dz dz' \, |V_N (x-y)| |D_{u'}^{(6)} (z'-x)|^2 \mathbbm{1}_{B_{u}} (z-x) \| a_{z} \Theta \xi_s \|^2 \Big)^{1/2} \\ \leq \; &C \ell_B^3  \| V_N \|_1 \Big( \sum_u \| D^{(6)}_u \|_2 \Big)^2 \| \cN^{1/2} \xi_s \|^2  \leq C N^{4\kappa-1+\beta/4+\epsilon}. \end{split} \] 
\end{proof} 

Finally, we prove (\ref{eq:cT1-cc}). 
\begin{lemma} \label{lm:interm3}
As in the proof of Lemma \ref{lm:tildeVN}, we let $\xi_s = e^{A^*\Theta-\hc}\Omega$, for all $s \in [0;1]$, $\theta = \gamma-1$, $V_N (x) = N^{2-2\kappa} V (N^{1-\kappa}x)$. Then 
\begin{equation*} \Big| \int dx dy \, V_N (x-y) \, \la \Theta\xi_s,  \big[ \cT_1 (x,y) , [ a (\check{\theta}_x) a (\check{\theta}_y) , A^*] \big] \Theta \xi_s \rangle \Big| \leq C N^{4\kappa-1} \, .\end{equation*} 
\end{lemma} 
\begin{proof} 
From (\ref{eq:inter-comm}), we find 
\begin{equation}\label{eq:lastt1} \begin{split} \int dx &dy \, V_N (x-y) \, \la \Theta\xi_s,  \big[ \cT_1 (x,y) , [ a (\check{\theta}_x) a (\check{\theta}_y) , A^*] \big] \Theta \xi_s \rangle \\ &= \sum_{j=1,3} \int dx dy dr dt \, \check{\theta} (x-r) \check{\theta} (y-t)  \la \Theta \xi_s, \big[ \cT_1 (x,y) , \cT_j^* (r , t)\big] \Theta \xi_s \rangle  .\end{split} \end{equation} 
Let us first focus on $j=3$. With (\ref{eq:T3T1}), we compute 
\[ \begin{split}  \big[ \cT_1 &(x,y) , \cT_3^* (r , t)\big] \\ &= \int dx_1 dx_2 \, \sum_{\substack{(z_1, w_1) \in \{ (x_1,x_2) , (x_2 ,x_1) \} \\ (z_2 ,w_2) \in \{ (r ,t) , (t, r) \}}} A^{(6)} (z_1 - y, z_1 - x) A^{(3)} (x_1 -x_2, x_1 - z_2) a_{w_1}^* a_{w_2}. \end{split}  \]
The contribution of the term $[\cT_1 (x,y), \cT^*_3 (r,t)]^{(0)}$, associated with $z_1 =x_1, w_1 = x_2, z_2 = r, w_2 = t$, to (\ref{eq:lastt1}) can be bounded, localizing the kernels $\check{\theta}$ and $A^{(3)}$ and introducing $D^{(6)} (x) = \sup_s A^{(6)} (s, x)$, by  
\[ \begin{split} \Big| \int dx &dy \, V_N (x-y) \, \la \Theta\xi_s, [\cT_1 (x,y), \cT^*_3 (r,t)]^{(0)} \Theta \xi_s \rangle \Big| \\ = \; & \Big| \int dx dy dr dt dx_1 dx_2 V_N (x-y) \check{\theta} (x-r) \check{\theta} (y-t) \\ &\hspace{3cm} \times A^{(6)} (x_1 - y, x_1 -x) A^{(3)} (x_1 -x_2 , x_1- r) \langle \Theta \xi_s, a_{x_2}^* a_t \Theta \xi_s \rangle \Big| \\ 
 \leq \; &\sum_{u,v,w,w' \in \Lambda_B} \Big( \int dx dy dr dt dx_1 dx_2 \, |V_N (x-y)| \check{\theta}_w (x-r)^2 \check{\theta}_{w'} (y-t)^2 \\ &\hspace{7cm} \times A^{(3)}_{u,v} (x_1 - x_2 , x_1 - r)^2 \| a_t \Theta \xi_s \|^2 \Big)^{1/2} \\ &\times \Big( 
 \int dx dy dr dt dx_1 dx_2 \, |V_N (x-y)| D^{(6)} (x_1 - x)^2 \\ &\hspace{2cm} \times  \mathbbm{1}_{B_{w'}} (y-t) \mathbbm{1}_{B_u} (x_1 - x_2) \mathbbm{1}_{B_v} (x_1 - r)  \| a_{x_2} \Theta \xi_s \|^2 \Big)^{1/2} \\ \leq \; &C \ell_B^{9/2} \| V_N \|_1 \| \cN^{1/2} \xi_s \|^2 \| D^{(6)} \|_2 \Big( \sum_{u,v} \| A^{(3)}_{u,v} \|_2 \Big) \Big( \sum_w \| \check{\theta}_w \|_2 \Big)^2 \leq C N^{4\kappa-1-\beta/8+ \epsilon}.  \end{split} \]
The other contributions arising from the commutator $[\cT_1 (x,y) , \cT^*_3 (r,t)]$ can be bounded similarly. As for the term on the r.h.s. of (\ref{eq:lastt1}) with $j=1$, we compute
\[ [\cT_1 (x,y), \cT^*_1 (r,t) ] = \int dz \, \check{A}^{(6)} (z-y, z-x) \check{A}^{(6)} (z-t, z-r). \]
The contribution of this term to (\ref{eq:lastt1}) can be estimated switching to Fourier space. We find 
\[ \begin{split} \Big| \int dx dy &dr dt dz V_N (x-y) \check{\theta} (x-r) \check{\theta} (y-t) \check{A}^{(6)} (z-y, z-x) \check{A}^{(6)} (z-t, z-r) \Big| \\ & = \Big| \sum_{p,q,k \in \Lambda^*} \hat{V}_N (k)  \theta_q \theta_p A^{(6)}_{q-k, p+k} A^{(6)}_{q,p} \Big| \leq CN^{-1+\kappa} \sum_{k,p,q \in \Lambda^*}  |\theta_q | | \theta_p| |A^{(6)}_{q-k,p+k}| |A^{(6)}_{p,q}|. \end{split} \]
With the bounds in Lemma \ref{lm:A3A6}, 
we conclude that 
\[  \Big| \int dx dy dr dt dz V_N (x-y) \check{\theta} (x-r) \check{\theta} (y-t) \check{A}^{(6)} (z-y, z-x) \check{A}^{(6)} (z-t, z-r) \Big| \leq C N^{4\kappa-1}. \]
\end{proof}

\subsection{Bounds for relevant cubic expectations} 
\label{subsec:cubic2} 

In this subsection, we conclude the estimate for the energy of the trial state $\psi_N = W_{N_0} e^{B^* - B} \xi$, proving bounds for the expectations $\la \xi, (\cC_N^* + \cC_N) \xi \ra$, $\la \xi, \cV_N \xi \ra$ and $\la \xi, \cK \xi \ra$, with $\xi  = e^{A^* \Theta -\hc} \Omega$. 
\begin{lemma}
\label{lemma:C_N}
In position space, we write the operator $\cC_N$ from \eqref{eq:lmHN1} as  
\begin{equation}\label{eq:CN-pos} 
\cC_N=
         N_0^{1/2}\int dxdydz N^{2-2\kappa}V(N^{1-\kappa}(x-y))\check{\sigma}(x-z)a^*_xa^*_ya^*_z . 
       \end{equation}
Recalling the definition (\ref{eq:A3A6}) of the kernel $\check{A}^{(6)}$, we find 
\begin{equation}
\label{eq_Cn_f}
\begin{split} 
 \la e^{A^*\Theta-\hc}\Omega, &(\cC_N+\cC_N^*) e^{A^*\Theta-\hc}\Omega\ra \\ &= 2N_0^{1/2}\int dxdy N^{2-2\kappa}V(N^{1-\kappa}x)\check{\sigma} (y) \check{A}^{(6)}(x,y) + O (N^{4\kappa-1}) . 
 \end{split} 
 \end{equation}
\end{lemma}

\begin{proof} 
Expanding $e^{-A^* \Theta+\hc} a_x^* a_y^* e^{A^* \Theta - \hc}$, recalling the notation $\xi_s = e^{s A^* \Theta- \hc} \Omega$, $V_N (x) = N^{2-2-\kappa} V (N^{1-\kappa} x)$ and the definition (\ref{eq:T3T1}) of the operators $\cT_3 (x,y), \cT_1 (x,y)$, we find 
\begin{equation} \label{eq:CN1} \begin{split} 
&\la \xi , \cC_N \xi \ra \\ &= N_0^{1/2}  \int_0^1 ds \, \int dx dy dz \, V_N (x-y) \check{\sigma} (x-z)  \la \xi_s, [A, a_x^* a_y^*] e^{-(1-s) A^* \Theta + \hc} a_z^* \xi \ra + \cE_\Theta \\
&= N_0^{1/2}  \int_0^1 ds \, \int dx dy dz \, V_N (x-y) \check{\sigma} (x-z)  \la \xi_s , \cT_1 (x,y) e^{-(1-s) A^* \Theta + \hc} a_z^* \xi \ra  \\ &\hspace{.2cm} + N_0^{1/2}  \int_0^1 ds \, \int dx dy dz \, V_N (x-y) \check{\sigma} (x-z)  \la \Theta \xi_s ,\cT_3 (x,y) e^{-(1-s) A^* \Theta + \hc} a_z^* \xi \ra + \cE_\Theta \\ &=: S_1 + S_3 + \cE_\Theta \end{split} \end{equation} 
where the error term $\cE_\Theta$ includes the contribution of all terms proportional to $1-\Theta$ (to handle these terms and show that $\pm \cE_\Theta \leq C N^{-\nu}$ for arbitrary $\nu > 0$, we can proceed as in (\ref{eq:insTheta}), choosing $\alpha > 0$ and $n \in \bN$ large enough). We can bound
\begin{equation} \label{eq:S31} \begin{split} |S_3| \leq \; & C N^{1/2}  \int_0^1 ds \Big( \int dx dy |V_N (x-y)| \| \cT_3^* (x,y) \Theta \xi_s \|^2 \Big)^{1/2}   \\ &\hspace{4cm} \times \Big( \int dx dy \, |V_N (y)| \Big\| \int dz \, \check{\sigma} (x-z) a_z^* \xi \Big\|^2 \Big)^{1/2} .\end{split} \end{equation}  
From Lemma \ref{intermediate_lemma2}, we find 
\[  \int dx dy \, |V_N (x-y)| \| \cT_3^* (x,y) \Theta \xi_s \|^2 \leq C N^{4\kappa-1-\beta+\epsilon} \]
for any $\epsilon > 0$. On the other hand, we have 
\begin{equation}\label{eq:Kastar} \int dx \, \Big\| \int dz \, \check{\sigma} (x-z) a_z^* \xi_s  \Big\|^2 = \| \check{\sigma} \|_2^2 + \| \tilde{\cN}^{1/2} \xi _s\|^2 \leq C N^{3\kappa/2}  \end{equation}
for all $s \in [0;1]$, by Lemma \ref{lemma:propertiesquadratickernels} and Lemma \ref{lemma:N_g}. With $\| V_N \|_1 \leq C N^{-1+\kappa}$, we conclude that 
\[ |S_3| \leq C N^{4\kappa-1} N^{-\beta/4 + \epsilon} \] 
for any $\epsilon > 0$. 

As for the term $S_1$, we expand 
 \[ \begin{split} 
 S_1 = \; &N_0^{1/2} \int_0^1 ds \int dx dy dz dz' \, V_N (x-y) \check{\sigma} (x-z) \check{A}^{(6)} (z' - y , z' - x)  \\ &\hspace{5cm} \times \la \xi, e^{(1-s) A^* \Theta -\hc} a_{z'} e^{-(1-s) A^* \Theta + \hc} a_z^* \xi \ra \\ 
 = \; &N_0^{1/2}  \int dx dy dz dz' \, V_N (x-y) \check{\sigma} (x-z) \check{A}^{(6)} (z' - y , z' - x)  \la \xi, a_{z'}  a_z^* \xi \ra \\
 &+N_0^{1/2} \int_0^1 ds \int_0^{1-s} dt \int dx dy dz dz' \, V_N (x-y) \check{\sigma} (x-z) \check{A}^{(6)} (z' - y , z' - x)  \\ &\hspace{5cm} \times \la \xi, e^{t A^* \Theta -\hc} [ A^* , a_{z'} ] e^{-t A^* \Theta + \hc} a_z^* \xi \ra
 \end{split} \]
 and thus
  \[ \begin{split} 
 S_1 = \; &N_0^{1/2} \int dx dy dz \, V_N (x-y) \check{\sigma}  (x-z) \check{A}^{(6)} (z - y , z - x) \\ &+ N_0^{1/2}  
\int dx dy dz dz' \, V_N (x-y) \check{\sigma} (x-z) \check{A}^{(6)} (z' - y , z' - x)  \la \xi,a_z^*  a_{z'}  \xi \ra \\
 &+ N_0^{1/2}  \int_0^1 ds \int_0^{1-s} dt \int dx dy dz dz' \, V_N (x-y) \check{\sigma} (x-z) \check{A}^{(6)} (z' - y , z' - x)  \\ &\hspace{5cm} \times \la \xi_{1-t}, [ A^* , a_{z'} ] e^{-t A^* \Theta + \hc} a_z^* e^{t A^* \Theta-\hc} \xi_{1-t} \ra.
 \end{split} \]
Expanding $e^{-t A^* \Theta + \hc} a_z^* e^{tA^* \Theta-\hc}$ and using the symmetry (\ref{eq:symmAxy}), we conclude that
\begin{equation}\label{eq:S1-4} \begin{split} 
 S_1 = \; &N_0^{1/2} \int dx dy \, V_N (y) \check{\sigma} (x) \check{A}^{(6)} (y , x) \\ &+ N_0^{1/2} 
\int dx dy dz dz' \, V_N  (y) \check{\sigma} (x-z) \check{A}^{(6)} (y , x-z')  \la \xi,a_z^*  a_{z'}  \xi \ra \\
 &+ N_0^{1/2} \int_0^1 ds \int_0^{1-s} dt \int dx dy dz dz' \, V_N (y) \check{\sigma} (x-z) \check{A}^{(6)} (y , x-z')  \\ &\hspace{9cm} \times \la \xi_{1-t}, [ A^* , a_{z'} ] a_z^* \xi_{1-t} \ra \\ 
 &+ N_0^{1/2}   \int_0^1 ds \int_0^{1-s} dt \int_0^{t} d\tau \int dx dy dz dz' \, V_N (y) \check{\sigma} (x-z) \check{A}^{(6)} (y , x-z')  \\ &\hspace{6cm} \times \la \xi_{1-t}, [ A^* , a_{z'} ]  e^{-\tau A^* \Theta + \hc} [ A, a_z^* ] \xi_{1-t +\tau} \ra \\ =: \; &S_{1,1} + S_{1,2} + S_{1,3} + S_{1,4}.
 \end{split} \end{equation} 
The term $S_{1,1}$ produces the main term on the r.h.s. of (\ref{eq_Cn_f}). To control $S_{1,2} , S_{1,3} , S_{1,4}$, we need bounds for 
 \[ \begin{split}  \int dx \Big\| \int dz' \, \check{A}^{(6)} (y, x-z')  a_{z'}  \xi_s \Big\|^2 &= \| \cN^{1/2}_{A_y \star A_y} \xi_s \|^2 \\  \int dx \Big\| \int dz \, \check{A}^{(6)} (y,x-z)  \big[ a_z^* , A \big] \xi_s \Big\|^2 &= \| \ddot{\cN}^{1/2}_{A_y \star A_y} \xi_s \|^2 \end{split} \]
where we set $A_y (x) = \check{A}^{(6)} (y,x)$. To obtain such bounds, we use Lemma \ref{lemma:N_g} (and the subsequent remark). We have 
\[ \begin{split}  \langle \Omega, &[ A, [ \cN_{A_y \star A_y} , A^* ] ] \Omega \rangle = \int dx_1 dx_2 dx_3 dz dz' \, \check{A}^{(6)} (y, z' - z) \check{A}^{(6)} (y, z' - x_3) \\ &\hspace{6cm} \times  \check{A}^{(3)} (x_1 - x_2, x_1 - x_3) \check{A}^{(6)} (z-x_1 , z-x_2) .\end{split} \]
Switching to Fourier, we find 
\[  \langle \Omega, [ A, [ \cN_{A_y \star A_y} , A^* ] ] \Omega \rangle = \sum_{p,q,r,s \in \Lambda^*} A^{(6)}_{r,q} A^{(6)}_{s,q} A^{(3)}_{p,q} A^{(6)}_{p,q} e^{i (r+s) \cdot y} .\]
With Lemma \ref{lemma:propertiesquadratickernels} and Lemma \ref{lm:A3A6}, we conclude that  
\[ \begin{split} | \langle \Omega, [ A, [ \cN_{A_y \star A_y} , A^* ] ] \Omega \rangle | &\leq \sum_{p,q,r,s} 
|A^{(6)}_{r,q}| |A^{(6)}_{s,q}| | A^{(3)}_{p,q}| | A^{(6)}_{p,q} | \\ &\leq N \sum_{j=3,6} \sum_{p,q \in \Lambda^*} (1+ \sigma^2_q) |A_{p,q}^{(j)}|^2  \leq C N^{3\kappa} \end{split} \]
uniformly in $y$. 

On the other hand, setting $D^{(6)} (x) = \sup_y \abs{\check{A}^{(6)} (y,x)}$ and $D_u^{(6)} (x)= D^{(6)} (x) \mathbbm{1}_{B_u} (x)$ for every $u \in \Lambda_B$, we observe that, with Lemma \ref{lm:A3A6},  
\[ \sup_y  \sum_{u \in \Lambda_B} \| (A_y)_u \|_2  \leq \sum_{u \in \Lambda_B} \| D^{(6)}_u \|_2 \leq C N^{3\kappa/4 + 1/2 + \epsilon} \]
for any $\epsilon > 0$. Hence, Lemma \ref{lemma:N_g} (and the remark after it) implies that 
\[ \begin{split}  \sup_y \| \cN_{A_y \star A_y}^{1/2} \xi_s \|^2 &\leq C N^{3\kappa} (1 + N^{-\beta/8+\epsilon}) \\  \sup_y \| \ddot{\cN}_{A_y \star A_y}^{1/2} \xi_s \|^2 &\leq C N^{3\kappa - \beta/8 + \epsilon} \end{split} \]
for every $\epsilon > 0$. Choosing $\epsilon > 0$ small enough, we conclude that
\begin{equation}\label{eq:NaNa}   \sup_y \| \cN_{A_y \star A_y}^{1/2} \xi_s \| ,  \; \sup_y \| \ddot{\cN}_{A_y \star A_y}^{1/2} \xi_s \| \leq C N^{3\kappa/2} \end{equation}  
for all $s \in [0;1]$. With (\ref{eq:NaNa}) and with the bounds
 \[ \begin{split}  \int dx \Big\| \int dz \, \check{\sigma} (x-z) a_z \xi_s \|^2 = \| \cN_{\check{\sigma} \star \check{\sigma}}^{1/2} \xi_s \| = \| \tilde{\cN}^{1/2} \xi_s \| \leq C N^{(3\kappa-1)/2} \\  \int dx \Big\| \int dz \, \check{\sigma} (x-z) [ a^*_z , A]  \xi_s \|^2 = \| \ddot{\cN}_{\check{\sigma} \star \check{\sigma}}^{1/2} \xi_s \| \leq C N^{(3\kappa-1)/2} \end{split} \] 
which also follow from Lemma \ref{lemma:N_g} and the subsequent remark, we can now estimate the r.h.s. of (\ref{eq:S1-4}). We have    
\begin{equation} \label{eq:S12} \begin{split} |S_{1,2}| \leq \; &C N^{1/2}  \int dx dy |V_N (y)| \,\Big\| \int dz \, \check{\sigma} (x-z) a_z \xi \Big\| \Big\| \int dz' \, \check{A}^{(6)}_y (x-z') a_{z'} \xi \Big\| 
 \\ \leq \; &C N^{1/2} \| V_N \|_1 \,  \| \tilde{\cN}^{1/2} \xi \|  \sup_y \| \cN^{1/2}_{A_y \star A_y} \xi \| \leq C N^{4\kappa -1}  \end{split} \end{equation} 
and, similarly,
\begin{equation}\label{eq:S14}  \begin{split} |S_{1,4}| \leq \; &C N^{1/2} \| V_N \|_1 \int_0^1 ds \int_0^{1-s} dt  \int_0^{t} d\tau \,  \| \ddot{\cN}^{1/2}_{\check{\sigma} \star \check{\sigma}} \xi_{1-t+\tau} \| \, \sup_y \| \ddot{\cN}^{1/2}_{A_y \star A_y} \xi_{1-t} \|  \\
 \leq \; &C N^{4\kappa-1} . \end{split} \end{equation} 
 
To estimate $S_{1,3}$, we expand it once again, writing  
\begin{equation}\label{eq:S130} \begin{split}  S_{1,3} = \; & N_0^{1/2}   \int_0^1 ds \int_0^{1-s} dt \int_0^{1-t} d\tau \int dx dy dz dz' \, V_N (y) \check{\sigma} (x-z) \check{A}^{(6)} (y, x-z') \\ &\hspace{3.5cm} \times \big( \la \xi_{\tau} , [A^*, a_{z'} ] [ A, a_z^* ]  \xi_\tau \ra +   \la \xi_{\tau} , [A, [A^*, a_{z'} ]]  a_z^*  \xi_\tau \ra \big) \\ =:\; & S_{1,3,1} + S_{1.3.2} \end{split} \end{equation} 
We find  
\[ |S_{1,3,1}| \leq C N^{1/2} \| V_N \|_1 \int_0^1 d\tau \, \| \ddot{\cN}^{1/2}_{\check{\sigma} \star \check{\sigma}} \xi_\tau \| \, \sup_y \| \ddot{\cN}^{1/2}_{A_y \star A_y} \xi_\tau \|  \leq C N^{4\kappa-1} . \]
As for $S_{1,3,2}$, we write
\begin{equation}\label{eq:T04}
S_{1,3,2} = T_0 + T_2 + T_4 + \cE_\Theta \end{equation} 
with 
\begin{equation*} \begin{split} T_0 = \; &N_0^{1/2}  \int_0^1 ds \int_0^{1-s} dt \int_0^{1-t} d\tau \, \int dx dy dz dz' dx_1 dx_2  \, V_N (y) \check{\sigma} (x-z) \check{A}^{(6)} (y, x-z') \\ &\hspace{5.5cm} \times \check{A}^{(3)} (x_1 - x_2 , x_1 -z')  \check{A}^{(6)} (x_1 - x_2 , x_1 - z)  \\
T_2 = \; &N_0^{1/2}   \int_0^1 ds \int_0^{1-s} dt \int_0^{1-t} d\tau  \int dx dy dz dz' dx_1 dx_2 dx_3 \, V_N (y) \check{\sigma} (x-z)  \\ &\hspace{.3cm} \times \check{A}^{(6)} (y, x-z') \check{A}^{(6)} (x_1 - x_2, x_1 - z') \check{A}^{(6)} (x_3 - x_2, x_3 - z) \la \xi_\tau , a_{x_1}^* a_{x_3} \xi_\tau \ra \\ &+N_0^{1/2}   \int_0^1 ds \int_0^{1-s} dt \int_0^{1-t} d\tau  \int dx dy dz dz' dx_1 dx_2 dx_3 \, V_N (y) \check{\sigma} (x-z) \\ &\hspace{.3cm} \times \check{A}^{(6)} (y, x-z') \check{A}^{(3)} (x_1 - x_2, x_1 - z') \check{A}^{(6)} (x_3 - x_1, x_3 - x_2) \la \xi_\tau , a_{z}^* a_{x_3} \xi_\tau \ra  \\
T_4 = \; &N_0^{1/2}  \int_0^1 ds \int_0^{1-s} dt \int_0^{1-t} d\tau  \int dx dy dz dz' dx_1 dx_2 dx_3 dx_4 \, V_N (y) \check{\sigma} (x-z) \\ & \times  \check{A}^{(6)} (y, x-z') \check{A}^{(6)} (x_1 -x_2 , x_1 - z')  \check{A}^{(3)} (x_3 - x_4 , x_3 - x_1) \la \Theta\xi_\tau, a_z^* a_{x_2}^* a_{x_3} a_{x_4} \Theta \xi_\tau \ra 
 \end{split} \end{equation*} 
 and where $\pm \cE_\Theta \leq C N^{-\nu}$ for arbitrary $\nu > 0$, if $\alpha > 0$ and $n \in \bN$ are large enough. The identity (\ref{eq:T04}) follows similarly as (\ref{eq:decoNg}), because the operator 
 \[ N_0^{1/2}  \int dx dy dz dz' \,  V_N (y) \check{\sigma} (x-z), \check{A}^{(6)} (y, x-z') [ A^*, a_{z'} ] a_z^*  \]
 has exactly the form of $[\cN_g , A^* ]$ in (\ref{eq:NgAstar}), if we choose 
 \[ g (z) = N_0^{1/2} \int dx dy \, V_N (y) \check{\sigma}  (x)  \check{A}^{(6)} (y, x-z). \]
In contrast to (\ref{eq:decoNg}), however, the contribution $T_{4,2}$ does not appear in (\ref{eq:T04})  (it corresponds to the term $S_{1,3,1}$). To bound $T_0$, we switch to Fourier space. We find
\begin{equation}\label{eq:T0F} T_0 = \frac{N_0^{1/2}}{3} \sum_{r,k,p \in \Lambda^*} \hat{V}_N (r) \sigma_k \, A^{(6)}_{r,k} A^{(3)}_{p,k} A^{(6)}_{p, k}. \end{equation} 
With Lemma \ref{lm:A3A6} (and Lemma \ref{lemma:propertiesquadratickernels}, for the bound $\| \sigma \|_4^4 \leq CN^{3\kappa/2}$), we find 
\[ |T_0| \leq C N \| \hat{V}_N \|_\infty \sum_{j=3,6} \sum_{p,k \in \Lambda^*} (1 + \sigma_k^2) |A_{p,k}^{(j)}|^2  \leq C N^{4\kappa-1}. \]
Also $T_2$ is conveniently expressed in Fourier space as 
\begin{equation}\label{eq:T2F} \begin{split} T_2 = \; &\frac{N_0^{1/2}}{2}  \int_0^1 d\tau \, (1- \tau^2) \\ &\hspace{.3cm} \times \sum_{r,k,p \in \Lambda^*} \hat{V}_N (r) \, \sigma_k  \, A_{r,k}^{(6)} \, A_{p,k}^{(3)} A_{p,k}^{(6)} \, \big\la \xi_\tau , \big( a_k^* a_k + a_p^* a_p + a_{p+k}^* a_{p+k} \big) \xi_\tau \big\ra .\end{split}  \end{equation}  
With $\| \hat{V}_N \|_\infty \leq C N^{-1+\kappa}$, $\| \sigma \|_\infty \leq C N^{\beta/8}$, Lemma \ref{lm:A3A6} (and the identity $A^{(6)}_{p,k} = A^{(6)}_{p+k, -k}$ to handle the contribution proportional to $a_{p+k}^* a_{p+k}$), we obtain 
\[ |T_2| \leq C N^{\kappa-\beta/4} \int_0^1 d\tau \, \big\la \xi_\tau , (\cN + \tilde{\cN}) \xi_\tau \rangle. \] 
With Lemma \ref{lm:boundN} and Lemma \ref{lemma:N_g}, we conclude that $|T_2| \leq C N^{4\kappa-1-\beta/4}$.  Finally, we consider $T_4$. Setting $D^{(6)} (x) = \sup_y |\check{A}^{(6)} (y,x)|$, we can bound  
\[ \begin{split} |T_4| \leq \; &C N^{1/2} \int_0^1 d\tau \int dx dy dz' dx_1 dx_2 dx_3 dx_4 \, |V_N (y)| |D^{(6)} (x-z')| |\check{A}^{(6)} (x_1 - x_2 , x_1 -z')| \\ &\hspace{3cm} \times  |\check{A}^{(3)} (x_3 - x_4, x_3 -x_1)|  \|  a_{x_3} a_{x_4} \Theta \xi_\tau \| \Big\| \int dz \check{\sigma} (x-z) a_z a_{x_2} \Theta \xi_\tau \Big\| .\end{split} \] 
 Next, we localize the kernels $D^{(6)}, \check{A}^{(6)}$ and we apply Cauchy-Schwarz. We find 
 \[  \begin{split} |T_4| \leq \; &CN^{1/2} \int_0^1 d\tau  \sum_{u_1, u_2 ,v_1, v_2,w \in \Lambda_B} \\ &\times \Big(\int dx dy dz' dx_1 dx_2 dx_3 dx_4 \, |V_N (y)| |D_w^{(6)} (x-z')|^2 |A_{u_1, v_1}^{(3)} (x_3 - x_4, x_3 -x_1)|^2 
 \\ &\hspace{1.3cm}  \times  \mathbbm{1} (|x_1 - x_2 - u_2 | , |x_1 - z' - v_2| \leq 10 \ell_B)  \Big\| \int dz \, \check{\sigma} (x-z) a_{x_2} a_z \Theta \xi_\tau \Big\|^2  \Big)^{1/2} \\ 
&\times  \Big(\int dx dy dz' dx_1 dx_2 dx_3 dx_4 \, |V_N (y)| |A_{u_2, v_2}^{(6)} (x_1 - x_2 , x_1 -z')|^2 \\ &\hspace{1.3cm} \times 
\mathbbm{1} (|x-z'-w|, |x_3 - x_4-u_1| , |x_3 - x_1-v_1| \leq 10\ell_B) \| a_{x_3} a_{x_4}\Theta  \xi_\tau \|^2   \Big)^{1/2} .
\end{split} \]
Applying Lemma \ref{lemma:propertiesL} and appropriately shifting the integration variables, we obtain 
\begin{equation}\label{eq:T4last} \begin{split} 
|T_4| \leq \; &C N^{1/2+\epsilon} \int_0^1 d\tau \sum_{u_1, u_2 ,v_1, v_2,w \in \Lambda_B} 
\\ &\times \Big(\int dx dy dz' dx_1 dx_2 dx_3 dx_4 \, |V_N (y)| |D_w^{(6)} (z')|^2 |A_{u_1, v_1}^{(3)} (x_4, x_3)|^2 
 \\ &\hspace{5cm}  \times  \mathbbm{1} ( |x_1| \leq 10 \ell_B)  \Big\| \int dz \, \check{\sigma} (x-z) a_z \xi_\tau \Big\|^2  \Big)^{1/2} \\ 
&\times  \Big(\int dx dy dz' dx_1 dx_2 dx_3 dx_4 \, |V_N (y)| |A_{u_2, v_2}^{(6)} (x_2 , z')|^2 \\ &\hspace{5cm} \times 
\mathbbm{1} (|x|, |x_1| \leq 10\ell_B) \| a_{x_3}  \xi_\tau \|^2   \Big)^{1/2} \\
 \leq \; &C N^{1/2+\epsilon} \| V_N \|_1 \ell_B^{9/2} \big(\sum_u \| D^{(6)}_u \|_2 \big) \big( \sum_{u,v} \| \check{A}^{(6)}_{u,v} \|_2 \big)  \big( \sum_{u,v} \| \check{A}^{(3)}_{u,v} \|_2 \big) \\ &\hspace{7cm} \times \int_0^1 d\tau \, \| \cN^{1/2} \xi_\tau \| \| \tilde{\cN}^{1/2} \xi_\tau \| \\ \leq \; &C N^{4\kappa-1 -\beta/8 + \epsilon} 
 \end{split} \end{equation} 
for any $\epsilon > 0$. Choosing $\epsilon > 0$ small enough, we conclude that $|T_4| \leq C N^{4\kappa-1}$ and therefore that $|S_{1,3}| \leq C N^{4\kappa-1}$. This concludes the proof of (\ref{eq_Cn_f}). 
\end{proof}

In the next lemma, we estimate the expectation of the potential energy operator $\cV_N$. 
\begin{lemma} \label{lm:VN} 
From \eqref{eq:lmHN1}, recall the operator 
\begin{equation}
\cV_N= \frac{1}{2}\int dxdy N^{2-2\kappa}V(N^{1-\kappa}(x-y))a^*_xa^*_ya_x a_y  .
\end{equation}
We have 
\begin{equation}
 \label{eq_Vn_f}
\begin{split} 
 \la e^{A^*\Theta-\hc}\Omega, &\cV_N e^{A^*\Theta-\hc}\Omega\ra \\ &= \frac{1}{2}\int dxdy  N^{2-2\kappa}V(N^{1-\kappa}x)\check{A}^{(6)}(x,y)^2 + O (N^{4\kappa-1})  .  
\end{split} \end{equation}
\end{lemma}

\begin{proof}
With $\xi_s = e^{sA^* \Theta -\hc} \Omega$, we find, for any $t \in [0;1]$, 
    \begin{equation}\label{eq:VN00} 
    \begin{split}
        &\la e^{t A^*\Theta-\hc}\Omega, \cV_N e^{t A^*\Theta-\hc}\Omega\ra 
        = 
        \int_0^t ds \la \xi_s , \big( [\cV_N, A^*] + \hc \big) \xi_s \ra + \cE_\Theta 
    \end{split}
    \end{equation}
where, as usual, $\cE_\Theta$ collects all contributions proportional to $1-\Theta$. Note that we used $[\cV_N, \Theta]=0$ as both are position space multiplication operators for each fixed number of particles. Using Lemma \ref{lemma:propertiesL} and proceeding similarly as in (\ref{eq:insTheta}), we can include this term in the error $O (N^{4\kappa-1})$, if we choose $\alpha > 0$ and $n \in \bN$ large enough. We write 
 \begin{equation}
\label{eq:comVN}
\begin{split}
    [\cV_N,A^*] +\hc =&\; \frac{1}{2}\int dxdydz \, V_N (x-y) \check{A}^{(6)}(x-y,x-z) a_x^*a_y^*a_z^* + \hc \\ &+\int dx dy dz_1 dz_2 \, V_N (x-y) \check{A}^{(3)}(z_1-z_2,z_1-x)a_x^*a_y^* a_{z_1}^* a_{z_2}^* a_y + \hc \\=& \; T_3+T_5, 
\end{split}
\end{equation}
where we used again the notation $V_N (x) = N^{2-2\kappa} V (N^{1-\kappa} x)$. To bound the expectation $\la \xi_s, T_5 \xi_s \ra$, we introduce the operator
\[ \cW_N =\frac{1}{2}\int dxdy \, |V_N(x-y)|\;a^*_{x}a^*_{y}a_{x}a_{y}. \]
and we localize the kernel $\check{A}^{(3)}$. Applying Cauchy-Schwarz and Lemma  \ref{lemma:propertiesL}, we find 
\begin{equation*}
    \begin{split}
&\abs{\la \xi_s,T_{5} \xi_s \ra} \\
&\leq \; C\sum_{u,v\in \Lambda_{\sigma}}
\Big( \int dxdydz_1dz_2 |V_N(x-y)| |\check{A}^{(3)}_{u,v}(z_1-z_2,z_1-x)|^2\|a_y \Theta\xi_s\|^2 \Big)^{\frac{1}{2}} \\
&\times \Big( \int dxdydz_1dz_2 |V_N(x-y)| \mathbbm{1}( |z_1-z_2-u|,|z_1-x+v|\leq \ell_B) \|a_{z_1}a_{z_2} a_xa_y\Theta\xi_s\|^2 \Big)^{\frac{1}{2}}\\ 
&+\cE_\Theta \\ 
\leq \; &CN^{\epsilon} \|V_N\|_1^{1/2}  
(\sum_{u,v\in \Lambda_{\sigma}}\|\check{A}^{(3)}_{u,v}\|) \| \cN^{1/2} \xi_s \| \| \cW_N^{1/2} \xi_s \| + \cE_\Theta  \end{split} 
\end{equation*} 
and thus 
\begin{equation}\label{eq:T5-fin} |\la \xi_s, T_5 \xi_s \ra | \leq C N^{4\kappa-1} + C N^{-\beta + \epsilon} \la \xi_s , \cW_N \xi_s \ra \end{equation} 
for all $s \in [0;t]$. We consider now the expectation of the cubic term $T_3$ in (\ref{eq:comVN}). We observe that this term has exactly the same form as the cubic operator $\cC_N$ considered in Lemma \ref{lemma:C_N}, with $N_0^{1/2} \check{\sigma} (x-z)$ replaced by $\check{A}^{(6)} (x-y, x-z)/2$. Thus, we can bound the expectation of $T_3$ exactly as we estimated $\cC_N$. Similarly as in (\ref{eq:CN1}), we write 
\[  \la \xi_s , T_3 \xi_s \ra =\text{Re }  \big( \wt{S}_1 + \wt{S}_3 \big) + \cE_\Theta \]
where
\[ \begin{split} \wt{S}_1 &= \int_0^s d\tau \int dx dy dz \, V_N (x-y) \check{A}^{(6)} (x-y,x-z) \la \xi_\tau , \cT_1 (x,y) e^{-(1-\tau) A^* \Theta + \hc} a_z^* \xi_s \ra \\
\wt{S}_3 &= \int_0^s d\tau  \int dx dy dz \, V_N (x-y) \check{A}^{(6)} (x-y,x-z) \la \xi_\tau, \cT_3 (x,y) e^{-(1-\tau) A^* \Theta + \hc} a_z^* \xi_s \ra. \end{split} \] 
Replacing the bound (\ref{eq:Kastar}) by the estimate $\sup_y \| \check{A}^{(6)} (y, \cdot) \|_2 \leq C N^{3\kappa/4+1/2}$ and by (\ref{eq:NaNa}), we obtain $|\wt{S}_3| \leq C N^{4\kappa-1-\beta/4+\epsilon}$. Proceeding as in the derivation of (\ref{eq:S1-4}), we write 
\[ \begin{split} 
 \wt{S}_1 = \; &s \int dx dy \, V_N (y) \, |\check{A}^{(6)} (y, x)|^2 \\ &+  s 
\int dx dy dz dz' \, V_N  (y) \check{A}^{(6)} (y, x-z) \check{A}^{(6)} (y , x-z')  \la \xi,a_z^*  a_{z'}  \xi \ra \\
 &+ \int_0^s d\tau_1 \int_0^{1-\tau_1} d\tau_2  \int dx dy dz dz' \, V_N (y) \check{A}^{(6)} (y,x-z) \check{A}^{(6)} (y , x-z')  \\ &\hspace{9cm} \times \la \xi_{1-\tau_2}, [ A^* , a_{z'} ] a_z^* \xi_{1-\tau_2} \ra \\ 
 &+ \int_0^s d\tau_1 \int_0^{1-\tau_1} d\tau_2 \int_0^{\tau_2} d\tau_3 \int dx dy dz dz' \, V_N (y) \check{A}^{(6)} (y,x-z) \check{A}^{(6)} (y , x-z')  \\ &\hspace{6cm} \times \la \xi_{1-\tau_2}, [ A^* , a_{z'} ]  e^{-\tau A^* \Theta + \hc} [ A, a_z^* ] \xi_{1-\tau_2+\tau_3} \ra \\ =: \; &\wt{S}_{1,1} + \wt{S}_{1,2} + \wt{S}_{1,3} + \wt{S}_{1,4}.
 \end{split} \] 
Analogously as in (\ref{eq:S12}), (\ref{eq:S14}), replacing the bounds \[ N^{1/2} \| \tilde{\cN}^{1/2} \xi \|  \leq C N^{3\kappa/2} , \quad N^{1/2}  \| \cN^{1/2}_{\check{\sigma} \star \check{\sigma}} \xi_{1-t-\tau} \| \leq C N^{3\kappa/2} \] with the estimates (\ref{eq:NaNa}), we obtain $|\wt{S}_{1,2}|, |\wt{S}_{1,4}| \leq CN^{4\kappa-1}$. Similarly to (\ref{eq:S130}), we can decompose $\wt{S}_{1,3} = \wt{S}_{1,3,1} + \wt{S}_{1,3,2}$ with
\[ \begin{split} \wt{S}_{1,3,1} &= \int_0^s d\tau_1 \int_0^{1-\tau_1} d\tau_2 \int_0^{1-\tau_2} d\tau_3 \\  &\hspace{.5cm} \times \int dx dy dz dz' \, V_N (y) \check{A}^{(6)} (y,x-z) \check{A}^{(6)} (y, x-z')  \la \xi_{\tau_3} , [A^*, a_{z'} ] [ A, a_z^* ]  \xi_{\tau_3} \ra \\
\wt{S}_{1,3,2} &=  \int_0^s d\tau_1 \int_0^{1-\tau_1} d\tau_2 \int_0^{1-\tau_2} d\tau_3 \\  &\hspace{.5cm} \times  \int dx dy dz dz' \, V_N (y) \check{A}^{(6)} (y,x-z) \check{A}^{(6)} (y, x-z')   \la \xi_{\tau_3} , [A, [A^*, a_{z'} ]]  a_z^*  \xi_{\tau_3} \ra. \end{split} \]
Replacing again $N^{1/2} \| \ddot{\cN}^{1/2}_{\check{\sigma} \star \check{\sigma}} \xi_\tau \| \leq C N^{3\kappa/2}$ with (\ref{eq:NaNa}), we can show $|\wt{S}_{1,3,1}| \leq C N^{4\kappa -1}$. As for the term $\wt{S}_{1,3,2}$, we can write it as $\wt{S}_{1,3,2} = \wt{T}_0 + \wt{T}_2 + \wt{T}_4$, where $\wt{T}_0, \wt{T}_2, \wt{T}_4$ are defined like the contributions $T_0, T_2, T_4$ in (\ref{eq:T04}), but with $N_0^{1/2} \check{\sigma} (x-z)$ replaced by $\check{A}^{(6)} (y, x-z)$ (and with adjusted time integrals). The bound $|\wt{T}_4| \leq C N^{4\kappa-1-\beta/8 + \epsilon}$ can be shown exactly as in (\ref{eq:T4last}) (replacing, in the last step, the estimate $N^{1/2} \| \tilde{\cN}^{1/2} \xi_\tau \| \leq C N^{3\kappa/2}$ by (\ref{eq:NaNa})). Similarly to (\ref{eq:T0F}), (\ref{eq:T2F}), we write 
\[ \begin{split} 
\wt{T}_0 = \; &\frac{s}{2} (1-s^2/3) \sum_{r,q,k,p \in \Lambda^*} \hat{V}_N (r) A_{q,k}^{(6)} A^{(6)}_{q+r, k} A^{(3)}_{p,k} A^{(6)}_{p,k} \\
\wt{T}_2 = \; &\int_0^1 d\tau \,\big[s (1-\tau) - \frac{(s-\tau)^2}{2} \chi (\tau \leq s) \big]  \\ &\hspace{1cm} \times   \sum_{r,q,k,p \in \Lambda^*} \hat{V}_N (r) A_{q,k}^{(6)} A^{(6)}_{q+r, k} A^{(3)}_{p,k} A^{(6)}_{p,k} \big\la \xi_\tau,  \big(a_p^* a_p + a_k^* a_k + a_{p+k}^* a_{p+k} \big) \xi_\tau \big\rangle  \end{split} \]
in Fourier space. With the bounds in Lemma \ref{lm:A3A6}, Lemma \ref{lm:boundN}, Lemma \ref{lemma:N_g}, we find that $|\wt{T}_0| \leq C N^{4\kappa-1}$ and $|\wt{T}_2| \leq C N^{4\kappa-1 - \beta/4 + \epsilon}$, for any $\epsilon > 0$. Choosing $\epsilon >0$ small enough, we conclude that $|\wt{S}_{1,3,2}| \leq C N^{4\kappa-1}$ and thus that 
\[  \la \xi_s, T_3 \xi_s \ra \leq s \int dx dy \, V_N (y) |\check{A}^{(6)} (y,x)|^2 + C N^{4\kappa-1} \]
for all $s \in [0;1]$. Combining this bound with (\ref{eq:VN00}), (\ref{eq:comVN}), (\ref{eq:T5-fin}), we conclude that 
\begin{equation}\label{eq:VN-fin} \begin{split} \la &e^{t A^* \Theta-\hc} \Omega, \cV_N e^{t A^* \Theta-\hc} \Omega \ra \\ &\leq \frac{t^2}{2}  \int dx dy \, V_N (y) |\check{A}^{(6)} (y,x)|^2 + C N^{-\beta + \epsilon} \int_0^t ds \, \la \xi_s , \cW_N \xi_s \ra +  C N^{4\kappa-1}. \end{split} \end{equation}  
Similarly, replacing $V_N$ with $|V_N|$ (and observing that all estimates only depend on $\| V_N \|_1$, $\| \hat{V}_N \|_2$), we obtain  
\begin{equation}\label{eq:WN-fin} \begin{split} \la &e^{t A^* \Theta-\hc} \Omega, \cW_N e^{t A^* \Theta-\hc} \Omega \ra \\ &\leq \frac{t^2}{2}   \int dx dy \, |V_N (y)| |\check{A}^{(6)} (y,x)|^2 + C N^{-\beta + \epsilon} \int_0^t ds \, \la \xi_s , \cW_N \xi_s \ra +  C N^{4\kappa-1} .\end{split} \end{equation}  
With $\| \widehat{|V|} \|_\infty \leq C \| V \|_1 \leq C$, we find   
\[ \Big| \int dx dy |V_N (y)| |\check{A}^{(6)} (y,x)|^2 \Big| = \Big| N^{-1+\kappa} \sum_{r,p,q \in \Lambda^*} \widehat{|V|} (r/N^{1-\kappa}) \, A^{(6)}_{p,q} A^{(6)}_{p+r,q} \Big| \leq C N^{5\kappa/2} .\]
Choosing $\epsilon > 0$ small enough, (\ref{eq:WN-fin}) implies that  
\[ \la \xi_t , \cW_N \xi_t \ra \leq C N^{5\kappa/2} + C \int_0^t ds \, \la \xi_s , \cW_N \xi_s \ra  \] for all $t \in [0;1]$. By Gronwall's Lemma, we conclude that $\la \xi_t , \cW_N \xi_t \ra \leq C N^{5\kappa/2}$ 
for all $t \in [0;1]$. Inserting this bound in (\ref{eq:VN-fin}) and observing that $5\kappa/2 - \beta = 4\kappa-1 - \beta/2$, we obtain (\ref{eq_Vn_f}).
\end{proof}

Finally, we evaluate the expectation of the kinetic energy operator in our trial state $\xi = e^{A^* \Theta - \hc} \Omega$. 
\begin{lemma} \label{lm:Kxi}
We have 
\[ \begin{split}  \la e^{A^* \Theta - \hc} &\Omega, \cK e^{A^* \Theta - \hc} \Omega \ra \\ &= - N^{1/2} \int dx dy \, N^{2-2\kappa} (Vf) (N^{1-\kappa} x) \check{\sigma} (y) \check{A}^{(6)} (x,y) + O (N^{4\kappa-1}). \end{split} \]
\end{lemma} 

\begin{proof} 
We use the notation $\xi_t = e^{t A^* \Theta - \hc} \Omega$. Since $\la \Omega, \cK \Omega \ra = 0$, we have 
\begin{equation}\label{eq:Kexp} \begin{split} 
\la \xi_t, \cK \xi_t \ra = \; &2 \text{Re } \int_0^t ds \, \la \xi_s , [ \cK, A^* \Theta ] \xi_t \ra \\ = \; &2\text{Re } \int_0^t ds \, \la \xi_s, [\cK, A^*] \Theta \xi_s \ra + 2\text{Re } \int_0^t ds \, \la \xi_s , A^* [ \cK, \Theta ] \xi_s \ra. \end{split} \end{equation} 
Expressing the cubic operator $A^*$ in momentum space, as 
\[ A^* = \sum_{p,q \in \Lambda^*} A_{p,q} a_{p+q}^* a_{-p}^* a_{-q}^* \, , \]
and recalling the definition (\ref{eq:cubicfunctionmomentumspace}) of the coefficients $A_{p,q}$, we find
\begin{equation}\label{eq:commKA} \begin{split} [ \cK, A^*] &= 2N^{-1/2} \sum_{p,q \in \Lambda^*} p^2 \eta_p \sigma_q a_{p+q}^* a_{-p}^* a_{-q}^* \\ &= - N^{\kappa-1/2} \sum_{p,q \in \Lambda^*} \widehat{Vf} (p/N^{1-\kappa}) \tilde{\chi} (p) \sigma_q \, a_{p+q}^* a_{-p}^* a_{-q}^* .\end{split} \end{equation}  
With Lemma \ref{lemma:propertiesL}, it follows that 
\begin{equation} \label{eq:commKA2} 2 \text{Re } \int_0^t ds \, \la \xi_s, [ \cK, A^*] \Theta \xi_s \ra = 2 \text{Re } \int_0^t ds \, \la \xi_s, [\cK, A^* ] \xi_s \ra + O (N^{4\kappa-1}) \end{equation} 
if we choose the parameter $\alpha >0$ in the definition of the cutoff $\Theta$ large enough. Moreover, from (\ref{eq:commKA}), we obtain $\| (\cN+1)^{-3/2} [ \cK, A^* ] \| \leq C N^{1 + \kappa/4}$.  From Lemma \ref{lemma:propertiesL}, this implies the rough but useful a-priori bound
\begin{equation}\label{eq:rough} \Big| \int_0^t ds \la \xi_s, [\cK, A^* ] \Theta \xi_s \ra \Big| \leq C N^{5\kappa /2 + 1} .
\end{equation} 
Let us now consider the second term on the r.h.s. of (\ref{eq:Kexp}). We observe that the cutoff $\Theta$ preserves the number of particles and, on the $m$-particle sector of the Fock space $\cF = \oplus_{m=0}^\infty L^2_s (\Lambda ^m)$, it acts as multiplication operator with the function 
\[ F_m (x_1, \dots , x_m) = \Theta \Big( \int_{\Lambda } dw \, \big( \sum_{j=1}^m \chi_{\ell_B} (x_j - w) + 1 \big)^n / N^\alpha \Big) .\]
Since also the kinetic energy operator $\cK$ leaves the number of particles invariant, on the $m$-particle sector, the action of the commutator $[\cK, \Theta]$ is given by 
\begin{equation} \label{eq:commFm}  \begin{split} \sum_{j=1}^m \big[ &-\Delta_{x_j} , F_m (x_1, \dots , x_m) \big] \\ &= - 2 \sum_{j=1}^m (\nabla_{x_j} F_m) (x_1, \dots , x_m) \cdot \nabla_{x_j} - \sum_{j=1}^m (\Delta_{x_j} F_m) (x_1, \dots , x_m) .\end{split} \end{equation} 
We compute 
\[ \begin{split} \nabla_{x_j} F_m (x_1, \dots , x_m) = \; &n N^{-\alpha}  \Theta'  \Big( \int_{\Lambda } dw \, \big( \sum_{j=1}^m \chi_{\ell_B} (x_j - w) + 1 \big)^n / N^\alpha \Big) \\ &\times \int_{\Lambda } dw \, \big( \sum_{j=1}^m \chi_{\ell_B} (x_j - w) + 1 \big)^{n-1} \nabla \chi_{\ell_B} (x_j - w) .\end{split} \] 
The fact that $\Theta'$ is supported in the interval $[1;2]$ implies, on the one hand, that 
\[ [\cK, \Theta ] = \mathbbm{1}_{[1;\infty)} (\cL_n / N^\alpha ) [ \cK, \Theta ] \]
and, on the other hand, the bound 
\[ \begin{split} 
|\nabla_{x_j} F_m (x_1, \dots , x_m)|^2 \leq \; &n^2 N^{-2\alpha} \Big| \Theta'  \Big( \int_{\Lambda } dw \, \big( \sum_{j=1}^m \chi_{\ell_B} (x_j - w) + 1 \big)^n / N^\alpha \Big) \Big|^2 \\ & \times \Big| \int_{\Lambda } dw \big( \sum_{j=1}^m \chi_{\ell_B} (x_j - w) + 1 \big)^{n-1} |\nabla \chi_{\ell_B} (x_j - w)| \Big|^2 \\ \leq \; &C \| \nabla \chi_{\ell_B} \|_\infty^2 \leq C  \ell_B^{-2}  \end{split} \]
for a constant $C > 0$ (depending on the parameter $n \in \bN$). Similarly, we find 
\[ |\Delta_{x_j} F_m (x_1, \dots , x_m)| \leq C \ell_B^{-2}. \]
From (\ref{eq:commFm}), we conclude that 
\[ \| \frac{1}{(\cN+1)} [\cK , \Theta ] \xi_s \| \leq C \ell_B^{-2} \| (\cK+1)^{1/2} \xi_s \| \] 
for every $s \in [0;1]$ and therefore that 
\[ \Big|  \int_0^t ds \, \la \xi_s , A^* [ \cK, \Theta ] \xi_s \ra \Big| \leq C \ell_B^{-2} \int_0^t ds \, \| (\cN+1) \mathbbm{1}_{[1;\infty)} (\cL_n / N^\alpha) A \xi_s \| \| (\cK + 1)^{1/2} \xi_s \|. \]
 With Lemma \ref{lemma:propertiesL} we obtain that, for any $\nu > 0$,  
 \begin{equation} \label{eq:commTheta} \Big|  \int_0^t ds \, \la \xi_s , A^* [ \cK, \Theta ] \xi_s \ra \Big| \leq C N^{-\nu} \int_0^t ds \| (\cK+1)^{1/2} \xi_s \| \end{equation} 
 for every $t \in [0;1]$, if $\alpha > 0$ is large enough. Inserting the last bound together with  (\ref{eq:rough}) in (\ref{eq:Kexp}), we arrive at 
 \[ \la \xi_t, \cK \xi_t \ra \leq C N^{5\kappa/2+1} + C N^{-\nu} \int_0^t ds \la \xi_s , \cK \xi_s \ra \]
 for every $t \in [0;1]$. By Gronwall's Lemma, we obtain $\la \xi_t, \cK \xi_t \ra \leq C N^{5\kappa/2+1}$ for all $t \in [0;1]$. Plugging this back into (\ref{eq:commTheta}), we conclude from (\ref{eq:Kexp}) and (\ref{eq:commKA2}) that
 \begin{equation}\label{eq:latK}  \la \xi_t, \cK \xi_t \ra = 2 \text{Re } \int_0^t ds \la \xi_s, [\cK, A^* ] \xi_s \ra + O (N^{4\kappa-1}) \end{equation} 
 if the parameter $\alpha > 0$ is large enough. Translating (\ref{eq:commKA}) to position space, we find 
 \begin{equation}\label{eq:commKA22}  \begin{split} [\cK, A^* ] = \; &- \sqrt{N} \int dx dy dz \, N^{2-2\kappa} (Vf) (N^{1-\kappa} (x-y)) \check{\sigma} (x-z) a_x^* a_y^* a_z^* \\ &+ \int dx dy dz  \, \check{A}_{\cK}^\text{low} (x-y,x-z) \, a_x^* a_y^* a_z^* \end{split} \end{equation} 
with $A_{\cK}^\text{low}$ as defined in (\ref{eq:kineticlowmomenta}). To bound the second term, with the kernel $A_{\cK}^\text{low}$, we apply Lemma \ref{lm:prelim1}; using the bounds $\| A_{\cK}^\text{low} \|_2 \leq C N^{(5\kappa-1)/2}$ and 
\[ \sum_{u,v \in \Lambda_B} \| \check{A}^\text{low}_{\cK,u,v} \|_2 \leq C N^{(5\kappa-1)/2 + \epsilon} \]
from Lemma \ref{lm:Klow}, we conclude that this contribution is negligible, of the order $N^{4\kappa-1}$. As for the first term on the r.h.s. of (\ref{eq:commKA22}), we compare it with  (\ref{eq:CN-pos}) and we observe that it has exactly the same form as the cubic operator $\cC_N$, just with the prefactor $N_0^{1/2}$ replacing by $N^{1/2}$ and with the potential $V$ replaced by $Vf$. Hence, proceeding exactly as in Lemma \ref{lemma:C_N}, we conclude that 
\[ 2 \text{Re } \la \xi_s, [\cK, A^* ] \xi_s \ra = -2s N^{1/2} \int dx dy \, N^{2-2\kappa} (Vf) (N^{1-\kappa} x) \check{\sigma} (y) \check{A}^{(6)} (x,y) + O (N^{4\kappa-1}) \]
 and therefore, from (\ref{eq:latK}), that
 \[ \begin{split} \la e^{A^* \Theta-\hc} \Omega , &\cK e^{A^* \Theta-\hc} \Omega \ra \\ &= -N^{1/2} \int dx dy \, N^{2-2\kappa} (Vf) (N^{1-\kappa} x) \check{\sigma} (y) \check{A}^{(6)} (x,y) + O (N^{4\kappa-1}) \, . \end{split} \]
\end{proof}

\subsection{Proof of Theorem \ref{theorem:Fockspaceresult}} 
\label{subsec:proof}

We combine the results of Lemma \ref{lemma:C_N}, Lemma \ref{lm:VN} and Lemma \ref{lm:Kxi}.
\begin{lemma} \label{lm:cubic-exp} 
Let $\cC_N$ be the cubic operator defined in (\ref{eq:CtCtV}), $\cV_N$ the quartic potential energy operator and $\cK$ the quadratic kinetic energy operator. Then, we have 
\[ \begin{split} \la \xi, (\cK + &\cC_N + \cC_N^* + \cV_N) \xi \ra \\ = \; &N^{\kappa-1} \sum_{p,q \in \Lambda^*} \hat{V} (p/N^{1-\kappa}) (\eta_{\infty,p} + \eta_{\infty,p+q}) \sigma_q^2  
\\ &- N^{4\kappa-1}\sum_{|p|, |q|\geq N^{\kappa/2}} (\widehat{Vf}(p/N^{1-\kappa})+\widehat{Vf}((p+q)/N^{1-\kappa})) \\ &\hspace{5cm} \times 
\frac{ \widehat{Vf}(p/N^{1-\kappa})\, \widehat{Vf}(q/N^{1-\kappa})\, \widehat{Vf}((p+q)/N^{1-\kappa})}{4p^2q^2(p^2+q^2+(p+q)^2)}\\
        &+N^{4\kappa-1} \sum_{|p|, |q| \geq N^{\kappa/2}} (\widehat{Vf}(p/N^{1-\kappa}) + \widehat{Vf}((p+q)/N^{1-\kappa})) 
\\ &\hspace{5cm} \times  \frac{(p\cdot q) \, \widehat{Vf}(p/N^{1-\kappa}) \, \widehat{Vf}(q/N^{1-\kappa})^2}{4 q^4 p^2(p^2+q^2+(p+q)^2)}
 +O(N^{4\kappa-1}) .
        \end{split}\]
\end{lemma} 

\begin{proof} 
From Lemma \ref{lemma:C_N} and Lemma \ref{lm:Kxi}, we find, writing all contributions in momentum space 
\begin{equation}\label{eq:KC} \begin{split} \langle \xi, (\cK + \cC_N + \cC_N^*) \xi \ra &= \; N^{\kappa-1/2} \sum_{p,q \in \Lambda^*} \big( 2 \widehat{Vw}  (p/N^{1-\kappa}) + \widehat{Vf} (p/N^{1-\kappa}) \big) \sigma_q \\ &\times (A_{p,q}+A_{q,p} + A_{-p-q,q}+A_{q,-p-q} + A_{p,-p-q}+A_{-p-q,p})\\
&+O(N^{4\kappa-1}).
\end{split} \end{equation} 
Note that we have replaced $N_0$ by $N$ in the contribution from $\cC_N$ using that $\abs{N_0^{1/2}-N^{1/2}}\leq N^{3\kappa/2-1}$ and that the whole term is of order $N^{5\kappa/2}$, see Lemma \ref{lemma:propertiesquadratickernels} and \ref{lemma:propertiesA}.
As for the expectation of $\cV_N$, we find, from Lemma \ref{lm:VN}, 
\[ \la \xi, \cV_N \xi \ra = \frac{N^{\kappa-1}}{2} \sum_{p,q,r \in \Lambda^*} \hat{V} (r/N^{1-\kappa}) A^{(6)}_{p,q} A^{(6)}_{p-r,q} + O(N^{4\kappa-1}) .\]
Since $A^{(6)}_{p,q} = A^{(6)}_{-p-q,q}$, we have  
\[ \begin{split} \frac{N^{\kappa-1}}{2} &\sum_{p,q,r \in \Lambda^*} \hat{V} (r/N^{1-\kappa}) A^{(6)}_{p,q} A^{(6)}_{p-r,q} \\
= \; &N^{\kappa-1} \sum_{p,q,r \in \Lambda^*} \hat{V} (r/N^{1-\kappa})(A_{p-r,q} + A_{q,p-r} + A_{-p-q+r, p-r})  A^{(6)}_{p,q}  \\ 
= \; &N^{\kappa-1} \sum_{p,q,r \in \Lambda^*} \hat{V} (r/N^{1-\kappa})(A_{p-r,q} + A_{q,p-r})  \\ &\hspace{2cm} \times (A_{p,q}+A_{q,p} + A_{-p-q,q}+A_{q,-p-q} + A_{p,-p-q}+A_{-p-q,p}) \\ &+N^{\kappa-1} \sum_{p,q,r \in \Lambda^*}  \hat{V} (r/N^{1-\kappa}) A_{p,q} \\ &\hspace{.2cm} \times (A_{q+r,-p-q} + A_{-p-q, q+r} + A_{p-r, -p-q} + A_{-p-q, p-r} +  A_{p-r, q+r}  + A_{q+r, p-r}) .
\end{split} \]
The last term corresponds to the contribution proportional to $A_{-p-q+r,p-r}$ in the first line, after switching to new variables $-p-q +r \to p, p-r \to q$. Rewriting 
\[\begin{split}  \sum_{p,q,r} \hat{V} (r/N^{1-\kappa}) &A_{p.q} (A_{q+r,-p-q} + A_{-p-q, q+r} + A_{p-r,-p-q} + A_{-p-q,p-r}) \\ &= \sum_{p,q,r} \hat{V} (r/N^{1-\kappa}) (A_{p,q} + A_{q,p}) (A_{p-r,-p-q} + A_{-p-q,p-r}) \\ &= \sum_{p,q,r} \hat{V} (r/N^{1-\kappa}) (A_{p-r,q} + A_{q,p-r})  (A_{p,-p-q} + A_{-p-q,p}) \end{split}   \]
we conclude that 
\[ \begin{split}  \la \xi, \cV_N \xi \ra = \; &N^{\kappa-1} \sum_{p,q,r \in \Lambda^*} \hat{V} (r/N^{1-\kappa})(A_{p-r,q} + A_{q,p-r})  \\ &\hspace{2cm} \times (A_{p,q}+A_{q,p} + A_{-p-q,q}+A_{q,-p-q} + 2A_{p,-p-q}+2A_{-p-q,p}) \\ &+N^{\kappa-1} \sum_{p,q,r \in \Lambda^*}  \hat{V} (r/N^{1-\kappa}) A_{p,q} (A_{p-r, q+r}  + A_{q+r, p-r}) 
+ O(N^{4\kappa-1}).
\end{split} \]
Next, we observe that 
\begin{equation} \label{eq:A+A} \begin{split} 
A_{p-r,q} &+ A_{q,p-r} = N^{-1/2}  \eta_{p-r} \sigma_q - N^{-1/2} \frac{2(p-r) \cdot q}{(p-r)^2 + q^2 + (p-r+q)^2} \sigma_q \eta_{p-r} \\ + \; &N^{-1/2} \frac{2q^2 \eta_q (\sigma_{p-r} - \eta_{p-r})}{(p-r)^2 + q^2 + (p-r+q)^2} + N^{-1/2}  \frac{2 q^2 \eta_{p-r} (\eta_q -\sigma_q)}{(p-r)^2 + q^2 + (p-r+q)^2}. \\
&  \end{split} \end{equation}
Hence, we can decompose 
\[ \la \xi , \cV_N \xi \ra = \sum_{j=1}^5 C_{\cV,j} + O(N^{4\kappa-1}) \]
where
\begin{equation}\label{eq:CV1-5} \begin{split} 
C_{\cV,1} &= N^{\kappa-3/2} \sum \hat{V} (r/N^{1-\kappa}) \eta_{p-r} \sigma_q  \\ &\hspace{3cm} \times (A_{p,q}+A_{q,p} + A_{-p-q,q}+A_{q,-p-q} + 2A_{p,-p-q}+2A_{-p-q,p})  \\
C_{\cV,2} &= -N^{\kappa-3/2} \sum \hat{V} (r/N^{1-\kappa}) \frac{2 (p-r) \cdot q}{(p-r)^2 + q^2 + (p-r+q)^2} \eta_{p-r} \sigma_q \\ &\hspace{3cm} \times (A_{p,q}+A_{q,p} + A_{-p-q,q}+A_{q,-p-q} + 2A_{p,-p-q}+2A_{-p-q,p})  \\
C_{\cV,3} &= N^{\kappa-3/2} \sum \frac{2q^2 (\eta_q - \sigma_q) \eta_{p-r}}{(p-r)^2 + q^2 + (p-r+q)^2} \\ &\hspace{3cm} \times (A_{p,q}+A_{q,p} + A_{-p-q,q}+A_{q,-p-q} + 2A_{p,-p-q}+2A_{-p-q,p})  \\
C_{\cV,4} &= N^{\kappa-3/2} \sum \frac{2q^2 \eta_q (\sigma_{p-r} - \eta_{p-r})}{(p-r)^2 + q^2 + (p-r+q)^2} \\ &\hspace{3cm} \times (A_{p,q}+A_{q,p} + A_{-p-q,q}+A_{q,-p-q} + 2A_{p,-p-q}+2A_{-p-q,p}) \\
C_{\cV,5} &= N^{\kappa-1} \sum \hat{V} (r/N^{1-\kappa}) A_{p,q} (A_{p-r,q+r} + A_{q+r,p-r}) .
\end{split} \end{equation} 
The last three terms are negligible. In fact, with the bounds in Lemma \ref{lemma:propertiesA}, we find
\begin{equation*}
\begin{split}
\abs{C_{\cV,5}} 
&\leq C N^{5\kappa-2} \sum_{\substack{p,q,r\in\Lambda^*\colon\\ p,q, p+r, q-r\neq 0}} \abs{\hat{V}(r/N^{1-\kappa})}\abs{p}^{-2}\abs{p+r}^{-2}\abs{q}^{-2}\abs{q-r}^{-2}\\
&\leq C N^{5\kappa-2} \sum_{r\in\Lambda^*_+} \abs{\hat{V}(r/N^{1-\kappa})}\abs{r}^{-2}
\leq C N^{4\kappa-1}.
\end{split}
\end{equation*}
Moreover, 
\begin{equation*}
\begin{split}
\abs{C_{\cV,4}} 
&\leq C N^{\kappa-2} \norm{\eta-\sigma}_1 (\norm{\eta}_2^2+\norm{\sigma}_2^2)(\norm{\eta}_1+\norm{\sigma}_1) \leq C N^{4\kappa-1} 
\end{split}
\end{equation*}
and 
\begin{equation*}
\begin{split}
\abs{C_{\cV, 3}}
&\leq C N^{\kappa-2} \norm{\eta/\abs{\cdot}^2}  \norm{\abs{\cdot}^2(\eta-\sigma)}_2(\norm{\eta}_2+\norm{\sigma}_2)(\norm{\eta}_1+\norm{\sigma}_1)
\leq C N^{4\kappa-1}.
\end{split}
\end{equation*}
Next, we consider the term $C_{\cV,1}$. From Lemma \ref{lemma:scattering}, we have 
\[ \begin{split} \Big| \sum_{r \in \Lambda^*} N^{\kappa-1} \hat{V} (r/N^{1-\kappa}) \eta_{\infty,p-r} + N^\kappa \widehat{Vw} (p/N^{1-\kappa}) \Big| &\leq CN^{2\kappa-1} \\  
\Big| \sum_{r \in \Lambda^*} N^{\kappa-1} \hat{V} (r/N^{1-\kappa}) (\eta_{p-r} - \eta_{\infty,p-r}) \Big| &\leq CN^{5\kappa/2-1}. \end{split} \]  
This implies that 
\[ \begin{split} C_{\cV,1} = \; &-N^{\kappa-1/2} \sum_{p,q \in \Lambda^*} \widehat{Vw} \, (p/N^{1-\kappa}) \sigma_q \\ &\times (A_{p,q}+A_{q,p} + A_{-p-q,q}+A_{q,-p-q} + 2A_{p,-p-q}+2A_{-p-q,p}) + O(N^{4\kappa-1}). \end{split} \]
With (\ref{eq:KC}), we obtain 
\begin{equation}\label{eq:C1+C2} \begin{split} 
\la \xi, (\cK + \cC_N &+ \cC_N^*) \xi \ra + C_{\cV,1} \\ =\; &N^{\kappa-1/2}\sum_{p,q\in\Lambda^*} \hat{V}(p/N^{1-\kappa}) \sigma_q (A_{p,q}+A_{q,p} + A_{-p-q,q}+A_{q,-p-q})\\
    &+N^{\kappa-1/2}\sum_{p,q\in\Lambda^*} \widehat{Vf}(p/N^{1-\kappa})) \sigma_q (A_{p,-p-q}+A_{-p-q,p})
    +O(N^{4\kappa-1}) \\ =\; & C_1 + C_2 + O(N^{4\kappa-1}).
\end{split} \end{equation}
Expanding $A_{p,q} + A_{q,p}$ and $A_{-p-q,q} + A_{q,-p-q}$ similarly to (\ref{eq:A+A}) and observing that the contribution of all terms containing the difference $\sigma - \eta$ are negligible (similarly as $C_{\cV,3}, C_{\cV_4}$ above), we conclude that 
\[ \begin{split} C_1 = \; & - N^{\kappa-1} \sum_{p,q \in \Lambda^*} \hat{V}(p/N^{1-\kappa}) \sigma_q^2 \Big( \frac{2\eta_p \, p \cdot q }{p^2 + q^2 + (p+q)^2} - \frac{2\eta_{p+q} (p+q) \cdot q}{p^2 + q^2 + (p+q)^2}  \Big) \\ & + 
N^{\kappa-1}\sum_{p,q\in\Lambda^*} \hat{V}(p/N^{1-\kappa}) \sigma^2_q (\eta_p + \eta_{p+q})  + O (N^{4\kappa-1}) .\end{split} \]
Switching $p+q \to -p$ in the second term on the first line and using that $\| \eta - \eta_\infty \|_1 \leq C N^{3\kappa/2}$ in the second line, we conclude that 
\begin{equation}\label{eq:C1-fin} \begin{split} C_1 = \; & - N^{\kappa-1} \sum_{p,q \in \Lambda^*} \big( \hat{V}(p/N^{1-\kappa}) + \hat{V} ((p+q)/N^{1-\kappa}) \big)  \frac{2\eta_p  \sigma_q^2 \, p \cdot q }{p^2 + q^2 + (p+q)^2}  \\ & + 
N^{\kappa-1}\sum_{p,q\in\Lambda^*} \hat{V}(p/N^{1-\kappa}) \sigma^2_q (\eta_{\infty,p} + \eta_{\infty, p+q})  + O (N^{4\kappa-1}) .\end{split} \end{equation} 
As for $C_2$, we expand $A_{p,-p-q} + A_{-p-q,p}$ similarly to (\ref{eq:A+A}) and we observe that contributions proportional to $\eta-\sigma$ are negligible  (similarly as $C_{\cV,3}, C_{\cV_4}$ above) and that the remaining terms can be combined into 
\[ \begin{split} -  &N^{-1/2} \frac{2p \cdot (-p-q)}{p^2 + q^2 + (p+q)^2} \eta_p \sigma_{p+q} + N^{-1/2}  \eta_p \sigma_{p+q} 
\\ &= 2 N^{-1/2} \frac{p^2 + (p+q)^2}{p^2 + q^2 + (p+q)^2} \eta_p \eta_{p+q} + 2 N^{-1/2} \frac{p^2 + (p+q)^2}{p^2 + q^2 + (p+q)^2} \eta_p (\sigma_{p+q} -\eta_{p+q}). \end{split}    \] 
Since also the contribution of the last term, proportional to $\sigma_{p+q} - \eta_{p+q}$, is negligible, we arrive at
\[ \begin{split} C_2 &= N^{\kappa-1} \sum_{p,q \in \Lambda^*} \widehat{Vf} (p/N^{1-\kappa}) \frac{2p^2 + 2(p+q)^2}{p^2 + q^2+ (p+q)^2} \sigma_q \eta_p \eta_{p+q} + O(N^{4\kappa-1}) \\ &= N^{\kappa-1}\sum_{p,q\in\Lambda^*} (\widehat{Vf}(p/N^{1-\kappa})+\widehat{Vf}((p+q)/N^{1-\kappa}))
\frac{2(p+q)^2 \eta_{p+q} \eta_p \sigma_q}{p^2+q^2+(p+q)^2} + O(N^{4\kappa-1}) .
\end{split} \]
Again, we can replace $\sigma_q$ by $\eta_q$, with a negligible error. Also, using the definition (\ref{eq:defeta}) for replacing $2(p+q)^2 \eta_{p+q}$ by $-N^\kappa\widehat{Vf} ((p+q)/N^{1-\kappa})$ and then removing the momentum restriction on $p+q$ produces contributions of order $N^{4\kappa-1}$. Similarly, we can replace the smooth cutoff contained in $\eta$ by a sharp one at $N^{\kappa/2}$ noting that the difference is supported on momenta between $N^{\kappa/2}$ and $2N^{\kappa/2}$.
We conclude that 
\begin{equation}\label{eq:C2-fin} \begin{split} C_2 &= - N^{4\kappa-1}\sum_{\substack{p,q\in\Lambda^*\colon\\|p|,|q| \geq N^{\kappa/2}}} (\widehat{Vf}(p/N^{1-\kappa})+\widehat{Vf}((p+q)/N^{1-\kappa}))
\\ 
&\hspace{2cm} \times \frac{\widehat{Vf} ((p+q)/N^{1-\kappa}) \widehat{Vf} (p/N^{1-\kappa}) \widehat{Vf} (q/N^{1-\kappa})}{4p^2 q^2 \big(p^2+q^2+(p+q)^2\big)} 
+ O(N^{4\kappa-1}) .
\end{split} \end{equation} 
Next, we consider the term $C_{\cV,2}$, from (\ref{eq:CV1-5}). We start by observing that 
\[ \begin{split} \Big| N^{\kappa-3/2} &\sum_{p,q,r} \hat{V} (r/N^{1-\kappa}) \frac{2 (p-r) \cdot q}{(p-r)^2 + q^2 + (p-r+q)^2} \eta_{p-r} \sigma_q (A_{p, -p-q} + A_{-p-q,p}) \Big| \\ &\leq  C N^{\kappa-3/2} \sum_{p,q,r} \frac{|r||q|}{r^2 + q^2 + (r+q)^2} |\eta_{r}| |\sigma_q| (|A_{p,-p-q}| + |A_{-p-q,p}|) \\ &\leq  C N^{4\kappa-2} \sum_{p,q,r} \frac{|r| |\eta_r|}{q^2 (q+r)^2} \leq C N^{4\kappa-1}  \end{split} \]
where we applied (\ref{eq:Appq}), from Lemma \ref{lemma:propertiesA}. Thus, we have 
\[ \begin{split} C_{\cV,2} = \; & -N^{\kappa-3/2} \sum \hat{V} (r/N^{1-\kappa}) \frac{2 (p-r) \cdot q}{(p-r)^2 + q^2 + (p-r+q)^2} \eta_{p-r} \sigma_q \\ &\hspace{3cm} \times (A_{p,q}+A_{q,p} + A_{-p-q,q}+A_{q,-p-q}) + O(N^{4\kappa-1}) \\
= \; & -N^{\kappa-3/2} \sum \hat{V} ((p-r)/N^{1-\kappa}) \frac{2 p \cdot q}{p^2 + q^2 + (p+q)^2} \eta_{p} \sigma_q \\ &\hspace{3cm} \times (A_{r,q}+A_{q,r} + A_{-r-q,q}+A_{q,-r-q}) + O(N^{4\kappa-1}) 
\end{split} \]
after switching $p-r \to p, r \to p$. We expand $A_{r,q} + A_{q,r}$ and $A_{-r-q,q} + A_{q,-r-q}$ similarly as in (\ref{eq:A+A}) and, as usual, we observe that contributions containing the difference $\sigma-\eta$ are negligible. Thus, we arrive at 
\[ \begin{split} C_{\cV,2} 
= \; & -N^{\kappa-2} \sum (\hat{V} ((p-r)/N^{1-\kappa}) + \hat{V} ((p+q-r)/N^{1-\kappa}) ) \frac{2 p \cdot q}{p^2 + q^2 + (p+q)^2} \eta_{p} \eta_{r} \sigma^2_q \\ &+ N^{\kappa-2} \sum \big( \hat{V} ((p-r)/N^{1-\kappa}) + \hat{V} ((p+q-r)/N^{1-\kappa}) \big) \\ & \hspace{2.5cm} \times  \frac{2 p \cdot q}{p^2 + q^2 + (p+q)^2} \frac{2r \cdot q}{r^2 + q^2 + (r+q)^2}  \eta_{p} \eta_r \sigma^2_q + O(N^{4\kappa-1}) .
\end{split} \]
Let us consider the second term, for example the contribution proportional to $\hat{V} ((p-r)/N^{1-\kappa})$ (the contribution proportional to $\hat{V} ((p+q-r)/N^{1-\kappa})$ can be treated analogously). With the change of variable $r \to -r$, we can write  
\[ \begin{split} N^{\kappa-2} &\sum \hat{V} ((p-r)/N^{1-\kappa})    \frac{2 p \cdot q}{p^2 + q^2 + (p+q)^2} \frac{2r \cdot q}{r^2 + q^2 + (r+q)^2}  \eta_{p} \eta_r \sigma^2_q \\ = \; &\frac{N^{\kappa-2}}{2} \sum \big( \hat{V} ((p-r)/N^{1-\kappa}) - \hat{V} ((p+r)/N^{1-\kappa}) \big)  \\ &\hspace{3cm} \times  \frac{2 p \cdot q}{p^2 + q^2 + (p+q)^2} \frac{2r \cdot q}{r^2 + q^2 + (r+q)^2}  \eta_{p} \eta_r \sigma^2_q \\ &-N^{\kappa-2} \sum \hat{V} ((p+r)/N^{1-\kappa})  \frac{2 p \cdot q}{p^2 + q^2 + (p+q)^2}  \\ &\hspace{3cm} \times \frac{(2 r \cdot q)^2}{\big(r^2 + q^2 + (r+q)^2\big) \big( r^2 + q^2 + (r-q)^2 \big)}
 \eta_{p} \eta_r \sigma^2_q .\end{split} \] 
In the first term, we bound $|\hat{V}((p-r)/N^{1-\kappa})-\hat{V}((p+r)/N^{1-\kappa})|\leq C |r| / N^{1-\kappa}$ and, using $\| \eta \|_1 \leq CN$, 
\[ \begin{split} N^{2\kappa-3} \sum_{p,q,r} &\frac{|p| |r|^2 |q|^2 |\eta_p| |\eta_r| |\sigma_q|^2}{(p^2 + q^2 + (p+q)^2) (r^2 + q^2 + (r+q)^2)} \\ &\hspace{.4cm} \leq CN^{2\kappa-2} \sum_{p,q} \frac{|p| |q|^2 |\eta_p| |\sigma_q|^2}{(p^2 + q^2 + (p+q)^2)} \leq CN^{4\kappa-2} \sum_{p,q} \frac{|p| |\eta_p|}{(p+q)^2 q^2} \leq C N^{4\kappa-1}. \end{split} \]
In the second term we proceed similarly, using $|r|^2 |\eta_r| \leq C N^{\kappa}$ and $\sum_r |r|^{-2} |r+q|^2 \leq C |q|^{-1}$. 
We conclude that 
\[ N^{\kappa-2} \sum \hat{V} ((p-r)/N^{1-\kappa})    \frac{2 p \cdot q}{p^2 + q^2 + (p+q)^2} \frac{2r \cdot q}{r^2 + q^2 + (r+q)^2}  \eta_{p} \eta_r \sigma^2_q = O (N^{4\kappa-1}) \]
and thus that 
\[ \begin{split} C_{\cV,2} 
= \; & N^{\kappa-1} \sum (\widehat{Vw} (p/N^{1-\kappa}) + \widehat{Vw} ((p+q)/N^{1-\kappa}) ) \frac{2\eta_{p} \sigma^2_q  \, p \cdot q}{p^2 + q^2 + (p+q)^2} + O(N^{4\kappa-1}) .
\end{split} \]
Combined with (\ref{eq:C1+C2}), (\ref{eq:C1-fin}), (\ref{eq:C2-fin}), we obtain
\begin{equation}\label{eq:lastC} \begin{split} \la \xi, (\cK &+ \cC_N +\cC_N^* + \cV_N) \xi \ra \\ 
= \; &N^{\kappa-1}\sum_{p,q\in\Lambda^*} \hat{V}(p/N^{1-\kappa})  (\eta_{\infty,p} + \eta_{\infty, p+q}) \sigma_q^2 \\
&- N^{4\kappa-1}\sum_{\substack{p,q\in\Lambda^*\colon\\|p|,|q| > N^{\kappa/2}}} (\widehat{Vf}(p/N^{1-\kappa})+\widehat{Vf}((p+q)/N^{1-\kappa}))
\\ &\hspace{2cm} \times \frac{\widehat{Vf} ((p+q)/N^{1-\kappa}) \widehat{Vf} (p/N^{1-\kappa}) \widehat{Vf} (q/N^{1-\kappa})}{4p^2 q^2 \big(p^2+q^2+(p+q)^2\big)} \\
&- N^{\kappa-1} \sum_{p,q \in \Lambda^*} \big( \widehat{Vf}(p/N^{1-\kappa}) + \widehat{Vf} ((p+q)/N^{1-\kappa}) \big)  \frac{2\eta_p  \sigma_q^2 \, p \cdot q }{p^2 + q^2 + (p+q)^2}  + O (N^{4\kappa-1}) .
\end{split} \end{equation}
We can replace $\eta_p$ by $\eta_{\infty,p}\mathbbm{1}(\abs{\cdot}\geq N^{\kappa/2})$, because 
\[ N^{\kappa-1} \sum_{|p| < 2N^{\kappa/2}, q\in \Lambda^*} \frac{|p| |q| |\eta_{\infty,p}| |\sigma_q|^2}{p^2 + q^2 + (p+q)^2}
\leq N^{5\kappa/2 - 1} \sum_{p,q \in \Lambda^*} \frac{|q| |\sigma_q|^2}{p^2 (p+q)^2}
\leq C N^{4\kappa-1}. \] 
Similarly, we can also replace $\sigma_q^2$ by $\eta_{\infty,q}^2\mathbbm{1}(\abs{\cdot}\geq N^{\kappa/2})$. This follows from a symmetrization using a change of variable $q \to -q$, because
\[ 
\begin{split} 
N^{\kappa-1} &\sum_{p,q\in\Lambda^*\colon \abs{p}\geq N^{\kappa/2}} \frac{|p\cdot q|^2 |\eta_{\infty,p}| \abs{\sigma_q^2-\eta_{\infty,q}^2\mathbbm{1}(\abs{\cdot}\geq N^{\kappa/2})}}{(p^2 + q^2 + (p+q)^2)(p^2 + q^2 + (p-q)^2)} \\ 
&\leq N^{2\kappa-1} \sum_{p,q\in\Lambda^*\colon \abs{p}\geq N^{\kappa/2}} \frac{|q|^2 \abs{\sigma_q^2-\eta_{\infty,q}^2\mathbbm{1}(\abs{\cdot}\geq N^{\kappa/2})}}{p^4}\\
&\leq N^{5\kappa/2 -1} \sum_{q\in\Lambda^*} \abs{\sigma_q-\eta_{\infty,q}\mathbbm{1}(\abs{\cdot}\geq N^{\kappa/2})} 
\leq C N^{4\kappa-1} 
\end{split} 
\]
and
\[
\begin{split}
    &N^{2\kappa-2} \sum_{\substack{p,q\in\Lambda^*\colon \abs{p}\geq N^{\kappa/2}}} \frac{|p| |q|^2 |\eta_{\infty,p}| \abs{\sigma_q^2-\eta_{\infty,q}^2\mathbbm{1}(\abs{\cdot}\geq N^{\kappa/2})}}{p^2 + q^2 + (p+q)^2}\\
    &\leq N^{3\kappa-2} \sum_{\substack{p,q\in\Lambda^*\colon \abs{p}\geq N^{\kappa/2}}} \frac{ |\eta_{\infty,p}| \abs{\sigma_q-\eta_{\infty,q}\mathbbm{1}(\abs{\cdot}\geq N^{\kappa/2})}}{\abs{p}}
    \leq C \log(N) N^{4\kappa-1-\beta/2}.
\end{split}
\]
We can therefore rewrite the last term on the r.h.s. of (\ref{eq:lastC}) as   
\[ \begin{split} 
- & N^{\kappa-1} \sum_{p,q \in \Lambda^*} \big( \widehat{Vf}(p/N^{1-\kappa}) + \widehat{Vf} ((p+q)/N^{1-\kappa}) \big)  \frac{2\eta_p  \sigma_q^2 \, p \cdot q }{p^2 + q^2 + (p+q)^2} \\
&= N^{4\kappa-1} \sum_{|p|,|q| \geq N^{\kappa/2}} \big( \widehat{Vf}(p/N^{1-\kappa}) + \widehat{Vf} ((p+q)/N^{1-\kappa}) \big) \\ &\hspace{4cm}  \times \frac{(p \cdot q) \, \widehat{Vf} (p/N^{1-\kappa}) \widehat{Vf} (q/N^{1-\kappa})^2}{4p^2 q^4 (p^2 + q^2 + (p+q)^2)} + O(N^{4\kappa-1}) \end{split} \]
concluding the proof of the lemma. 
\end{proof} 

We are now ready to show Theorem  \ref{theorem:Fockspaceresult}.
\begin{proof}[Proof of Theorem  \ref{theorem:Fockspaceresult}]
Combining Lemma \ref{lemma:quadraticrenormalization} with Lemma \ref{lm:boundN}, Lemma \ref{lm:tildeCN}, Lemma \ref{lemma:N_g}, Lemma \ref{lm:tildeVN} and Lemma \ref{lm:cubic-exp}, we obtain, for any $\kappa \geq 1/2$, 
\begin{equation} \label{eq:S1S2} \begin{split} 
\la \psi_N, &\cH_N \psi_N \ra \\ = \; &4\pi\aa N^{1+\kappa}\\
        &+ \frac{1}{2}\sum_{p\in\Lambda^*_+}
        \Big(
        \sqrt{\abs{p}^4+2p^2 N^\kappa\widehat{Vf}(\frac{p}{N^{1-\kappa}})}
        - p^2 - N^\kappa\widehat{Vf}(\frac{p}{N^{1-\kappa}})
        + \frac{N^{2\kappa}\widehat{Vf}(\frac{p}{N^{1-\kappa}})^2}{2\abs{p}^2}
        \Big) 
\\ &- N^{4\kappa-1}\sum_{|p|, |q|\geq N^{\kappa/2}} (\widehat{Vf}(p/N^{1-\kappa})+\widehat{Vf}((p+q)/N^{1-\kappa})) \\ &\hspace{5cm} \times 
\frac{ \widehat{Vf}(p/N^{1-\kappa})\, \widehat{Vf}(q/N^{1-\kappa})\, \widehat{Vf}((p+q)/N^{1-\kappa})}{4p^2q^2(p^2+q^2+(p+q)^2)}\\
        &+N^{4\kappa-1} \sum_{|p|, |q| \geq N^{\kappa/2}} (\widehat{Vf}(p/N^{1-\kappa}) + \widehat{Vf}((p+q)/N^{1-\kappa})) 
\\ &\hspace{5cm} \times  \frac{(p\cdot q) \, \widehat{Vf}(p/N^{1-\kappa}) \, \widehat{Vf}(q/N^{1-\kappa})^2}{4 q^4 p^2(p^2+q^2+(p+q)^2)}
 +O(N^{4\kappa-1}) \\
 =: \; &4\pi\aa N^{1+\kappa} + S_1 + S_2 + S_3+ O (N^{4\kappa-1}) .
 \end{split} \end{equation} 
Let us consider the term $S_1$ first. Since the summand decays as $N^{3\kappa} |p|^{-4}$, as $|p| \to \infty$, we can restrict the sum to $|p| \leq N^{1-\kappa}$, up to a negligible error. We define 
\[ g_p (t) = 
        \sqrt{p^4+2p^2 N^\kappa\widehat{Vf}(\frac{tp}{N^{1-\kappa}})}
        - p^2 - N^\kappa\widehat{Vf}(\frac{tp}{N^{1-\kappa}})
        + \frac{N^{2\kappa}\widehat{Vf}(\frac{tp}{N^{1-\kappa}})^2}{2\abs{p}^2}. \]
 Then 
 \[ \begin{split} S_1 =& \; \frac{1}{2}\sum_{\substack{p\in\Lambda^*_+\\ \abs{p}\leq N^{1-\kappa}}}
         \left(
        \sqrt{|p|^4+16\pi \mathfrak{a} N^\kappa p^2}
        - p^2 -  8\pi \mathfrak{a}N^\kappa 
        + \frac{N^{2\kappa}(8\pi \mathfrak{a})^2}{2\abs{p}^2}
        \right)\\
        &+ \frac{1}{2}\int_0^1 ds\int_0^s dt \sum_{\substack{p\in\Lambda^*_+\\ \abs{p}\leq N^{1-\kappa}}}
        \frac{d^2}{dt^2} g_p(t)
        +O(N^{4\kappa-1}) \end{split} \]
where we used $\widehat{Vf} (0) = 8\pi \aa$ and $g'_p (0) = 0$, because $\int x V(x) f(x) dx = 0$ by symmetry. Explicit computation of the derivatives of $g_p$ leads to the estimate 
 \begin{equation*}
            \abs{g_p''(t)} 
            \leq 
          \left\{ \begin{array}{ll}  
               C N^{4\kappa-2} &\textrm{ if } \abs{p} \leq K N^{\kappa/2}\\
               C N^{4\kappa-2} \frac{N^{\kappa}}{\abs{p}^2} &\textrm{ if } \abs{p} > K N^{\kappa/2}
            \end{array} \right. 
    \end{equation*}
if $K > 0$ is large enough. Hence, 
\[         \begin{split}
        S_1 
        &= \frac{1}{2}\sum_{\substack{p\in\Lambda^*_+}}
        \left(
        \sqrt{p^4+16\pi \mathfrak{a} N^\kappa p^2}
        - p^2 -  8\pi \mathfrak{a}N^\kappa
        + \frac{N^{2\kappa}(8\pi \mathfrak{a})^2}{2\abs{p}^2}
        \right)
        +O(N^{4\kappa-1}) .\end{split} \]
Defining 
\[          h(p) = 
        \sqrt{p^4+16\pi \mathfrak{a} N^\kappa p^2}
        - p^2 -  8\pi \mathfrak{a}N^\kappa
        + \frac{N^{2\kappa}(8\pi \mathfrak{a})^2}{2\abs{p}^2}
    \]
we can write 
\[ \begin{split} S_1 =\; & \frac{1}{2(2\pi)^3}\int dp^3
        \left(
        \sqrt{p^4+16\pi \mathfrak{a} N^\kappa p^2}
        - p^2 -  8\pi \mathfrak{a}N^\kappa
        + \frac{N^{2\kappa}(8\pi \mathfrak{a})^2}{2\abs{p}^2}
        \right)\\
        &+ \frac{1}{2(2\pi)^3}\sum_{\substack{p\in\Lambda^*_+}}
        \int_{[-\pi,\pi]^3} d\zeta \, (h(p)-h(p+\zeta))
        +O(N^{2\kappa})+O(N^{4\kappa-1})\\
        \end{split}
    \] 
where the error $O (N^{2\kappa})$ arises from the integral around $p=0$. By explicit computation, we find 
\[         \begin{split}
            \Big| \frac{d^2}{ds^2} h(p+s\zeta) \Big| 
               \leq
               C \frac{N^{2\kappa}}{\abs{p+s\zeta}^4} .
        \end{split}
    \] 
Thus, 
 \begin{equation*}
        \begin{split}
        \sum_{\substack{p\in\Lambda^*_+}}
        \int_{[-\pi,\pi]^3} d\zeta \, (h(p)-h(p+\zeta))
        &= - \sum_{\substack{p\in\Lambda^*_+}}
        \int_{[-\pi,\pi]^3} d\zeta \int_0^1 ds \, \frac{d}{ds} h(p+s\zeta)\\
        &= - \sum_{\substack{p\in\Lambda^*_+}}
        \int_{[-\pi,\pi]^3} d\zeta \int_0^1 ds \int_0^s \, \frac{d^2}{dt^2} h(p+t\zeta)   
        =O(N^{2\kappa})
        \end{split}
    \end{equation*}
because $(d/ds) h (p+s\zeta)$, evaluated at $s=0$, is an odd function of $p$ (and therefore its sum over $p$ vanishes). We conclude that, for $\kappa \geq 1/2$, 
  \begin{equation*}
        \begin{split}
            S_1 
            &=
            \frac{1}{2(2\pi)^3}\int dp^3
        \left(
        \sqrt{p^4+16\pi \mathfrak{a} N^\kappa p^2}
        - p^2 -  8\pi \mathfrak{a}N^\kappa
        + \frac{N^{2\kappa}(8\pi \mathfrak{a})^2}{2\abs{p}^2}
        \right)
        +O(N^{4\kappa-1})\\
                    &=
            \frac{(8\pi\mathfrak{a}N^\kappa)^{5/2}}{2(2\pi)^3}\int dp^3
        \left(
        \sqrt{p^4+2  p^2}
        - p^2
        -  1
        + \frac{1}{2\abs{p}^2}
        \right)
        +O(N^{4\kappa-1})\\
        &= \frac{512\sqrt{\pi}}{15}\mathfrak{a}^{5/2}N^{5\kappa/2} 
        +O(N^{4\kappa-1}).
        \end{split}
    \end{equation*}
Next, we consider the term $S_2$, as defined in (\ref{eq:S1S2}). Since 
\[ \begin{split}  \sum_{|p| > N^{1-\kappa} , |q| \leq N^{1-\kappa}} \frac{1}{|p|^4 |q|^2} &\leq C , \\ \sum_{|p|,|q| > N^{1-\kappa}} |\widehat{Vf} (p/N^{1-\kappa})| \frac{1}{p^2 |q|^4} &\leq C N^{-1+\kappa} \| \widehat{Vf} \|_2 \| \chi (|.| > N^{1-\kappa}) / |.|^2 \|_2 \leq C \end{split} \]
we can restrict the sum to $|p|, |q| \leq N^{1-\kappa}$. Estimating $|\widehat{Vf} (p/N^{1-\kappa}) - 8\pi \aa| \leq C p^2 / N^{2-2\kappa}$, we obtain 
\[ S_2 = - 2048 N^{4\kappa-1} \pi^4 \aa^4 \sum_{\substack{p,q\in\Lambda^*\\ N^{\kappa/2}\leq \abs{p},\abs{q}\leq N^{1-\kappa}}} 
\frac{1}{p^2q^2(p^2+q^2+(p+q)^2)}
        +O(N^{4\kappa-1}). \]
As for the term $S_3$, we can proceed as we did for $S_2$ to restrict the sum to $|q| \leq N^{1-\kappa}$. In order to restrict it to $|p| \leq  N^{1-\kappa}$, on the other hand, we need additionally a symmetrization argument (replacing $q$ with $-q$). This argument is very similar to \cite[Eq. (2.39)]{COSS}; we skip the details. After restricting to $|p|,|q| \leq N^{1-\kappa}$, we can use $|\widehat{Vf} (p/N^{1-\kappa}) - 8\pi \aa| \leq C p^2 / N^{2-2\kappa}$ as we did for $S_2$, to conclude that 
\[ S_3 = 2048 N^{4\kappa-1} \pi^4 \aa^4  \sum_{\substack{p,q\in\Lambda^*\\N^{\kappa/2}\leq\abs{p}, \abs{q}\leq N^{1-\kappa}}}
 \frac{p\cdot q}{q^4 p^2(p^2+q^2+(p+q)^2)}
 +O(N^{4\kappa-1}).
\] 
We arrive at
\[ S_2 + S_3 = 2048 \,N^{4\kappa-1} \pi^4 \aa^4 \sum_{\substack{p,q\in\Lambda^*\\N^{\kappa/2}\leq\abs{p}, \abs{q}\leq N^{1-\kappa}}}
 \frac{p\cdot q - q^2}{q^4 p^2(p^2+q^2+(p+q)^2)} +O(N^{4\kappa-1}).
\] 
Proceeding exactly as in \cite[Section 2]{COSS}, we can replace the sum by an integral, with only a negligible error. We find 
\[ \begin{split} S_2 + S_3 = \; & \frac{16 \aa^4}{\pi^2}  \, N^{4\kappa-1} 
            \int_{\substack{N^{\kappa/2}\leq \abs{p},\abs{q}\leq N^{1-\kappa}}} dpdq
\frac{p\cdot q -q^2}{p^2q^4(p^2+q^2+p\cdot q)}
        \\ &+ O(N^{2\kappa-\beta/2}) + O(N^{4\kappa-1}) .
 \end{split} \]
Rescaling the integration variables, $p,q \to p/N^{\kappa/2}, q/N^{\kappa/2}$, we obtain 
\[ \begin{split} S_2 +S_3 =   \frac{16 \aa^4}{\pi^2}  \, N^{4\kappa-1} 
            \int_{\substack{1 \leq \abs{p},\abs{q}\leq N^{\beta/2}}} dpdq
\frac{p\cdot q -q^2}{p^2q^4(p^2+q^2+p\cdot q)}
       + O(N^{4\kappa-1}) .
 \end{split} \]
 for $\kappa \geq 1/2$. Up to an error of order $N^{4\kappa-1}$, we can remove the restrictions on $p$. By explicit computation, we conclude that 
\[\begin{split}  S_2 + S_3 &=  \frac{16 \aa^4}{\pi^2}  \, N^{4\kappa-1} 
            \int_{\substack{1 \leq \abs{q}\leq N^{\beta/2}}} dpdq 
\frac{p\cdot q -q^2}{p^2q^4(p^2+q^2+p\cdot q)} +O(N^{4\kappa-1})  \\ &= -32\pi \left(\frac{4\pi}{3}-\sqrt{3}\right) \mathfrak{a}^4 N^{4\kappa-1} \log(N^{\beta})
 +O(N^{4\kappa-1} ) \, . \end{split}  \]
\end{proof}

\end{document}